\documentclass[a4paper,twocolumn,10pt,accepted=2025-10-24]{quantumarticle}
\pdfoutput=1
\usepackage[utf8]{inputenc}
\usepackage[english]{babel}
\usepackage[T1]{fontenc}
\usepackage{amsmath}
\usepackage{hyperref}
\usepackage[numbers]{natbib}

\newif\ifarxiv % Comment out the next line for journal-style version
\arxivtrue

\usepackage{amssymb,amsthm,xparse,xcolor,tikz,bm}
\usepgfmodule{decorations}
\usetikzlibrary{decorations,decorations.pathmorphing,patterns}

\usepackage[nodayofweek]{datetime}  
\usepackage[normalem]{ulem}
\usepackage{xcolor}
\usepackage{blindtext}
\usepackage[capitalise]{cleveref}
\usepackage{enumerate}   
\definecolor{mylinkcolor}{rgb}{0,0,0.8} % set link color here as red,green,blue.
\hypersetup{unicode=true, %
  bookmarksnumbered=false,bookmarksopen=false,bookmarksopenlevel=1, %
  breaklinks=true,pdfborder={0 0 0},colorlinks=true}%
\hypersetup{%
  anchorcolor=mylinkcolor,citecolor=mylinkcolor, %
  filecolor=mylinkcolor,linkcolor=mylinkcolor, %
  menucolor=mylinkcolor,runcolor=mylinkcolor, %
  urlcolor=mylinkcolor}%

\newtheorem{lemma}{Lemma}
\newtheorem{corollary}{Corollary}
\newtheorem{conjecture}{Conjecture}
\newtheorem{proposition}{Proposition} 
\theoremstyle{definition} %\newtheorems from here on come without italic
\newtheorem{definition}{Definition}

\newcommand{\ket}[1]{| #1 \rangle}
\newcommand{\bra}[1]{\langle #1 |}

\newcommand{\ketbra}[2]{|#1\rangle\!\langle#2|}

\newcommand{\tr}{\mathrm{tr}}

\begin{document}

\title{Expanding bipartite Bell inequalities for maximum multi-partite randomness}

\author{Lewis Wooltorton}
    \email{lewis.wooltorton@ens-lyon.fr}
    \affiliation{Department of Mathematics, University of York, Heslington, York, YO10 5DD, United Kingdom}
    \affiliation{Quantum Engineering Centre for Doctoral Training, H. H. Wills Physics Laboratory and Department of Electrical \& Electronic Engineering, University of Bristol, Bristol BS8 1FD, United Kingdom}
    \affiliation{Inria, ENS de Lyon, LIP, 46 Allee d’Italie, 69364 Lyon Cedex 07, France}
\orcid{0000-0002-7950-3844}
\author{Peter Brown}
    \email{peter.brown@telecom-paris.fr}
    \affiliation{Télécom Paris, LTCI, Institut Polytechnique de Paris,
19 Place Marguerite Perey, 91120 Palaiseau, France}
\orcid{0000-0001-9593-0136}
\author{Roger Colbeck}
    \email{roger.colbeck@kcl.ac.uk}
    \affiliation{Department of Mathematics, King's College London, Strand, London, WC2R 2LS, United Kingdom}
    \affiliation{Department of Mathematics, University of York, Heslington, York, YO10 5DD, United Kingdom}
\orcid{0000-0003-3591-0576}

\date{$1^{\text{st}}$ December 2025}

\begin{abstract}
    Nonlocal tests on multi-partite quantum correlations form the basis of protocols that certify randomness in a device-independent (DI) way. Such correlations admit a rich structure, making the task of choosing an appropriate test difficult. For example, extremal Bell inequalities are tight witnesses of nonlocality, but achieving their maximum violation places constraints on the underlying quantum system, which can reduce the rate of randomness generation. As a result there is often a trade-off between maximum randomness and the amount of violation of a given Bell inequality. Here, we explore this trade-off for more than two parties. More precisely, we study the maximum amount of randomness that can be certified by correlations with a particular violation of the Mermin-Ardehali-Belinskii-Klyshko (MABK) inequality. For any even number of parties, we find that maximum randomness cannot occur beyond a threshold quantum violation, which increases with the number of parties, and we give a conjectured form of the maximum randomness in terms of the MABK value. We also show that maximum randomness can be obtained for any MABK violation for odd numbers of parties. To obtain our results, we derive new families of Bell inequalities certifying maximum randomness from a technique for randomness certification, which we call ``expanding Bell inequalities''. Our technique allows a bipartite Bell expression to be used as a seed, and transformed into a multi-partite Bell inequality tailored for randomness certification, showing how intuition learned in the bipartite case can find use in more complex scenarios.           
\end{abstract}
\maketitle
\twocolumngrid

\section{Introduction}
Distant measurements on a shared quantum system can display correlations inaccessible to any locally realistic model \cite{Bell_book,Brunner_review}. Such correlations are termed nonlocal and provide a device-independent (DI) witness of useful quantum characteristics, opening up a new paradigm of information processing~\cite{BarrettNonlocalResource}. Tasks such as randomness expansion~\cite{ColbeckThesis,PAMBMMOHLMM,CK2,MS1,MS2}, amplification~\cite{CR_free} and key distribution~\cite{Ekert,BHK,ABGMPS,PABGMS,VV2,ADFRV} can all be achieved in the DI regime, where security is proven based on the observed correlations and without trusting the inner workings of the devices. Moreover, nonlocal correlations can be used to prove the presence of a particular state or sets of measurements, known as self-testing~\cite{mayers2004self,McKagueSinglet,YangSelfTest,Coladangelo2017,Supic_2023,SupicSelfTest}.   

Certification of DI randomness typically uses a Bell inequality, and different Bell inequalities give different amounts of certified randomness. It is natural to consider extremal Bell inequalities, which form the minimal set separating local and non-local correlations. The set of correlations compatible with high violation is tightly constrained and can negatively affect the certification of randomness. For example, in the simplest bipartite scenario, correlations with higher violation of the only extremal Bell inequality, the Clauser-Horne-Shimony-Holt (CHSH) inequality, cannot certify maximal DI randomness, while those with a lower violation can~\cite{AcinRandomnessNonlocality,WBC}. Further complications arise when scenarios with more inputs, outputs or parties are considered because the number of different classes of extremal Bell inequalities grows and it becomes more difficult to navigate the trade-offs between extremal Bell violation and maximum certifiable randomness. Alongside foundational interest, this question is also motivated practically, in finding more robust device-independent randomness expansion (DIRE) protocols~\cite{WBC}. 

In this work we study the family of multi-partite Bell inequalities due to Mermin, Ardehali, Belinskii, and Klyshko (MABK)~\cite{Mermin,Ardehali92,Belinskii_1993} and their application to DIRE. These all have two inputs and two outputs per party. The one and two sided randomness given a tripartite MABK violation was bounded in Refs.~\cite{WoodheadMermin,GrasselliMulti21}, and compared to other tripartite inequalities in Ref.~\cite{Grasselli22}. Considering global randomness, Ref.~\cite{Bamps_2018} gave a bipartite construction for certification of $M$ bits of randomness using $M$ copies of almost unentangled states. 

Ref.~\cite{DharaMaxRand} showed how, for an odd number of parties $N$, maximum MABK violation can certify $N$ bits of DI randomness, whilst Ref.~\cite{delaTorreMaxNonlocal} provided an alternative family of Bell inequalities certifying $N$ bits in the even case. This leaves open the following questions: how suitable is MABK for maximum randomness generation when $N$ is even? Does the MABK violation trade-off with the maximum certifiable randomness, as was shown for the CHSH inequality~\cite{WBC}? Similarly, are there correlations with a non-maximal MABK violation that achieve maximum randomness when $N$ is odd? 

To analyse multi-partite nonlocality, we require tools to construct multi-partite Bell tests. To do so, one approach is to generalize intuition from the well understood CHSH scenario; this has been done to extend the number of measurements per party~\cite{Braunstein90,Supic_2016}, or the number of parties~\cite{Curchod_2019}, where it was shown how to construct multi-partite Bell inequalities from a bipartite seed. Self-testing results using projections onto bipartite subsystems~\cite{Supic_2018,Zwerger19,Li18} suggest stronger capabilities of the technique presented in Ref.~\cite{Curchod_2019}, prompting the question: can it be leveraged for multi-partite randomness certification?

Our work seeks to answer these questions. We first show that maximum MABK violation certifies $N + 1/2 - \log_{2}(1+\sqrt{2})/\sqrt{2} \approx N - 0.4$ bits of global randomness when $N$ is even, contrasting the $N$ bits certified when $N$ is odd~\cite{DharaMaxRand}. We then investigate the limits on certifiable randomness from quantum behaviours that violate an MABK inequality. To do so, we construct new families of $N$-partite Bell inequalities whose maximum violation certifies $N$ bits for arbitrary $N$. When $N$ is even, we show that for any MABK violation $m$ up to a non-maximal value $m^{*}$, there exists a quantum behaviour achieving $m$ which maximally violates an inequality from our family, and hence certifies maximum global randomness. For violations above $m^{*}$, we find (using a second family of inequalities) that randomness monotonically decreases with MABK violation. Our results lower bound this trade-off, and we provide numerical evidence of tightness, indicating an extension of the known tight bound when $N=2$~\cite{WBC} to arbitrary $N$. Additionally, we show that this trade-off diminishes (in the sense that the gap between $m^*$ and maximal MABK violation decreases) as the number of parties grows. Finally, for odd $N$, we use our family of Bell inequalities to show that $N$ bits of randomness can be certified for a range of MABK violations.

Our constructions are formed from a type of concatenation of bipartite Bell inequalities, allowing us to lift our intuition derived from the bipartite case~\cite{WBC} to the multipartite setting. Specifically, we extend the technique presented in Ref.~\cite{Curchod_2019}, originally introduced to witness genuinely multi-partite nonlocality. Our extension exploits self-testing properties of the seed to certify randomness. Specifically, our main technical contribution is a decoupling lemma (see \cref{lem:BIent}), which shows that under certain conditions the maximum violation of an expanded Bell expression implies a decoupling between the honest parties and Eve. Such decoupling occurs for extremal quantum correlations, and guarantees security~\cite{Franz11}. We envisage that this technique will be a useful tool in multi-partite DI cryptography and of independent interest. 

The paper is structured as follows. In \cref{sec:back} we provide the necessary background and notation. In \cref{sec:CB} we detail our enhancement of the technique from \cite{Curchod_2019}, along with the decoupling lemma. In \cref{sec:ex1} we describe our constructions, and show that they certify $N$ bits of randomness for a broad range of MABK violations, which we conjecture to be optimal; this result is summarized in \cref{lem:RvM1_even,lem:RvM1_odd}. We additionally show how the conjectured highest value of MABK at which maximum randomness remains possible tends to the maximum quantum value as the number of parties increases (see \cref{thm:asymp}). In \cref{sec:ex2} we consider MABK values beyond the conjectured threshold for maximum randomness, up to the maximum quantum value. Here we give a candidate form for the exact trade-off, which is summarized in \cref{lem:RvM2}. We conclude with a discussion in \cref{sec:dis}. All technical proofs can be found in the appendices, and we provide an accompanying python script for the numerical aspects at the GitHub repository~\cite{code}.     
 
\section{Background}
\label{sec:back}
\subsection{Multiparty DI-scenario}
We consider an $N$ party, 2-input 2-output DI scenario, where $N$ isolated devices are given a random input, labelled $x_{k} \in \{0,1\}$, and produce an output $a_{k} \in \{0,1\}$ %\footnote{We label the measurement outcomes $0,1$ instead of $\pm1$ for notational convenience later on.}
stored in a classical register $R_{k}$, where $k \in \{1,...,N\}$ indexes the party. We use bold to denote tuples; for example, $\bm{a} = (a_{1},...,a_{N})$ denotes an $N$ bit string of device outputs. The behaviour of the devices is then described by the joint conditional probability distribution $p(\bm{a}|\bm{x})$, which must be no-signalling because the devices are isolated.

A behaviour $p(\bm{a}|\bm{x})$ is quantum if there exists (following Naimark's dilation theorem \cite{paulsen_2003}) a pure state and sets of orthonormal projectors that can reproduce the distribution via the Born rule. Specifically, we consider an adversary Eve, who wishes to guess the device outputs, and let $\ket{\Psi}_{\tilde{\bm{Q}}E}$ denote the global state in the Hilbert space $\mathcal{H}_{\tilde{\bm{Q}}} \otimes \mathcal{H}_{E}$, where $\mathcal{H}_{E}$ is held by Eve, $\mathcal{H}_{\tilde{\bm{Q}}} = \bigotimes_{k=1}^{N}\mathcal{H}_{\tilde{Q}_{k}}$ and $\mathcal{H}_{\tilde{Q}_{k}}$ is held by device $k$. Throughout the text, tildes will denote elements pertaining to physical Hilbert spaces $\mathcal{H}_{\tilde{Q}_{k}}$ (whose dimension is unknown), whilst no tilde will describe qubit Hilbert spaces $\mathcal{H}_{Q_{k}}$. We let $\{\tilde{P}_{a_{k}|x_{k}}^{(k)}\}_{a_{k}}$ be binary-outcome projective measurements on $\mathcal{H}_{\tilde{Q}_{k}}$, and we denote the corresponding observables of each party by $\tilde{A}_{x_{k}}^{(k)} = \tilde{P}_{0|x_{k}}^{(k)} - \tilde{P}_{1|x_{k}}^{(k)}$; here bracketed superscripts will typically keep track of the party to which the object belongs (i.e., the Hilbert space on which it acts). Following measurements $\bm{x}$ by the $N$ honest parties, we obtain the classical quantum state $\rho_{\bm{R}E|\bm{x}} = \sum_{\bm{a}}\ketbra{\bm{a}}{\bm{a}}_{\bm{R}} \otimes \rho_{E}^{\bm{a}|\bm{x}}$, where $\rho_{E}^{\bm{a}|\bm{x}} = \mathrm{Tr}_{\tilde{\bm{Q}}}\big[(\tilde{P}_{\bm{a}|\bm{x}} \otimes \mathbb{I}_{E})\ketbra{\Psi}{\Psi}_{\tilde{\bm{Q}}E}\big]$ and $\tilde{P}_{\bm{a}|\bm{x}} = \bigotimes_{k=1}^{N}\tilde{P}_{a_{k}|x_{k}}^{(k)}$. The behaviour (as seen by the $N$ honest parties) is recovered by $p(\bm{a}|\bm{x}) = \mathrm{Tr}\big[\rho_{E}^{\bm{a}|\bm{x}}\big]$.

\subsection{Multiparty Nonlocality}
Given an observed distribution $p(\bm{a}|\bm{x})$, it will be useful to quantify its distance from the local boundary. When $N=2$, the CHSH inequality, as the only non-trivial facet, is a natural choice. However, the multiparty scenario is more complex, and the number of classes of Bell inequality increases rapidly with $N$~\cite{Werner01}.  Instead, a general way to quantify this distance for an arbitrary scenario can be used, which is related to how much the local polytope needs to be diluted to encompass a given non-local correlation. Computing this can be done using linear programming (see \cref{app:LP}). 

In this work, we choose to study one such CHSH generalization and its relationship to DI randomness; the MABK family \cite{Mermin,Ardehali92,Belinskii_1993}, 
\begin{equation}
     % \langle M_{N} \rangle  = \frac{1}{2}\Bigg(\frac{1-i}{2}\Bigg)^{N-1} \Big \langle \prod_{k=1}^{N} (\tilde{A}_{0}^{(k)} + i \tilde{A}_{1}^{(k)}) \Big \rangle   \\ + \frac{1}{2}\Bigg(\frac{1+i}{2}\Bigg)^{N-1} \Big \langle \prod_{k=1}^{N} (\tilde{A}_{0}^{(k)} - i \tilde{A}_{1}^{(k)}) \Big \rangle \\
     \langle M_{N} \rangle = 2^{\frac{1-N}{2}} \sum_{\bm{x}}\cos \Big[ \frac{\pi}{2}\Big(\frac{N-1}{2} - \sum_{k=1}^{N}x_{k}\Big)\Big]  \langle \tilde{A}_{\bm{x}}  \rangle , \label{eq:MABK}
\end{equation}
% This is the form if we want to switch to sigma_X and sigma_Y as the observables: $2^{\frac{1-N}{2}} \sum_{\bm{x}}\cos \Big[ \frac{\pi}{2}\Big(1 - \sum_{k=1}^{N}x_{k}\Big)\Big]  \langle \tilde{A}_{\bm{x}}  \rangle$
where $\langle \tilde{A}_{\bm{x}} \rangle = \sum_{\bm{a}}(-1)^{\sum_{k=1}^{N}a_{k}}p(\bm{a}|\bm{x}) = \bra{\Psi}\tilde{A}_{x_{1}}^{(1)} \otimes \cdots \otimes \tilde{A}_{x_{N}}^{(N)} \otimes \mathbb{I}_{E} \ket{\Psi}$ when the behaviour is quantum. [In effect the prefactor $(-1)^{\sum_{k=1}^{N}a_{k}}$ corresponds to relabelling the outcomes $0\mapsto1$ and $1\mapsto-1$ to match the usual formulation of observables with eigenvalues $\pm1$.] The local bound is given by $\langle M_{N} \rangle \leq 1$, and the maximum quantum value is $2^{(N-1)/2}$. Note that the MABK functional is characterized by the fact that the coefficients $c_{\bm{x}} = \cos \Big[ \frac{\pi}{2}\Big(\frac{N-1}{2} - \sum_{k=1}^{N}x_{k}\Big)\Big]$ are equal for all strings $\bm{x}$ with the same Hamming weight (that is, the number of $1$s in the string $\bm{x}$). When $N$ is even, $c_{\bm{x}} \in \{-1/\sqrt{2},1/\sqrt{2}\}$, and when $N$ is odd $c_{\bm{x}} \in \{-1,0,+1\}$. This is illustrated by writing out the first few cases. To simplify the notation, we let $A_{x}^{(1)} = A_{x}, \ A_{y}^{(2)} = B_{y}, \ A_{z}^{(3)} = C_{z}$ and $A_{t}^{(4)} = D_{t}$. 
\begin{widetext}
\begin{equation}
    \begin{aligned}
        \langle M_{2} \rangle &= \frac{1}{2} \Big[ \langle A_{0} B_{0} \rangle +  \langle A_{0} B_{1} \rangle + \langle A_{1} B_{0} \rangle - \langle A_{1} B_{1} \rangle \Big], \\
        \langle M_{3} \rangle &= \frac{1}{2} \Big[ \langle A_{0} B_{0} C_{1} \rangle + \langle A_{0} B_{1} C_{0} \rangle + \langle A_{1} B_{0} C_{0} \rangle - \langle A_{1}B_{1}C_{1} \rangle \Big], \\
        \langle M_{4} \rangle &= \frac{1}{4}\Big[  -\langle A_{0}B_{0}C_{0}D_{0} \rangle + \langle A_{0}B_{0}C_{0}D_{1} \rangle + \langle A_{0}B_{0}C_{1}D_{0} \rangle + \langle A_{0}B_{1}C_{0}D_{0} \rangle + \langle A_{1}B_{0}C_{0}D_{0} \rangle + \langle A_{0}B_{0}C_{1}D_{1} \rangle\\
        &\hspace{0.9cm}  + \langle A_{1}B_{1}C_{0}D_{0} \rangle + \langle A_{0}B_{1}C_{0}D_{1} \rangle + \langle A_{1}B_{0}C_{1}D_{0} \rangle + \langle A_{0}B_{1}C_{1}D_{0} \rangle + \langle A_{1}B_{0}C_{0}D_{1} \rangle  \\ 
        & \hspace{0.9cm}- \big( \langle A_{1}B_{1}C_{1}D_{0} \rangle + \langle A_{1}B_{1}C_{0}D_{1} \rangle + \langle A_{1}B_{0}C_{1}D_{1} \rangle + \langle A_{0}B_{1}C_{1}D_{1} \rangle \big) - \langle A_{1}B_{1}C_{1}D_{1} \rangle \Big].
    \end{aligned}
\end{equation}
\end{widetext}
Above, $\langle M_{2} \rangle$ is the CHSH functional, whereas $\langle M_{3} \rangle$ is the Mermin functional, associated with the GHZ paradox~\cite{GHZ}. The property that all input strings of the same Hamming weight have the same coefficient makes the MABK expressions invariant under party relabelings. Additionally, note that for $N$ even, every correlator $\langle A_{\bm{x}} \rangle$ for $\bm{x} \in \{0,1\}^{n}$ is included in the expression, whereas for $N$ odd, exactly half of the correlators are included.

\subsection{Self-testing}

It is known that maximum quantum violation of the MABK family is uniquely achieved, up to local isometries, by the GHZ state and pairs of maximally anticommuting observables \cite{KaniewskiSelfTest1,KaniewskiSelfTest2}. This form of uniqueness arising from a Bell expression is known as self-testing. In this work, we take the choice $\ket{\psi_{\mathrm{GHZ}}}=(\ket{0}^{\otimes N} + i\ket{1}^{\otimes N})/\sqrt{2}$, $A_0^{(k)}=\cos \theta_{N}^{+} \, \sigma_{X}+\sin \theta_{N}^{+} \, \sigma_{Y}$, and $A_1^{(k)} = \cos \theta_{N}^{-} \, \sigma_{X} + \sin \theta_{N}^{-} \, \sigma_{Y}$, where $\theta_{N}^{\pm}~=~(\pi/4)(1/N\pm 1)$. 

Whilst self-testing statements can be formally defined between $N$ parties (see, e.g., Refs.~\cite{Supic_2018,Supic_2023}), in our formulation it will only be necessary to use a bipartite definition.

\begin{definition}[Bipartite self-test]
Let $k,l \in \{1,...,N\}$ index two distinct parties. Define the sets of target qubit projectors $\{P_{a_{k}|x_{k}}^{(k)}\}_{a_{k}}$, $\{P_{a_{l}|x_{l}}^{(l)}\}_{a_{l}}$, and a target two qubit state $\ket{\Phi}_{Q_{k}Q_{l}}$. Let $I^{(k,l)}$ be a Bell operator between parties $k$ and $l$, with maximum quantum value $\eta^{\mathrm{Q}}$. The inequality $\langle I^{(k,l)} \rangle \leq \eta^{\text{Q}}$ \emph{self-tests} the target state and measurements if, for all physical quantum strategies $(\tilde{\rho}_{\tilde{Q}_{k}\tilde{Q}_{l}},\tilde{P}_{a_{k}|x_{k}}^{(k)},\tilde{P}_{a_{l}|x_{l}}^{(l)})$ that satisfy $\langle I^{(k,l)} \rangle = \eta^{\mathrm{Q}}$, there exists a local isometery $V:\mathcal{H}_{\tilde{Q}_{k}}\otimes \mathcal{H}_{\tilde{Q}_{l}} \otimes \mathcal{H}_{E} \rightarrow \mathcal{H}_{Q_{k}}\otimes \mathcal{H}_{Q_{l}} \otimes \mathcal{H}_{\mathrm{Junk}}$ and ancillary state $\ket{\xi}_{\mathrm{Junk}}$, such that 
\begin{multline}
    V \Big(\tilde{P}_{a_{k}|x_{k}}^{(k)} \otimes \tilde{P}_{a_{l}|x_{l}}^{(l)} \otimes \mathbb{I}_{E}\Big) \ket{\Psi}_{\tilde{Q}_{k}\tilde{Q}_{l}E}\\ 
    =   \Big(P_{a_{k}|x_{k}}^{(k)} \otimes P_{a_{l}|x_{l}}^{(l)} \Big) \ket{\Phi}_{Q_{k}Q_{l}}   \otimes \ket{\xi}_{\mathrm{Junk}}\,,
\end{multline}
for all $a_k$, $a_l$, $x_k$, $x_l$, where $\ket{\Psi}$ is a purification of $\tilde{\rho}$.
\label{def:selfTest}
\end{definition} 

We define the shifted Bell operator as $\bar{I}^{(k,l)} = \eta^{\mathrm{Q}}\mathbb{I} - I^{(k,l)}$, and we say $\bar{I}^{(k,l)}$ admits a sum-of-squares (SOS) decomposition if there exists a set of polynomials, $\{M^{(k,l)}_{i}\}_{i}$, of the operators $\tilde{P}_{a_{k}|x_{k}}^{(k)},\tilde{P}_{a_{l}|x_{l}}^{(l)}$, satisfying
\begin{equation}
    \bar{I}^{(k,l)} = \sum_{i} M^{(k,l)\dagger}_{i}M^{(k,l)}_{i}.
\end{equation}
SOS decompositions can be used to enforce algebraic constraints on any quantum state $\tilde{\rho}$, and measurements $\tilde{P}_{a_{k}|x_{k}}^{(k)},\tilde{P}_{a_{l}|x_{l}}^{(l)}$, satisfying $\langle \bar{I}^{(k,l)} \rangle = 0$. For example, when $\tilde{\rho} = \ketbra{\psi}{\psi}$ is pure, it must satisfy 
\begin{equation}
    M_{i}^{(k,l)}\ket{\psi} = 0, \ \forall i, \label{eq:const0}
\end{equation}
(or $\mathrm{Tr}[(M_{i}^{(k,l)})^\dagger M_{i}^{(k,l)}\tilde{\rho}]=0$ more generally). Typically, we find that constraints of the form in \cref{eq:const0} are only satisfied by a unique state (up to local isometeries), and when that is the case we call \cref{eq:const0} the self-testing criteria.

\subsection{DI randomness certification} 
We use the conditional von Neumann entropy, $H(\bm{R}|\bm{X}=\bm{x}^{*},E)_{\rho_{\bm{R}E|\bm{x}^*}}$ to measure the asymptotic DI randomness generation rate, where for a bipartite state $\rho_{AB}$, $H(A|B)_{\rho} = H(AB)_{\rho} - H(B)_{\rho}$ and $H(A)_{\rho} = -\tr[\rho_{A}\log \rho_{A}]$ is the von Neumann entropy of a state $\rho_{A}$. $H(\bm{R}|\bm{X}=\bm{x}^{*},E)_{\rho_{\bm{R}E|\bm{x}^*}}$ is the correct quantity for spot-checking DI random number generation, where $\bm{x}^*$ is a specified measurement choice used to generate randomness (see~\cite{Bhavsar2023} for a discussion of this and other possibilities). The asymptotic rate of DI randomness generation is given by a function $R_{f}(\omega)$ defined as
\begin{equation}
 \inf_{\substack{\ket{\Psi}_{\tilde{\bm{Q}}E},\big\{\{\tilde{P}_{a_{k}|x_{k}}^{(k)}\}_{a_{k}} \big\}_{k} \\ \mathrm{compatible \ with} \ f(P_{\mathrm{obs}}) = \omega} } H(\bm{R}|\bm{X}=\bm{x}^{*},E)_{\rho_{\bm{R}E|\bm{x}^*}},\label{eq:ent}
\end{equation}
and we require lower bounds on this quantity over states and measurements compatible with a linear constraint on the observed statistics $P_{\mathrm{obs}}$, $f(P_{\mathrm{obs}}) = \omega$, where $f$ defines a Bell expression and $\omega$ is its observed value. This definition can be straightforwardly extended to multiple functions. The asymptotic rate can be used as a basis for rates with finite statistics using tools such as the entropy accumulation theorem~\cite{DFR,EAT2,MetgerGEAT} (an alternative approach for handling finite statistics is the quantum probability estimation framework~\cite{Zhang20}, but this does not reduce to the single round von Neumann entropy). 

Because we study fundamental limits on certifiable DI randomness we work in the noiseless scenario, using self-testing statements to certify DI randomness. Here, $f(P_{\mathrm{obs}})$ is a Bell expression, and we show that all states and measurements that achieve its maximum violation correspond to a post-measurement state in tensor product with Eve, $\rho_{\bm{R}E|\bm{x}^*} = \rho_{\bm{R}|\bm{x}^*} \otimes \rho_{E}$. This allows us to directly evaluate the conditional entropy $H(\bm{R}|\bm{X}=\bm{x}^{*},E) = H(\bm{R}|\bm{X}=\bm{x}^{*})$, and the infimum becomes trivial since all compatible strategies give rise to the same classical distribution (see Appendix~\ref{app:dec} for details).

\section{Expanding Bell inequalities} \label{sec:CB}
In this section we discuss and enhance the technique presented in~\cite{Curchod_2019}, which allows us to derive the main results of this work. Ref.~\cite{Curchod_2019} introduced a method for building multi-partite Bell inequalities by expanding a bipartite inequality, called the ``seed'', which can be used to witness genuinely multi-partite nonlocality. A new Bell expression is constructed by summing the seed over different subsets of parties, whilst the remaining parties perform some fixed measurement. The more multi-partite nonlocal the correlations, the more bipartite terms are violated. This resembles other recent results that enable multi-partite self-testing by projections onto bipartite subsystems \cite{Li18,Supic_2018,Zwerger19}. 

In this work we extend this technique to make it suitable for DI cryptographic purposes. More precisely, we consider cases where the seed is a bipartite self-test, and use the maximum quantum violation of the expanded Bell expression, constructed in an equivalent manner to \cite{Curchod_2019}, to draw conclusions about the post-measurement state held between the honest parties and Eve. These conclusions allow us to derive rates for randomness certification.

\subsection{Definition}
We begin by introducing a generic formulation for expanding Bell expressions, shortly followed by an example. 

\begin{definition}[Expanded Bell expressions]
    Let $\bm{C}$ be an $N \times N$ nonzero matrix with entries $c_{k,l}$ satisfying $c_{k,l} \in \{0,1\}$ if $k<l$, and $c_{k,l}=0$ otherwise. For each pair $k,l$ for which $c_{k,l} \neq 0$, let $\{I^{(k,l)}_{\bm{\mu}}\}_{\bm{\mu}}$ denote a set of bipartite Bell expressions between parties $k$ and $l$, where $\bm{\mu}\in\{0,1\}^{N-2}$ is a set of $N-2$ measurement outcomes of all parties excluding $k,l$. Further suppose each Bell expression is equivalent up to output relabellings, with a local bound $\eta^{\mathrm{L}}$ strictly less than the maximum quantum value $\eta^{\mathrm{Q}}$. Then the Bell operator $I$ seeded by the set $\{I^{(k,l)}_{\bm{\mu}}\}_{\bm{\mu}}$ is defined as
    \begin{equation}
        I = \sum_{k,l} c_{k,l} \Bigg(\sum_{\bm{\mu}}  \tilde{P}_{\bm{\mu}|\bm{0}}^{\overline{(k,l)}} I^{(k,l)}_{\bm{\mu}}\Bigg), \label{eq:con1}
    \end{equation}
    where $\tilde{P}_{\bm{\mu}|\bm{0}}^{\overline{(k,l)}} = \prod_{k'\in\{1,...,N\}\setminus\{k,l\}} \tilde{P}_{\mu_{k'}|0}^{(k')}$ is the projector for all parties excluding $k$ and $l$, corresponding to the joint input setting $\bm{0}$ and joint outcome $\bm{\mu}$\footnote{We have written the Bell operator in \cref{eq:con1} in terms of projectors, which is convenient since we consider quantum strategies. We can however rewrite it in a theory independent way, by taking the expectation value and making the substitution $p(\bm{a}|\bm{x}) = \langle \tilde{P}_{\bm{a}|\bm{x}} \rangle$.}. \label{def:expandedBI}
\end{definition}
To help understand the definition, consider the following tripartite example. To simplify the notation we label the observables of the three parties $A_{x}=P^{A}_{0|x} - P^{A}_{1|x}$ and similarly for $B$ and $C$. From Ref.~\cite{WBC}, observing saturation of the following Bell expression certifies 2 bits of global randomness,
\begin{equation}
    -3\sqrt{3} \leq  \langle I^{A,B} \rangle  \leq 3\sqrt{3},
\end{equation}
where $I^{A,B} = A_{0}B_{0} + 2\big( A_{0}B_{1} + A_{1}B_{0} - A_{1}B_{1}\big)$. The bound $\langle I^{A,B} \rangle  = \pm 3\sqrt{3}$ is uniquely achieved, up to local isometeries, by the bipartite strategy 
\begin{equation}
    \begin{gathered}
        \rho_{Q_{A}Q_{B}} = \ketbra{\Phi_{\pm}}{\Phi_{\pm}}, \ \ket{\Phi_{\pm}} = \frac{\ket{00}\pm i\ket{11}}{\sqrt{2}},\\
        A_{0} = B_{0} = \sigma_{X}, \ A_{1} = B_{1} = \frac{-\sigma_{X} + \sqrt{3} \sigma_{Y}}{2}.
    \end{gathered}
\end{equation}
Now consider a tripartite extension of the above strategy,
\begin{equation}
    \begin{gathered}
        \rho_{Q_{A}Q_{B}Q_{C}} = \ketbra{\psi_{\mathrm{GHZ}}}{\psi_{\mathrm{GHZ}}}, \\
        A_{0} = B_{0} = C_{0} = \sigma_{X}, \ A_{1} = B_{1} = C_{1} = \frac{-\sigma_{X} + \sqrt{3} \sigma_{Y}}{2}.
    \end{gathered} \label{eq:exampleStr}
\end{equation}
One can verify that when all three parties measure $x=y=z=0$, their outcomes are uniformly distributed, resulting in 3 bits of raw randomness; how can we certify this device-independently? Notice $\ket{\psi_{\mathrm{GHZ}}}$ has the following property: when a single party, say $A$, measures $\sigma_{X}$, the leftover state held by $BC$ is $\ket{\Phi_{\pm}}$ where the sign depends on the parity of $A$'s measurement outcome $\mu \in \{0,1\}$. Parties $B$ and $C$ can now saturate one of the bounds $\langle I^{B,C} \rangle  = \pm 3\sqrt{3}$ conditioned on $A$'s measurement. By using $I^{A,B}_{\mu} = (-1)^{\mu}I^{A,B}$ as the seed, and choosing the matrix $\bm{C}$ with party indices $k,l \in \{A,B,C\}$,
\begin{equation}
    \bm{C} = \begin{bmatrix}
    0 & 0 & 1 \\
    0 & 0 & 1 \\
    0 & 0 & 0
    \end{bmatrix},
\end{equation}
we can construct the expanded Bell expression according to \cref{def:expandedBI}:
\begin{multline}
    I = \underbrace{P_{\mu = 0|x = 0}^{A}I^{B,C} - P_{\mu = 1|x = 0}^{A}I^{B,C}}_{k=A,l=C}  \\ + \underbrace{P_{\mu = 0|y = 0}^{B}I^{A,C} - P_{\mu = 1|y = 0}^{B}I^{A,C}}_{k=B,l=C} \\ = A_{0}I^{B,C} + B_{0}I^{A,C}.
\end{multline}
Due to the properties of $\ket{\psi_{\mathrm{GHZ}}}$ discussed above, we find $\langle I \rangle = 2 \cdot 3\sqrt{3}$ is the maximum quantum value of $\langle I \rangle$, and achieved by the strategy in \cref{eq:exampleStr}. Moreover, it can be shown that this is strictly greater than the maximum local value of $\langle I \rangle$, implying $\langle I \rangle$ is a nontrivial Bell functional. In later sections, we will prove that $\langle I \rangle = 2 \cdot 3\sqrt{3}$ is a sufficient condition for certifying maximum randomness.

Returning to the general form of $I$ from \cref{eq:con1}, we call $I$ an expanded Bell expression, since it is built by combining a bipartite seed, $I_{\bm{\mu}}^{(k,l)}$, conditioned on fixed measurement settings $\bm{0}$ and outcomes $\bm{\mu}$ for the remaining $N-2$ parties. We then sum over all possible outcomes of this $N-2$ party measurement, changing the Bell expression accordingly, and over different combinations of parties $(k,l)$ chosen according to $\bm{C}$. There needs to be a gap between the local and quantum bounds of $\langle I \rangle$ for $I$ to define a nontrivial Bell inequality; typically, we find the following upper bound is achievable, which is strictly greater than the local bound: 
\begin{lemma}
    Let $I$ be an expanded Bell expression according to \cref{def:expandedBI}. The maximum quantum value of $\langle I \rangle$ is upper bounded by $\eta^{\mathrm{Q}}_{N} := \sum_{k,l}c_{k,l}\eta^{\mathrm{Q}}$.\label{lem:conBI}
\end{lemma} 

A proof of \cref{lem:conBI} is given in \cref{app:conBI}. If achievable, the only strategy that can give $\langle I \rangle = \eta^{\mathrm{Q}}_{N}$ is the one satisfying $\big \langle \sum_{\bm{\mu}}\tilde{P}_{\bm{\mu}|\bm{0}}^{\overline{(k,l)}} I^{(k,l)}_{\bm{\mu}} \big \rangle = \eta^{\mathrm{Q}}$ for all pairing combinations $k,l$ with $c_{k,l} > 0$, i.e., the reduced state held between parties $k$ and $l$, following the projection $\tilde{P}_{\bm{\mu}|\bm{0}}^{\overline{(k,l)}}$ on the global state, must achieve the maximum quantum value of $\langle I^{(k,l)}_{\bm{\mu}} \rangle = \eta^{\mathrm{Q}}$. This explains why we must include a dependence on $\bm{\mu}$ for the Bell expression, since it should be tailored to the bipartite state following the outcome $\bm{\mu}$. For simplicity we assume each bipartite state is equivalent up to local unitaries, meaning we effectively use a single seed $I_{\bm{\mu}}^{(k,l)}$. In principle however, one could use different seeds depending on the value of $\bm{\mu}$, which tailors the construction to non-symmetric states~\cite{Curchod_2019}. For our work we only consider the symmetric case. Additionally, one could further generalize by choosing specific measurement choices $\bm{y}$ dependent on the choice of non-projecting parties $k,l$, instead of fixing $\bm{y}=\bm{0}$ as we have done here.    

\subsection{Expanding Bell expressions and entropy evaluation}
Next we present the characteristic of expanded Bell expressions that allows us to certify DI randomness from witnessing its maximum quantum value, forming the main technical contribution of our work. 
\begin{lemma}[Decoupling lemma]
    Let $I$ be an expanded $N$-party Bell expression defined in \cref{def:expandedBI} with binary inputs and outputs, and $c_{k,l} = 1$ if $l=N$ and $k<N$, and zero otherwise. Suppose for every $I_{\bm{\mu}}^{(k,N)}$, there exists an SOS decomposition that self-tests the same pure bipartite entangled state $\ket{\Phi}$ between parties $k$ and $N$, along with some ideal measurements $P^{(k)}_{a_{k}|x_{k}},P^{(N)}_{a_{N}|x_{N}}$, according to \cref{def:selfTest}, satisfying $\bra{\Phi}P^{(k)}_{a_{k}|0} \otimes P^{(N)}_{a_{N}|0} \ket{\Phi} > 0$ for all $a_{k},a_{N}$. Then for any strategy $\ket{\Psi}_{\tilde{\bm{Q}}E},\big\{\{\tilde{P}_{a_{k}|x_{k}}^{(k)}\}_{x_{k}}\big\}_{k}$ that achieves $\langle I \rangle = \eta^{\mathrm{Q}}_{N}$, the post-measurement state $\rho_{\bm{R}E|\bm{x}}$, for measurement settings $\bm{x} = \bm{0}$, admits the tensor product decomposition,
    \begin{equation}
        \rho_{\bm{R}E|\bm{0}} = \rho_{\bm{R}|\bm{0}} \otimes \rho_{E}.
    \end{equation} \label{lem:BIent}
\end{lemma}
Having established that Eve is decoupled it is then straightforward to evaluate the rate in \cref{eq:ent} conditioned on observing the maximum quantum value of $I$ (the quantity $R_I(\eta_{N}^{\mathrm{Q}})$ in the corollary below is an instance of a function of the form $R_{f}(\omega)$ defined in \cref{eq:ent}).
\begin{corollary}
    Let $I,\eta_{N}^{\mathrm{Q}}$ be defined as in \cref{lem:BIent}. Then 
    \begin{equation}
        R_{I}(\eta_{N}^{\mathrm{Q}})
        = H(\{ p(\bm{a}|\bm{0}) \}),
    \end{equation}
    where $H(\{p_{i}\})$ is the Shannon entropy of a distribution $\{p_{i}\}_{i}$. \label{lem:rate}
\end{corollary}
The proof can be found in \cref{app:dec}. \cref{lem:BIent} allows one to relate observing the maximum quantum value of the expanded Bell expression $I$, to a condition on the post-measurement state held by the parties and Eve, namely that it must be in tensor product with the purifying system $E$. This allows one to directly evaluate the conditional entropy according to \cref{lem:rate}. 

\section{Certifying $N$ bits of DI randomness} \label{sec:ex1}

We now consider constructions for generating $N$ bits of DI randomness. As previously mentioned, when $N$ is odd the MABK family can be used~\cite{DharaMaxRand}. For an arbitrary number of parties, we apply the previously outlined techniques to derive a suitable Bell expression for this task. This will involve generalizing a one parameter family of quantum strategies, symmetric under permutation of the parties. 

\subsection{$N$ odd} \label{sec:Nodd}
Using symmetry arguments, the authors of Ref.~\cite{DharaMaxRand} showed that maximum violation of the MABK inequality, for $N$ odd, implies maximum randomness.  
\begin{proposition}[Maximum randomness for $N$ odd~\cite{DharaMaxRand}] When $N$ is odd, maximum quantum violation of the MABK family of Bell inequalities, given by \cref{eq:MABK}, certifies $N$ bits of global DI randomness, i.e.
\begin{equation}
    R_{M_{N}}( 2^{(N-1)/2} ) = N. \label{eq:rateNodd}
\end{equation}
 \label{thm:odd}
\end{proposition}
From a self-testing point of view, maximum violation of the MABK family can only be achieved with a GHZ state and maximally anticommuting observables~\cite{KaniewskiSelfTest1,KaniewskiSelfTest2}. For the form given in \cref{eq:MABK}, recall the following strategy achieves maximum violation:
\begin{equation}
\begin{gathered}
    \rho_{\bm{Q}} = \ketbra{\psi_{\mathrm{GHZ}}}{\psi_{\mathrm{GHZ}}},  \\
    A_{0}^{(k)} = \cos \theta_{N}^{+} \, \sigma_{X}+\sin \theta_{N}^{+} \, \sigma_{Y},  \\ A_{1}^{(k)} = \cos \theta_{N}^{-} \, \sigma_{X} + \sin \theta_{N}^{-} \, \sigma_{Y}, \\ k \in \{1,...,N\},
\end{gathered}\label{eq:oddStrat}
\end{equation}
    with $\theta_{N}^{\pm} = \pi/(4N) \pm \pi/4$, and satisfies $A_{0}^{(k)}A_{1}^{(k)}+A_{1}^{(k)}A_{0}^{(k)}=0$ for each $k$. Following this self-testing property, the infimum in \cref{eq:rateNodd} reduces to a single point up to symmetries, and the fact that the parties hold a pure state implies that Eve must be uncorrelated, i.e., $\rho_{\bm{R}E|\bm{x}^*} = \rho_{\bm{R}|\bm{x}^*} \otimes \rho_{\bm{E}}$. By direct calculation, one can verify for $N=4n-1$, $n =1,2,3,...$, $p(\bm{a}|\bm{0}) = 1/2^{N}, \forall \bm{a}$, hence $H(\{p(\bm{a}|\bm{0})\}) = N$. Similarly for $N= 4n+1$, $n=1,2,3,...$, $p(\bm{a}|\bm{1}) = 1/2^{N}, \forall \bm{a}$, implying $H(\{p(\bm{a}|\bm{1})\}) = N$. 

As an example, taking $N=3$ in the above construction we have $\langle M_3 \rangle =\big( \langle \tilde{A}_{1}^{(0)}\tilde{A}_{0}^{(1)}\tilde{A}_{0}^{(2)} \rangle + \langle \tilde{A}_{0}^{(0)}\tilde{A}_{1}^{(1)}\tilde{A}_{0}^{(2)} \rangle + \langle \tilde{A}_{0}^{(0)}\tilde{A}_{0}^{(1)}\tilde{A}_{1}^{(2)} \rangle - \langle \tilde{A}_{1}^{(0)}\tilde{A}_{1}^{(1)}\tilde{A}_{1}^{(2)} \rangle \big) / 2$, and the strategy in \cref{eq:oddStrat} reaches the algebraic bound of $2$. We therefore have that $\langle M_3 \rangle$ implies $3$ bits of randomness when all three parties use measurement $0$. On the other hand, the correlators in the inequality must take the values $\pm1$ to reach the algebraic bound, so no string of inputs which appears as a correlator in the inequality can generate more than 2 bits of global randomness (in fact they give exactly 2 bits when $\langle M_3 \rangle =2  $).

\subsection{$N$ even}
When $N$ is even, all combinations $\langle \tilde{A}_{x_{1}}^{(1)}\tilde{A}_{x_{2}}^{(2)} \cdots \tilde{A}_{x_{N}}^{(N)}\rangle $ appear in $\langle M_{N} \rangle$, and all must have non-zero weight to achieve the largest possible quantum value. This is incompatible with maximum randomness from one input combination. Hence, achieving the quantum bound does not certify maximum randomness; by direct calculation, maximum MABK violation certifies $N + 1/2 - \log_{2}(1+\sqrt{2})/\sqrt{2} \approx N-0.4$ bits of global DI randomness when all parties use measurement 0.

We now proceed to construct new Bell expressions which certify $N$ bits of randomness in the even case. Specifically, we use the techniques introduced in \cref{sec:CB} to generalize the $N=2$ case, which was addressed in~\cite{WBC}.  

\subsubsection{Bipartite self-tests}
To begin with, we summarize the results obtained for the $N=2$ case. 
\begin{lemma}[$I_{\theta}$ family of self-tests] Define the set $\mathcal{G} = (\pi/4,\pi/2) \cup (\pi/2,3\pi/4) \cup (5\pi/4,3\pi/2) \cup (3\pi/2,7\pi/4)$\footnote{Note that $\mathcal{G}=\left\{\theta\in[0,2\pi]\mid\cos(2\theta)<0,\ \cos(\theta)\neq0\right\}$.}. Let $\theta\in\mathcal{G}$, and define the family of Bell expressions parameterized by $\theta$, for parties $k,l\in \{1,...,N\}$,
\begin{multline}
    \langle I_{\theta}^{(k,l)} \rangle = \cos \theta \cos 2\theta \langle A_{0}^{(k)}A_{0}^{(l)}\rangle - \\ \cos 2\theta \big( \langle A_{0}^{(k)}A_{1}^{(l)} \rangle + \langle A_{1}^{(k)}A_{0}^{(l)}\rangle \big) + \cos \theta \langle A_{1}^{(k)}A_{1}^{(l)} \rangle. \label{eq:2thetaineq}
\end{multline}
Then we have the following:
\begin{enumerate}[(i)]
    \item The local bounds are given by $\pm\eta_{\theta}^{\mathrm{L}}$, where
    \begin{multline}
        \eta_{\theta}^{\mathrm{L}} =  
        \mathrm{max}\Big( | \cos\theta (   1-\cos 2\theta) |, \\
        | \cos\theta\big( 1+\cos 2\theta \big) | + | 2  \cos 2 \theta | \Big).
    \end{multline}
    
    \item The quantum bounds are given by $\pm\eta_{\theta}^{\mathrm{Q}}$, where $\eta_{\theta}^{\mathrm{Q}} = 2\sin^{3}\theta$.
    \item Up to local isometries, there exists a unique strategy that achieves $\langle I_{\theta} \rangle = \eta_{\theta}^{\mathrm{Q}}$:
    \begin{equation}
    \begin{aligned}
        \rho_{Q_{k}Q_{l}} &= \ketbra{\psi}{\psi}, \ \mathrm{where} \ \ket{\psi} = \frac{1}{\sqrt{2}}\big( \ket{00} + i\ket{11}\big),  \\
        A_{0}^{(k)} &= A_{0}^{(l)} = \sigma_{X},  \\
        A_{1}^{(k)} &= A_{1}^{(l)} = \cos \theta \, \sigma_{X} + \sin \theta \, \sigma_{Y}. 
    \end{aligned}\label{eq:thSt}
    \end{equation}
\end{enumerate} \label{lem:n2}
\end{lemma}
The above lemma is a rewriting of Proposition~1 in~\cite{WBC}, and can be recovered as a special case of \cref{lem:n2_2} to follow. Importantly, the self-testing property implies that when their inputs are 0, both parties measure $\sigma_{X}$ on $\ket{\psi}$, which results in 2 bits of randomness, $H(\{p(ab|00)\})=2$. We also remark that the state $\ket{\psi'} = (\ket{00} - i\ket{11})/\sqrt{2}$ with the same measurements achieves $\langle I_{\theta} \rangle = - \eta_{\theta}^{\mathrm{Q}}$. This is a symmetry of the case above, since relabelling the measurement outcomes of party $k$ yields the new Bell expression $I_{\theta}' = -I_{\theta}$. Therefore $\langle I_{\theta} \rangle = - \eta_{\theta}^{\mathrm{Q}}$ implies $\langle I_{\theta}' \rangle = \eta_{\theta}^{\mathrm{Q}}$, which self-tests $\ket{\psi}$ and the negated measurements, or equivalently $\ket{\psi'}$ and the original measurements.

\subsubsection{Target strategy}

For the $N$ partite case, we generalize the above bipartite strategy:
\begin{equation}
\begin{gathered}
    \rho_{\bm{Q}} = \ketbra{\psi_{\mathrm{GHZ}}}{\psi_{\mathrm{GHZ}}}, \\
    A_{0}^{(k)} = \sigma_{X}, \\
    A_{1}^{(k)} = \cos \theta \, \sigma_{X} + \sin \theta \, \sigma_{Y}, \ k \in \{1,...,N\},
\end{gathered}\label{eq:evenStrat}
\end{equation}
where $\theta \in \mathcal{G}$ is a free parameter. As before, we find $p(\bm{a}|\bm{0}) = 1/2^{N}, \forall \bm{a}$, hence this is a target strategy for certifying $N$ bits of global randomness. Our objective now will be to design a tailored Bell inequality that can certify $N$ bits by witnessing its maximum value, attained by the strategy in \cref{eq:evenStrat}. 

\subsubsection{Constructing an $N$-partite Bell inequality}
We use the techniques introduced in \cref{sec:CB} to construct the desired Bell expression by expanding the seed $I_{\theta}$. Studying the strategy in \cref{eq:evenStrat}, we notice that when $N-2$ parties perform their first measurement, $\sigma_{X}$, depending on the outcome parity the post-measurement state for the remaining parties is $(\ket{00} \pm i\ket{11})/\sqrt{2}$. Now, performing the measurements in \cref{eq:evenStrat} on that post-measurement state achieves $\langle I_{\theta} \rangle = \pm\eta_{\theta}^{\mathrm{Q}}$, which, according to the self-testing properties of $I_{\theta}$, will certify 2 bits of randomness. Since the strategy is symmetric under party permutations, we can repeat this argument over different subsets of $N-2$ parties, and conclude that, when every bipartite self-test is satisfied, the output setting $\bm{X} = \bm{0}$ must generate uniform DI randomness. The corresponding Bell inequality is constructed according to the following lemma.  

\begin{lemma}[Bell inequality for maximum randomness]
Let $\bm{\mu}$ be a tuple of $N-2$ measurement outcomes for all parties excluding $k,l$, and $n_{\bm{\mu}} \in \{0,1\}$ be the parity of $\bm{\mu}$. Let $\theta \in \mathcal{G}$ and $\{ I_{\bm{\mu}}^{(k,l)}\}_{\bm{\mu}}$ be a set of bipartite Bell expressions between parties $k,l$, where
\begin{equation}
    I_{\bm{\mu}}^{(k,l)} = (-1)^{n_{\bm{\mu}}}I_{\theta}^{(k,l)} .
\end{equation}
Then the expanded Bell expression given by
\begin{equation}
    I_{\theta} = \sum_{k = 1}^{N-1} \Bigg( \sum_{\bm{\mu}} \tilde{P}^{\overline{(k,N)}}_{\bm{\mu}|\bm{0}}I_{\bm{\mu}}^{(k,N)}\Bigg) \label{eq:NthetaBI}
\end{equation}
has quantum bounds $\pm\eta_{N,\theta}^{\mathrm{Q}}$ where $\eta_{N,\theta}^{\mathrm{Q}} = 2(N-1)\sin^{3}\theta$. Moreover, $\langle I_{\theta} \rangle = \eta_{N,\theta}^{\mathrm{Q}}$ is achieved up to relabellings by the strategy in \cref{eq:evenStrat}, and cannot be achieved classically. \label{lem:NevenBI}
\end{lemma}
\noindent The proof can be found in \cref{app:lem5proof}.

We can now state our first main result:
\begin{proposition}[Maximum randomness certification]
\textit{Achieving the maximum quantum value of the Bell inequality in \cref{lem:NevenBI} certifies $N$ bits of global randomness, i.e.,      
\begin{equation}
   R_{I_{\theta}}(\eta_{N,\theta}^{\mathrm{Q}}) = N. \label{eq:rateNeven}
\end{equation}
}
\label{thm:even}
\end{proposition}
The proof is obtained directly by relating maximal Bell violation to decoupling Eve from the post-measurement state, proven in \cref{lem:BIent,lem:rate}. Specifically, using \cref{lem:NevenBI}, $I_{\theta}$ is an expanded Bell expression  with quantum bounds $\pm \eta^{\mathrm{Q}}_{N,\theta}$. Furthermore, for every sub-expression $I_{\bm{\mu}}^{(k,N)}$, there exists an SOS decomposition that self-tests the pure bipartite state and ideal measurements in \cref{eq:thSt}\footnote{See the SOS decompositions~\cite[Lemma 9]{WBC3} and the proof of \cref{lem:n2_2} for their connection to the family $I_{\theta}$.}. This meets the conditions required to apply \cref{lem:BIent} and \cref{lem:rate}. Proposition~\ref{thm:even} then follows from the fact that $p(\bm{a}|\bm{0}) = 1/2^{N}$ for the strategy in \cref{eq:evenStrat}.

\subsection{MABK value and maximum randomness}

Above we gave a new one parameter construction for certifying $N$ bits of device-independent randomness in any $N$-party 2-input 2-output scenario. Now we will study some properties of the correlations used to certify maximum randomness. We are specifically interested in how maximal randomness and MABK violation tradeoff against each other. In particular, how large can the MABK violation of a correlation be whilst certifying maximal randomness?

\subsubsection{Achievable MABK values with maximum randomness}

When $N$ is odd, we have discussed that one can always certify $N$ bits of DI randomness from maximum MABK violation; for the even case, it is unclear how large the MABK value can be whilst certifying maximum randomness. Using the fact that, for the strategy in \cref{eq:evenStrat}, $\langle A_{\bm{x}} \rangle = \sin n \theta$, where $n = \sum_{k=1}^{N}x_{k}$, its MABK value is given by\footnote{This can be established using the identities $\cos(x) = (e^{ix} + e^{-ix})/2, \ \sin(x) = (e^{ix} - e^{-ix})/(2i)$ and $(1+e^{ix})^{N} = 2^{N}e^{ixN/2}\cos^{N}(x/2)$.}
\begin{multline}
    \langle M_{N}(\theta) \rangle = 2^{\frac{N-1}{2}}\Big(  \cos^{N}\big(\theta/2 + \pi/4\big) \sin \big(N\theta/2 + \pi/4 \big) \\ + \cos^{N}\big(\theta/2 - \pi/4\big) \sin \big(N\theta/2 - \pi/4 \big) \Big).\label{eq:mermSth}
\end{multline} 
Due to \cref{thm:even}, for every $\theta \in \mathcal{G}$, $N$ bits of randomness can be certified by a quantum correlation achieving the MABK value $\langle M_{N}(\theta) \rangle$. We are hence interested in finding the maximum of this function over $\theta \in \mathcal{G}$. In the following proposition, we introduce $\theta^*_N$, which, for each $N$, is a candidate value of $\theta$ achieving this maximum, and prove that it belongs to the set $\mathcal{G}$.

\begin{proposition}
For even $N$, let
    \begin{equation}\label{eq:thetanstar}
        \theta^{*}_{N} = \frac{2\pi t_{N}}{N+1},
    \end{equation}
    where $t_{N}$ is the $(N/2)^{\mathrm{th}}$ element of the sequence $1,1,5,7,3,3,11,13,5,5,...$ given by
    \begin{multline}
        t_{N} = \begin{cases}
            N/4 + 1/2, \  \mathrm{if} \ N = 8n + 2, \\
            N/4,\  \mathrm{if} \ N = 8n + 4, \\
            3N/4 + 1/2,\  \mathrm{if} \ N = 8n + 6, \\
            3N/4 + 1,\  \mathrm{if} \ N = 8n + 8, \ n \in \mathbb{N}_{0}.
        \end{cases} \label{eq:m}
    \end{multline}
    Then $\theta^{*}_{N} \in \mathcal{G}$.
    \label{prop:theta}
\end{proposition}
\noindent Proof can be found in \cref{app:prop1}. 
Using this result, we are able to prove the following. It will be convenient to define
\begin{equation}
    m_{N}^{*} = \langle M_{N}(\theta^{*}_{N}) \rangle.
\end{equation}

\begin{proposition}
Let $N$ be an even integer. For every MABK value $s$ in the range $(1,m_{N}^{*}]$, there exists a $\theta_{s} \in \mathcal{G}$ that satisfies $s = \langle M_{N}(\theta_{s})\rangle$. \label{prop:even1}   
\end{proposition}

\noindent The proof of Proposition~\ref{prop:even1} can be found in \cref{app:evi2}. In \cref{fig:merm1}, we plot the MABK value for a given $N$ as a function of $\theta \in \mathcal{G}$, showing the full range of violations $(1,m_{N}^{*}]$ is accessible. Specifically, the dashed lines in \cref{fig:merm1} indicate an MABK value of 1, while the peak of each graph is $m_{N}^{*}$. As $\theta$ varies, the graph shows a continuous curve between these values. We next state the following technical conjecture.
\begin{conjecture}
The maximum MABK value achievable by quantum strategies with uniformly random outputs on input $\bm{0}$, i.e., $p(\bm{a}|\bm{0}) = 1/2^{N}, \ \forall \bm{a}$, is $m_{N}^{*}$.
\label{conj:maxMerm}
\end{conjecture}
Conjecture \ref{conj:maxMerm} is known to hold in the case $N=2$~\cite{WBC}, and we provide numerical evidence that it holds for general $N$ in~\cref{app:evC1,app:conj2E}. 

In \cref{fig:merm1}, we illustrate that $m_{N}^{*}$ is the largest MABK value that can be achieved by the family of strategies in \cref{eq:evenStrat}, all of which generate maximum randomness. However, we have not ruled out the existence of a strategy outside this family, also generating maximum randomness, achieving an MABK value greater than $m_{N}^{*}$. \cref{conj:maxMerm} is that this is not the case, and $m_{N}^{*}$ is indeed the true maximum. Based on this, we can now state the range of MABK violations for which maximum DI randomness can be certified when $N$ is even. 
\begin{lemma}[MABK violations achievable with maximum DI randomness, even case] \label{lem:RvM1_even}
For even $N$, we have:
\begin{enumerate}[(i)]
    \item For every $s\in(1,m_{N}^{*}]$, there exists a quantum correlation with MABK value $s$ that certifies $N$ bits of device-independent randomness.
\end{enumerate}
Suppose \cref{conj:maxMerm} holds. Then we additionally have:
\begin{enumerate}[(ii)]
    \item $m_{N}^{*}$ is the largest MABK value compatible with strategies generating $N$ bits of device-independent randomness. Or conversely, any strategy achieving an MABK value $s > m_{N}^*$ certifies strictly less than $N$ bits of device-independent randomness. 
\end{enumerate}      
\end{lemma}

Part~$(i)$ can be established by the following reasoning: by varying $\theta \in \mathcal{G}$, the family of strategies in \cref{eq:evenStrat} achieves every MABK value in the interval $(1,m_{N}^{*}]$ by \cref{prop:even1}. Since $\theta \in \mathcal{G}$ for all these values, the corresponding bipartite strategy is self-testable according to \cref{lem:n2}; we can hence apply \cref{lem:NevenBI} to expand the Bell expression to $N$ parties, and it follows from \cref{lem:rate} that this Bell expression certifies $N$ bits of DI randomness. 

An implication of \cref{prop:theta} is that the MABK value $m_{N}^{*}$ becomes an achievable lower bound on the maximum MABK value for quantum correlations certifying maximum randomness. This follows from the same reasoning outlined in the previous paragraph, since $\theta_{N}^{*} \in \mathcal{G}$ implies maximum DI randomness can be certified from the correlations generated by the associated strategy in \cref{eq:evenStrat}. Part~(ii) of Lemma~\ref{lem:RvM1_even} is that this is optimal (if \cref{conj:maxMerm} holds), in the sense that one cannot achieve a larger MABK value whilst simultaneously certifying maximum randomness. As evidence we derive a numerical technique that can generate upper bounds on this MABK value, and show these upper bounds match the lower bounds for some values $N$. See \cref{app:evC1} for the details. We also remark that \cref{conj:maxMerm} is known to hold for the case $N=2$~\cite{WBC}.

We now consider the case of odd MABK expressions. 
\begin{proposition}
    Let $N$ be an odd integer. For every $s\in(1,2^{(N-1)/2})$, there exists a $\theta_{s} \in \mathcal{G}$ that satisfies $s = \langle M_{N}(\theta_{s})\rangle$. \label{prop:odd1}
\end{proposition}
The proof can be found in \cref{app:evi2}. Similar to the even case, in \cref{fig:merm2} we plot the MABK value as a function of $\theta$ for different $N$, and see that all quantum achievable values can be obtained by \cref{eq:evenStrat} for some $\theta \in \mathcal{G}$. We then have the following consequence:

\begin{lemma}[MABK violations achievable with maximum DI randomness, odd case] \label{lem:RvM1_odd}
Let $N$ be odd. For every $s\in(1,2^{(N-1)/2}]$ there exists a quantum correlation with MABK value $s$ that certifies $N$ bits of device-independent randomness.
\end{lemma}

The above lemma asserts that when $N$ is odd, for all quantum-achievable MABK values maximum randomness can be realized simultaneously. For the case of $\langle M_{N} \rangle = 2^{(N-1)/2}$, i.e., the maximum quantum value, randomness can be certified according to \cref{thm:odd}. For all other MABK values, \cref{prop:odd1} implies maximum randomness can be certified by $I_{\theta}$ for some $\theta \in \mathcal{G}$. 

We remark on the nontrivial maximal MABK violation compatible with maximum randomness when $N$ is even, contrasting the odd case; as discussed in \cref{sec:Nodd}, this can be seen as a consequence MABK expressions for even $N$ containing every $N$ party correlator. When maximum randomness is being certified, one of these correlators must be zero, which restricts the maximum MABK violation to be strictly less than the optimal quantum value.  

\begin{figure}[h]
\includegraphics[width=8.4cm]{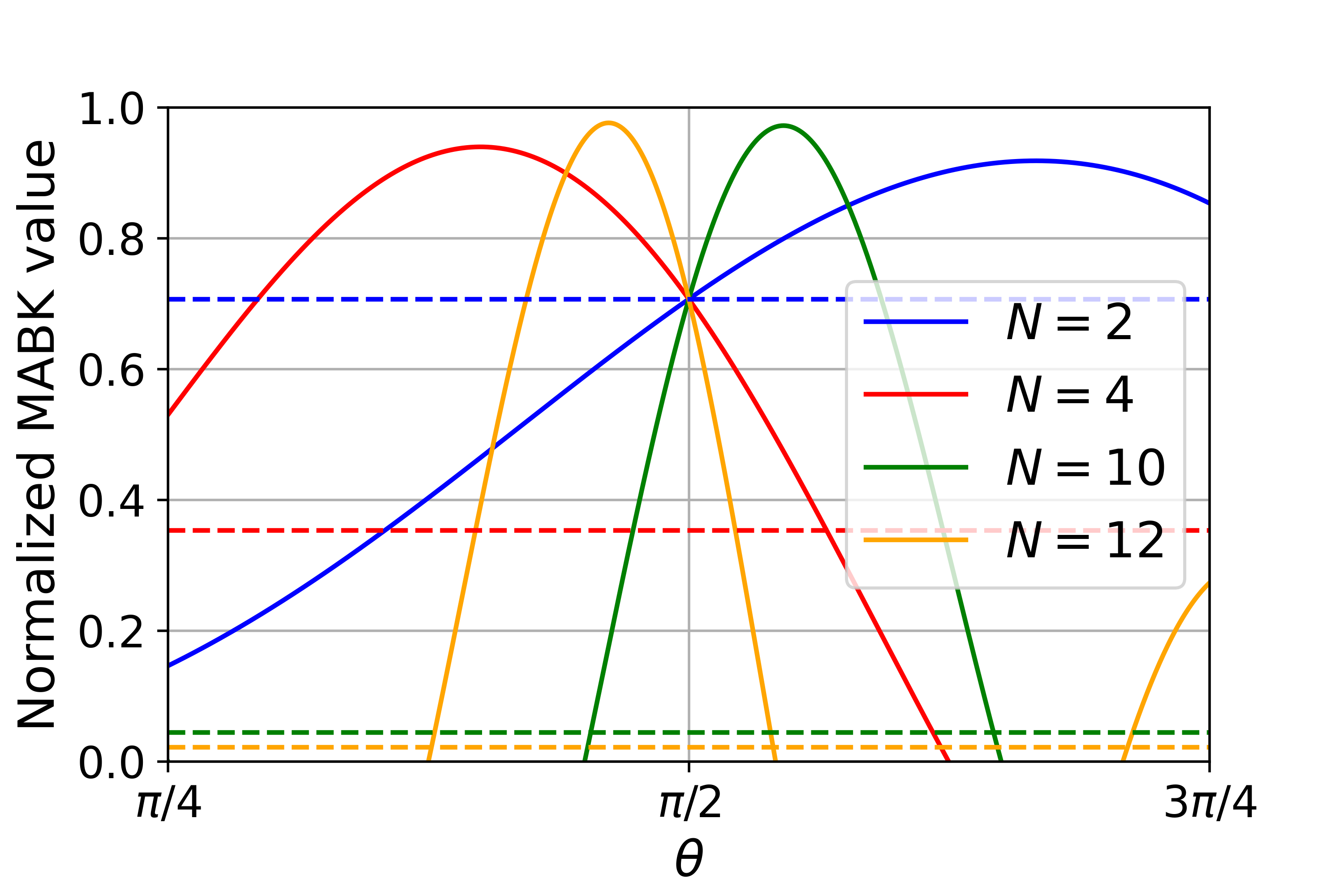}
\includegraphics[width=8.4cm]{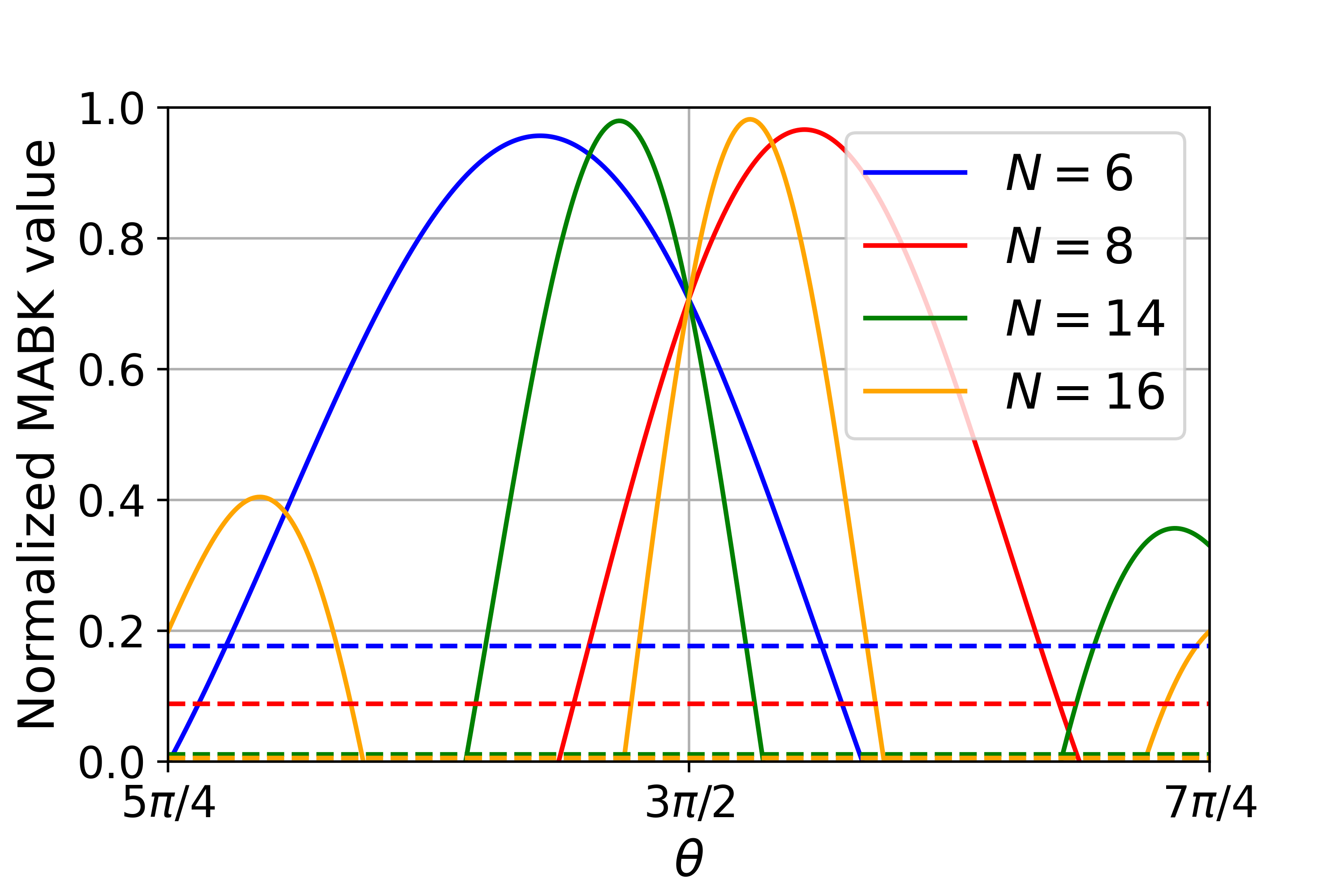}
\centering
\caption{Using the family of strategies in~\eqref{eq:evenStrat} for even $N$ we plot the MABK value renormalized by the maximum quantum value over all strategies (i.e., $2^{(N-1)/2}$) in terms of $\theta \in [\pi/4,3\pi/4]$, and $\theta \in [5\pi/4,7\pi/4]$ respectively. The dashed lines indicate where the strategy becomes local. All points in this interval, excluding the boundaries and center point (since $\pi/2$ and $3\pi/2$ are not in $\mathcal{G}$), correspond to strategies that certify $N$ bits of randomness device-independently using our expanded Bell expressions.}
\label{fig:merm1}
\end{figure}

\begin{figure}[h]
\includegraphics[width=8.4cm]{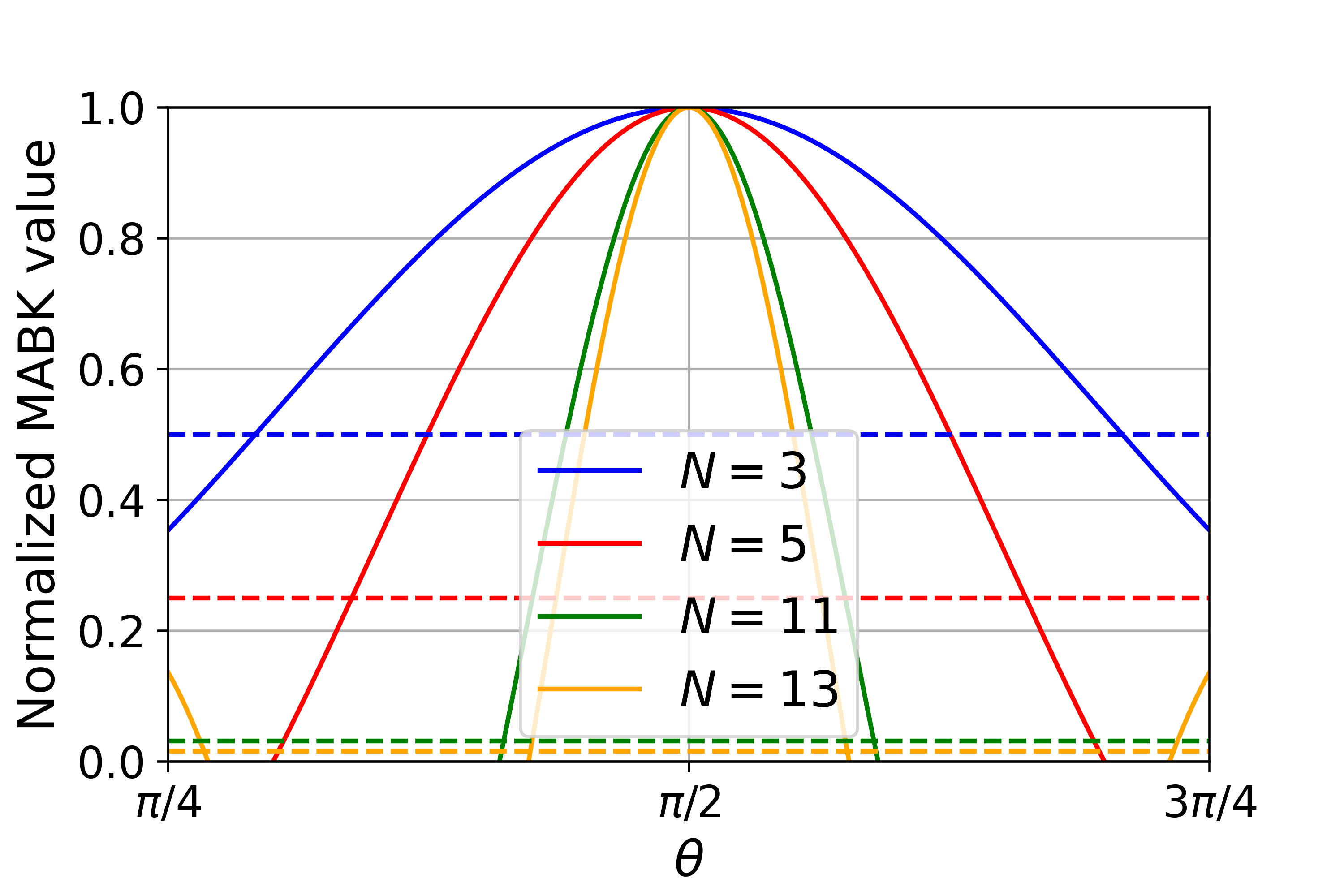}
\includegraphics[width=8.4cm]{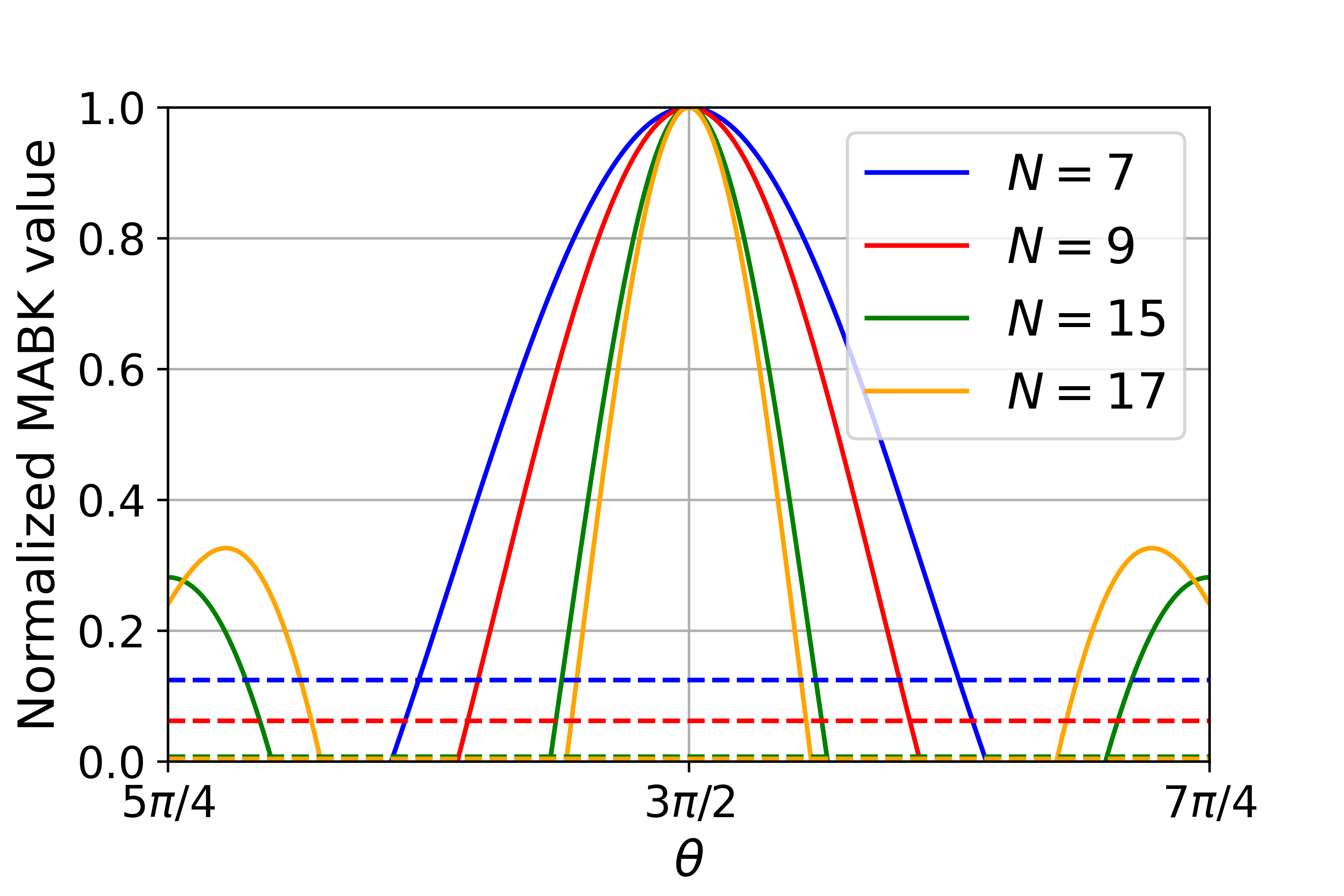}
\centering
\caption{A similar plot to that of Fig.~\ref{fig:merm1} for odd $N$ and second measurement angle of all parties $\theta \in [\pi/4,3\pi/4]$, and $\theta \in [5\pi/4,7\pi/4]$ respectively. All points in this interval, excluding the boundaries, correspond to strategies that can certify $N$ bits of randomness device-independently, using our expanded Bell expressions, or using the MABK inequality for the center points. For $N = 3,7,11,15$, we use the MABK expression given by \cref{eq:MABK}, and for $N = 5,9,13,17$ we use the same expression after relabelling every parties inputs followed by their first measurement's output.}
\label{fig:merm2}
\end{figure}

\subsubsection{Asymptotic behaviour}

We now consider the behaviour of the conjectured maximal MABK value achievable with maximum randomness, $m_{N}^{*}$, for increasingly large even $N$. We show that $m_{N}^{*}$ converges to the largest possible quantum value in this limit. Note this is not based on a conjecture; $m_{N}^{*}$ is an achievable lower bound, and in the following proposition we show that this lower bound tends towards the global upper bound, namely the maximal quantum MABK value.   

\begin{proposition}[Maximum randomness in the asymptotic limit]
In the limit of large even $N$, one can achieve arbitrarily close to the maximum quantum violation of the $N$ party MABK inequality, $2^{(N-1)/2}$, whilst certifying maximum device-independent randomness. \label{thm:asymp}
\end{proposition}
\noindent This is proven in \cref{app:asympP}. 

\subsubsection{Nonlocality and maximum randomness}
Whilst we have studied the MABK values achieved by the strategies in \cref{eq:evenStrat}, we also consider how nonlocal they are, quantified by how much the local set needs to be ``diluted'' to contain them. We refer to this measure as the local dilution, which can be computed via a linear program for small $N$, and we give details in \cref{app:LP}. Our findings are presented in \cref{fig:NLvsTH}. 

\begin{figure}[h]
\includegraphics[width=8.4cm]{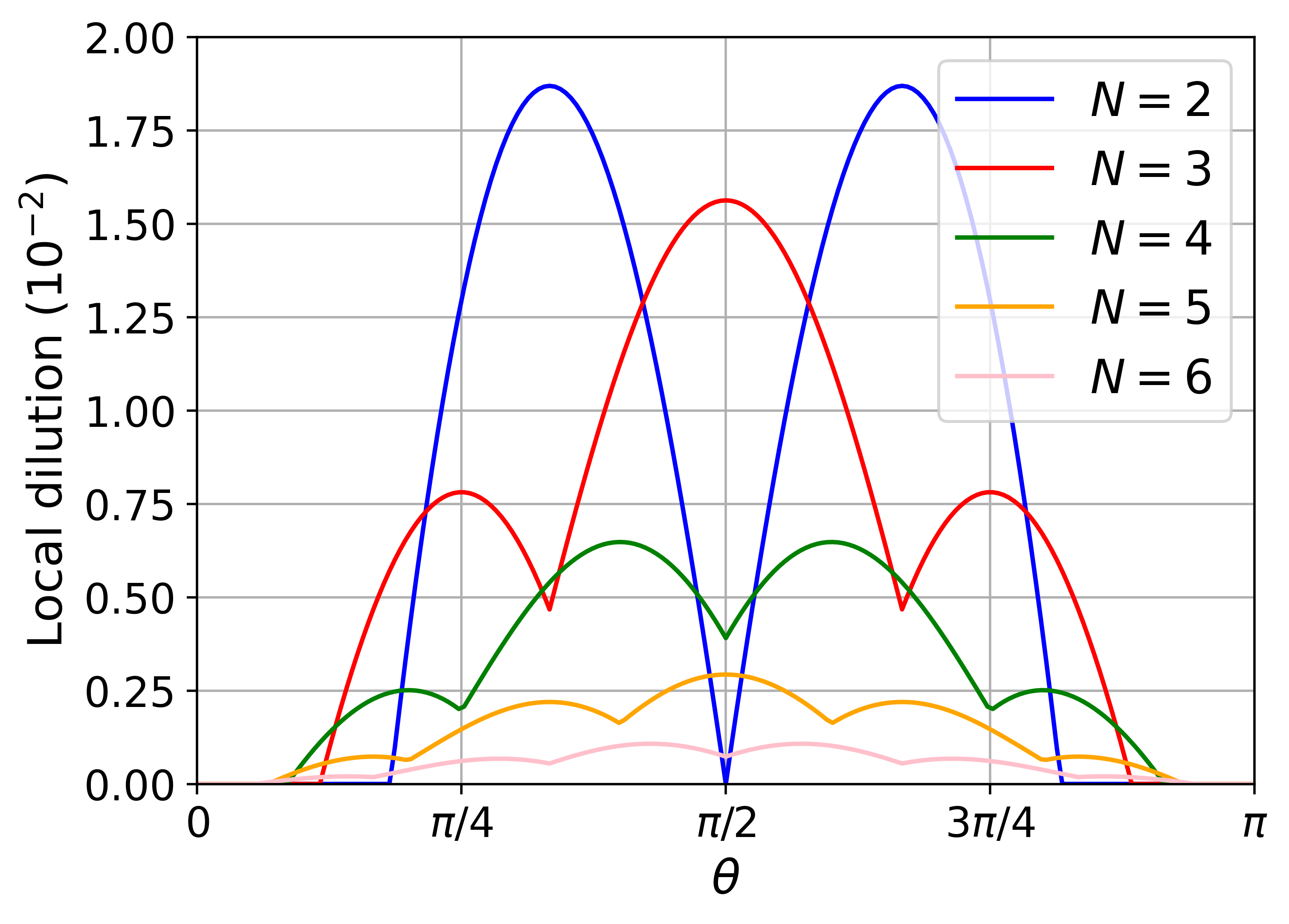}
\centering
\caption{Nonlocality, as measured using local dilution, of the strategies in \cref{eq:evenStrat}, for second measurement angle of all parties $\theta \in [0,\pi]$. All values of $\theta \in (\pi/4,\pi/2) \cup (\pi/2,3\pi/4)$ correspond to strategies which can certify $N$ bits of maximum randomness using our technique for expanding Bell expressions. $\theta = \pi/2$ can also correspond to maximum randomness when $N$ is odd by testing the MABK inequality. }
\label{fig:NLvsTH}
\end{figure}

Interestingly, for $N > 2$ the correlations of the strategy in \cref{eq:evenStrat} are bounded away from the local boundary for any $\theta \in \mathcal{G}$, even for small MABK violation; whilst the violation of a given MABK inequality can be made arbitrarily small, the correlations are still distant from the local boundary. In fact, by comparing \cref{fig:NLvsTH} with \cref{fig:merm1} and \cref{fig:merm2}, one can see how correlations achieving at or below the maximum local MABK value, $\langle M_{N} \rangle \leq 1$, can still be nonlocal and certify maximum randomness. To illustrate this point, consider the $N = 3$ (blue) curve in \cref{fig:merm2}. Provided $\theta \in (\pi/4,\pi/2) \subset \mathcal{G}$, $N$ bits of DI randomness can be certified following previous arguments. However, the blue dashed line in \cref{fig:merm2}, indicating the local bound ($\langle M_{3} \rangle = 1$), intersects the solid curve inside the region $(\pi/4,\pi/2)$, implying that the corresponding strategy in \cref{eq:evenStrat} achieves $\langle M_{3} \rangle \leq 1$ while simultaneously certifying maximum randomness. This is explained by the $N=3$ (red) curve in \cref{fig:NLvsTH}, where it can be seen that in the closed region $[\pi/4,\pi/2]$ the correlations are nonlocal and hence must violate some other Bell inequality. This illustrates the complexity of the tradeoff between randomness versus nonlocality in the multi-partite scenario.  
% \peter{They can even be below the Mermin local bound.}  

\section{Trade-off between DI randomness and MABK violation} \label{sec:ex2}
In this section, we present an achievable lower bound on the maximum device-independent randomness that can be generated by correlations achieving any given MABK violation. Moreover, we conjecture this lower bound to be tight based on numerical evidence. This is achieved by introducing a new family of two parameter quantum strategies along with their self-testing Bell expressions. Using this as the seed, we construct a multi-partite Bell expression which certifies their randomness. 

\subsection{Constructing the Bell expression}

\subsubsection{Bipartite self-tests}
We begin by introducing a versatile family of bipartite self-tests which generalize the results in~\cite{WBC}.
\begin{lemma}[$J_{\phi, \theta}$ family of self-tests]\label{lem:n2_2}
Let $(\phi,\theta) \in \mathbb{R}^{2}$ such that $\cos(2\theta)\cos(2\phi)<0$ and $\cos(\theta-\phi)\neq0$\footnote{The reason for this condition is explained in \cref{app:selfTest}.}. Define the family of Bell expressions parameterized by $\phi$ and $\theta$,
\begin{multline}
    \langle J_{\phi,\theta}^{(k,l)} \rangle = \cos 2\theta \cos (\theta - \phi) \langle A_{0}^{(k)}A_{0}^{(l)}\rangle \\- \cos 2\theta \cos 2 \phi \big( \langle A_{0}^{(k)}A_{1}^{(l)} \rangle + \langle A_{1}^{(k)}A_{0}^{(l)} \rangle \big) \\+ \cos 2\phi \cos (\theta - \phi) \langle A_{1}^{(k)}A_{1}^{(l)} \rangle.
\end{multline}
Then we have the following:
\begin{enumerate}[(i)]
    \item The local bounds are given by $\pm\eta_{\phi,\theta}^{\mathrm{L}}$, where 
    \begin{multline}
        \eta_{\phi,\theta}^{\mathrm{L}} = \max \big\{ |\cos(\theta - \phi)\big(\cos 2\theta  - \cos 2\phi\big) | , \\ | \cos(\theta - \phi)\big(\cos 2\theta + \cos 2\phi \big) | + |2\cos 2\phi \cos 2 \theta | \big \}.
    \end{multline}
    \item The quantum bounds are given by $\pm\eta_{\phi,\theta}^{\mathrm{Q}}$, where $\eta_{\phi,\theta}^{\mathrm{Q}} = 2\sin^{2}(\theta+\phi)\sin (\theta - \phi)$.
    \item  $|\eta^{\mathrm{Q}}_{\phi,\theta}| > \eta^{\mathrm{L}}_{\phi,\theta}$.
    \item Up to local isometries, there exists a unique strategy that achieves $\langle J_{\phi,\theta}^{(k,l)} \rangle = \eta_{\phi,\theta}^{\mathrm{Q}}$:
    \begin{equation}
    \begin{aligned}
        \rho_{Q_{k}Q_{l}} &= \ketbra{\psi}{\psi}, \ \mathrm{where} \ \ket{\psi} = \frac{1}{\sqrt{2}}\big( \ket{00} + i\ket{11}\big),  \\
        A_{0}^{(k)} &= A_{0}^{(l)} = \cos (\phi) \, \sigma_{X} - \sin(\phi) \, \sigma_{Y},  \\
        A_{1}^{(k)} &= A_{1}^{(l)} = \cos (\theta)\, \sigma_{X} + \sin(\theta)\, \sigma_{Y}. 
    \end{aligned}
    \end{equation}
\end{enumerate} 
\end{lemma}
Note that when $\phi = 0$ \cref{lem:n2_2} reduces to \cref{lem:n2}. Moreover, this Bell expression retains the same symmetry properties of the $I_{\theta}$ family, namely that $\langle J_{\phi,\theta} \rangle = - \eta^{\mathrm{Q}}_{\phi,\theta}$ for the state $\ket{\psi'}$ and the same measurements, and is hence a self-test. \cref{lem:n2_2} can be obtained as a corollary of the self-testing results from Refs.~\cite{Le2023,barizien2024,WBC3} (see \cref{app:selfTest} for details).

For future convenience, we define the set $ \mathcal{F} =  \Big[ (-\pi/4,\pi/4) \times \mathcal{G} \Big] \cup \Big[ (-\pi/4,\pi/4) \setminus \{0\} \times \{\pi/2,3\pi/2\}  \Big] \subset \mathbb{R}^{2}$. One can verify that points $(\phi,\theta) \in \mathcal{F}$ satisfy $\cos(2\theta)\cos(2\phi) < 0$ and $\cos(\theta - \phi) \neq 0$, and therefore define a valid self-test according to \cref{lem:n2_2}. 

\subsubsection{Target strategy} 
For the $N$-partite case, we will consider the following strategy:
\begin{equation}
\begin{gathered}
    \rho_{\bm{Q}} = \ketbra{\psi_{\mathrm{GHZ}}}{\psi_{\mathrm{GHZ}}},  \\
    A_{0}^{(k)} = \cos \phi \, \sigma_{X} - \sin \phi \, \sigma_{Y},  \\
    A_{1}^{(k)} = \cos \theta \, \sigma_{X} + \sin \theta \, \sigma_{Y}, \ k \in \{1,...,N\}.
\end{gathered}\label{eq:evenStrat_2}
\end{equation}
Using the fact that $\langle A_{\bm{x}} \rangle = \sin( n\theta - (N-n)\phi)$, where $n = \sum_{k}x_{k}$, the MABK value, defined in \cref{eq:MABK}, of the above strategy is given by 
\begin{multline}
    \langle M_{N}(\phi,\theta) \rangle = 2^{\frac{N-1}{2}}\Big(  \cos^{N}\big[(\theta + \phi)/2 + \pi/4\big] \\ \cdot \sin \big[N(\theta-\phi)/2 + \pi/4 \big] \\ + \cos^{N}\big[(\theta+\phi)/2 - \pi/4\big] \sin \big[N(\theta-\phi)/2 - \pi/4 \big] \Big). \label{eq:MABKval2}
\end{multline}

\subsubsection{Constructing an $N$-partite Bell inequality}
Using the previous two building blocks, we construct the following Bell inequality using the $J_{\phi,\theta}$ expressions as the seed. 
\begin{lemma}
Let $\bm{\mu}$ be a tuple of $N-2$ measurement outcomes for all parties excluding $k$, $l$, and $n_{\bm{\mu}} \in \{0,1\}$ be the parity of $\bm{\mu}$. Let $(\phi,\theta) \in \mathbb{R}^{2}$,
\begin{equation}
    \phi':= \frac{\phi N}{2}, \ \theta' := \theta - \frac{N-2}{2}\phi, \label{eq:defphithp}
\end{equation}and 
\begin{equation}
    I_{\bm{\mu}}^{(k,l)} = (-1)^{n_{\bm{\mu}}} J_{\phi',\theta'}^{(k,l)}.
\end{equation}
Define the following Bell polynomial
\begin{equation}
    I_{\phi',\theta'} := \sum_{k = 1}^{N-1} \Bigg( \sum_{\bm{\mu}} \tilde{P}^{\overline{(k,N)}}_{\bm{\mu}|\bm{0}}I_{\bm{\mu}}^{(k,N)}\Bigg).  \label{eq:genGam}
\end{equation}
If $(\phi',\theta') \in \mathcal{F}$, $I_{\phi',\theta'}$ is an expanded Bell expression, and has quantum bounds $\pm\eta_{N,\phi',\theta'}^{\mathrm{Q}}$ where $\eta_{N,\phi',\theta'}^{\mathrm{Q}} = 2(N-1)\sin^{2}(\theta'+\phi')\sin (\theta'-\phi')$, which can be achieved up to relabellings by the strategy
\begin{equation}
\begin{gathered}
    \rho_{\bm{Q}} = \ketbra{\psi_{\mathrm{GHZ}}}{\psi_{\mathrm{GHZ}}},  \\
    A_{0}^{(k)} = \cos \phi \, \sigma_{X} - \sin \phi \, \sigma_{Y},  \\
    A_{1}^{(k)} = \cos \theta \, \sigma_{X} + \sin \theta \, \sigma_{Y}, \ k \in \{1,...,N\}.
\end{gathered} 
\end{equation} 
In addition, this quantum bound cannot be achieved classically.  \label{lem:NevenBI_2}
\end{lemma}
\noindent The proof is given in \cref{app:lem8proof}. Note that the Bell expressions~\eqref{eq:genGam} are written in terms of the shifted parameters $(\phi',\theta')$, instead of the measurement angles $(\phi,\theta)$. This is because the $N-2$ parties performing the projection no longer use $\sigma_{X}$, but use $\cos (\phi) \, \sigma_{X} - \sin (\phi) \, \sigma_{Y}$. This accumulates a phase factor on the state of the remaining parties, $(\ket{00} + i(-1)^{n_{\bm{\mu}}}e^{i(N-2)\phi}\ket{11})/\sqrt{2} =: \ket{\Phi_{\bm{\mu},\phi}}$, which is equivalent to the action of some local unitary $U$ on $(\ket{00}+i(-1)^{n_{\bm{\mu}}}\ket{11})/\sqrt{2}$. To correct for this, we use the Bell expression $J_{\phi',\theta'}^{(k,l)}$ as the seed. See \cref{app:lem8proof} for the full details.    

\subsection{Randomness versus MABK value}

Using the previously derived Bell inequality, we consider the trade-off between maximum device-independent randomness and MABK violation. Let us define
\begin{equation}
    \theta(\phi) = \frac{N-1}{N+1}\phi + \theta^{*}_{N}, \label{eq:thPhi}
\end{equation}
where $\theta_N^*$ is defined in \cref{eq:thetanstar} and let 
\begin{equation}
    \phi_{N}^{*} = \mathrm{sgn}[\sin(2\theta_{N}^{*})]\frac{\pi}{4N}. \label{eq:thPhi_2}
\end{equation}
For the strategy in \cref{eq:evenStrat_2}, direct calculation yields
\begin{equation}
    \langle M_{N}(\phi_{N}^{*},\theta(\phi_{N}^{*})) \rangle = 2^{(N-1)/2},
\end{equation}
and $\langle M_{N}(0,\theta(0)) \rangle = m_{N}^{*}$. Hence, by varying $\phi \in [0,\phi_{N}^{*}]$\footnote{To ease notation, we write $[0,\phi_{N}^{*}]$ when $\phi_{N}^{*} > 0$ and $\phi_{N}^{*} <0$, interpreting the latter case as $[\phi_{N}^{*},0]$.} and choosing $\theta = \theta(\phi)$, one can obtain the desired range of MABK values $[m_{N}^{*},2^{(N-1)/2}]$, from the conjectured maximum MABK value with maximum randomness to the maximum quantum value. Choosing this parameterization, the raw randomness, $H(\bm{R}|\bm{X}=\bm{x}^{*}) = H(\{p_{\phi}(\bm{a}|\bm{0})\})$, of the strategy in \cref{eq:evenStrat_2} (with $\theta = \theta(\phi)$), as a function of $\phi$, is given by
\begin{equation}
    H(\{p_{\phi}(\bm{a}|\bm{0})\}) \equiv r(\phi) = N - 1 + H_{\mathrm{bin}}\Big[ \frac{1 - \sin N \phi}{2} \Big], \label{eq:rphi}
\end{equation}
where $H_{\mathrm{bin}}$ is the binary entropy function and $r(\phi)$ is a smooth, monotonically decreasing function of $\phi$ in the range $[0,\phi_{N}^{*}]$. 

As shown in \cref{sec:ex1} (by combining \cref{lem:NevenBI} with \cref{prop:theta}), when $\phi=0$ we can certify maximum DI randomness. It also follows from the self-testing properties of the MABK family that when $\phi = \phi^{*}_{N}$ we obtain $N-1 + 3/2 - \log_{2}(1+\sqrt{2})/\sqrt{2} \approx N-0.4$ bits of global DI randomness when all parties use measurement 0. What remains then, is to apply the new Bell expression constructed in \cref{lem:NevenBI_2} to certify $r(\phi)$ bits of DI randomness for $\phi \in (0,\phi^{*}_{N})$. To do so, the following proposition shows that for every $\phi \in [0,\phi_{N}^{*}]$, there exists a valid Bell expression given by \cref{lem:NevenBI_2}. 
\begin{proposition}
    Let $\phi \in [0,\phi_{N}^{*}]$, and $\theta = \theta(\phi)$ as defined in \cref{eq:thPhi,eq:thPhi_2}. Then $(\phi',\theta') \in \mathcal{F}$, where $\phi'$ and $\theta'$ are defined in \cref{lem:NevenBI_2}.
 \label{prop:phi}
\end{proposition}
This is proven in \cref{app:prop2}. \cref{prop:phi} implies that, for every $\phi$ in the range we are interested in, the bipartite expression $(-1)^{n_{\bm{\mu}}}J_{\phi',\theta'}$ is a valid self-test of the strategy
\begin{equation}
\begin{aligned}
        \rho_{Q_{k}Q_{l}} &= \ketbra{\Phi_{\bm{\mu}}}{\Phi_{\bm{\mu}}},  \\
        A_{0}^{(k)} &= A_{0}^{(l)} = \cos \phi' \, \sigma_{X} - \sin \phi' \, \sigma_{Y},  \\
        A_{1}^{(k)} &= A_{1}^{(l)} = \cos \theta'\, \sigma_{X} + \sin \theta'\,\sigma_{Y}. 
\end{aligned}
\end{equation}
After applying the local unitary $U_{\phi} \otimes U_{\phi}$, where $U_{\phi} = \ketbra{0}{0} + e^{i(N-2)\phi/2}\ketbra{1}{1}$, this is equivalent to
\begin{equation}
\begin{aligned}
        \rho_{Q_{k}Q_{l}} &= \ketbra{\Phi_{\bm{\mu},\phi}}{\Phi_{\bm{\mu},\phi}},   \\
        A_{0}^{(k)} &= A_{0}^{(l)} = \cos \phi \, \sigma_{X} - \sin \phi \, \sigma_{Y} , \\
        A_{1}^{(k)} &= A_{1}^{(l)} = \cos \theta\, \sigma_{X} + \sin \theta\, \sigma_{Y}, 
\end{aligned}
\end{equation}
which is exactly the bipartite strategy we wanted to self-test, since it is the one held by parties $k$ and $N$ after the projector $P_{\bm{\mu}|\bm{0}}^{\overline{(k,N)}}$ is applied to the global state $\ket{\psi_{\mathrm{GHZ}}}$. As a result, the correlations generated by the strategy in \cref{eq:evenStrat_2}, by choosing $\theta = \theta(\phi)$ and varying $\phi \in [0,\phi_{N}^{*}]$, maximally violate the Bell inequality $J_{\phi',\theta'}$ with $\phi'$ and $\theta'$ given in \cref{eq:defphithp}. We can therefore employ the decoupling lemma to make the rate unconditioned on Eve, $r(\phi)$, device-independent.
\begin{proposition}
    Achieving the maximum quantum value of the Bell inequality in \cref{lem:NevenBI_2} certifies $r(\phi)$ bits of randomness, i.e., 
    \begin{equation}
        R_{I_{\phi',\theta'}}(\eta^{\mathrm{Q}}_{N,\phi',\theta'}) = r(\phi).
    \end{equation}
\end{proposition}

We have established that, for every MABK value $s \in [m_{N}^{*} , 2^{(N-1)/2}]$, there exists a $\phi_{s} \in [0,\phi_{N}^{*}]$ which defines a quantum strategy achieving $s$, for which its generated randomness $r(\phi_{s})$ is device-independent. We have therefore related every MABK value with a DI rate. We now conjecture that the curve $(s,r(\phi_{s}))$ is optimal in terms of raw randomness, which can then be made device-independent following the above discussion --- see \cref{lem:RvM2}.   

\begin{conjecture}
For $N$ even, the maximum randomness unconditioned on Eve, $r$, that can be generated by quantum strategies achieving an MABK value, $s$, is given by  
\begin{equation}
    r(s) = \begin{cases}
        N, \ s \in (1,m_{N}^{*}], \\
        r(\phi_{s}), \ s \in (m_{N}^{*} ,2^{(N-1)/2} ],
        \end{cases}\label{eq:maxRand}
\end{equation}
where $r(\phi)$ is defined in \cref{eq:rphi}, and 
\begin{equation}
    \phi_{s}\!=\! \mathrm{arg \, min} \big\{ |\phi| : \phi \in [0,\phi_{N}^{*}], \langle M_{N}(\phi,\theta(\phi))\rangle = s \big\}.
\end{equation}
\label{conj:maxMerm_2}
\end{conjecture}
\noindent For the range $s \in (1,m_{N}^{*}]$, and $s\in (m_{N}^{*} ,2^{(N-1)/2} ]$, the rate $r(s)$ and MABK value $s$ is achieved by the family of quantum strategies in \cref{eq:evenStrat} and \cref{eq:evenStrat_2}, and certified device-independently by the Bell expressions in \cref{lem:NevenBI} and \cref{lem:NevenBI_2}, respectively. Similarly to \cref{conj:maxMerm}, \cref{conj:maxMerm_2} is known to hold for the case $N=2$~\cite{WBC}. 

The minimization in the definition of $\phi_{s}$ is included since we have not shown the set $\{ \ \phi \in [0,\phi_{N}^{*}] \ : \ \langle M_{N}(\phi,\theta(\phi))\rangle = s \}$ is unique. Intuitively, the closer $\phi_{s}$ is to zero the closer we are to maximum randomness; hence we expect $r(\phi)$ to be monotonically decreasing with $|\phi|$ for $\phi \in [0,\phi_{N}^{*}]$, and minimization guarantees the best rate at a given $s$. For the examples we have computed, the minimization turns out to be trivial, i.e., $\phi_{s}$ is the unique solution to $\langle M_{N}(\phi_{s},\theta(\phi_{s}))\rangle = s$.   

\begin{lemma} \label{lem:RvM2}
Suppose \cref{conj:maxMerm_2} holds. Then the maximum amount of device-independent randomness that can be certified for the range of MABK values $(1,2^{(N-1)/2}]$ is given by \cref{eq:maxRand}. Moreover, this rate is achievable.
\end{lemma}
The above lemma tells us that if the rate in \cref{eq:maxRand} is optimal without conditioning on Eve, it will also be optimal conditioned on Eve, i.e., the rate can be made device-independent. This is because \cref{prop:theta,prop:phi} ensure that, for every MABK value, the corresponding quantum strategy achieving it, given in \cref{eq:evenStrat} or \cref{eq:evenStrat_2}, can generate certifiable randomness using the corresponding Bell expression $I_{\theta}$ or $I_{\phi',\theta'}$. A rewriting of \cref{lem:RvM2} that does not rely on Conjecture~\ref{conj:maxMerm_2}, would say that~\eqref{eq:maxRand} corresponds to an achievable lower bound on the maximum DI randomness as a function of MABK value. Our results are summarized in \cref{fig:RvMerm}. 

\begin{figure}[h]
\includegraphics[width=8.4cm]{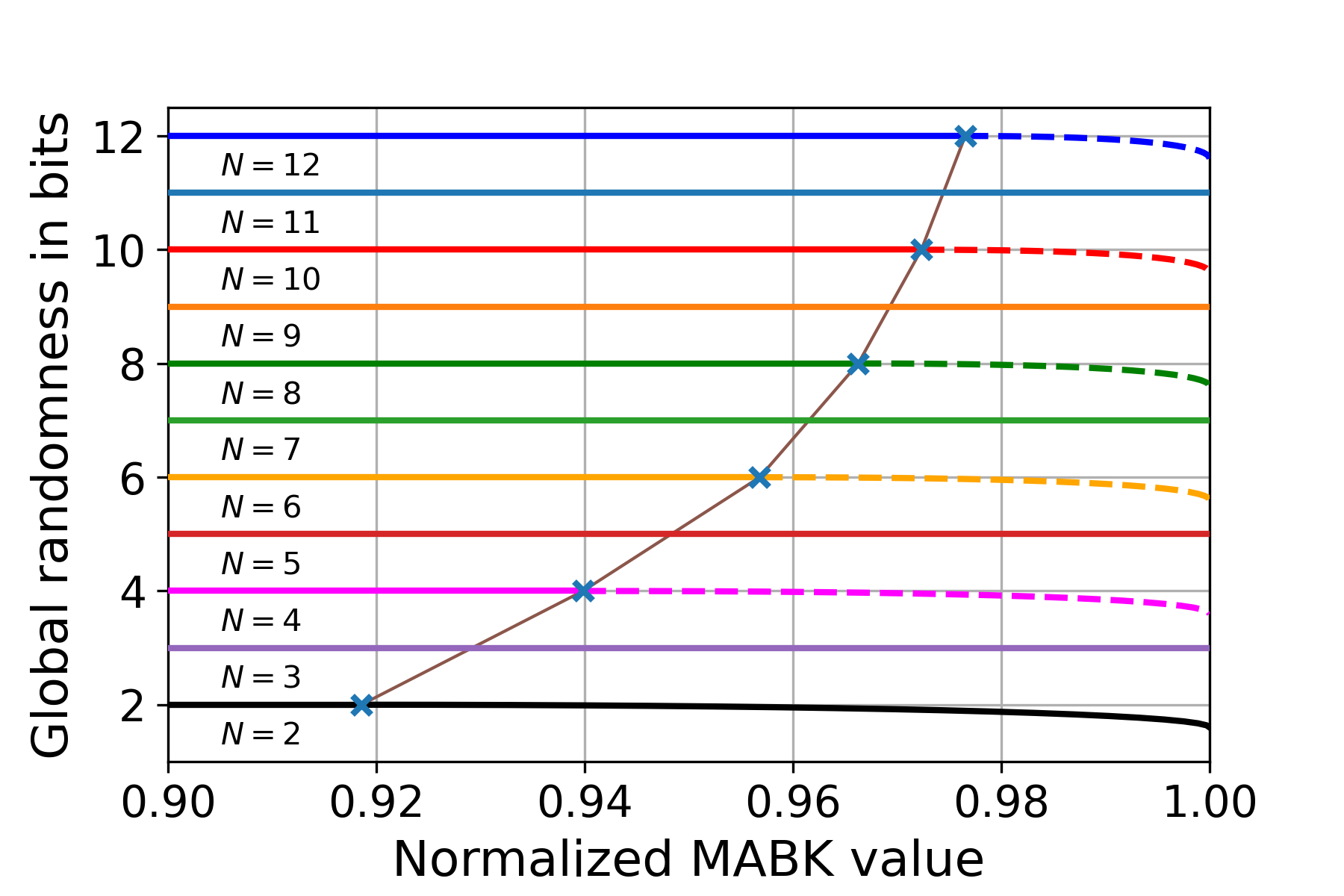}
\centering
\caption{The conjectured curves of maximum device-independent randomness versus MABK value, where the MABK value is normalized by its maximum quantum value $2^{(N-1)/2}$. When $N$ is odd, maximum global randomness can be achieved for any MABK value between the local and quantum bound, indicated by solid lines. The blue crosses indicate the conjectured maximum MABK value for which maximum randomness can be achieved when $N$ is even, $m_{N}^{*}$, which tends to the maximum quantum value as $N \rightarrow \infty$. To the left of the blue crosses, $N$ bits of randomness can be achieved for MABK values between the local bound (not shown on this plot) and $m_{N}^{*}$; since this is the global maximum, it is optimal, indicated by solid lines. To the right of the blue crosses, dashed lines indicate lower bounds on the trade-off between maximum randomness and MABK value from $m_{N}^{*}$ to the maximum quantum value, which are conjectured to be tight. The case of $N=2$ was proven in Ref~\cite{WBC}, which we reproduce with the results of the present paper.}
\label{fig:RvMerm}
\end{figure}

We provide evidence for Conjectures~\ref{conj:maxMerm} and~\ref{conj:maxMerm_2} in \cref{app:conj2E}, where we derive a numerical technique for upper bounding the maximum randomness that can be generated by quantum strategies achieving a given MABK value $s$. An implementation of this technique in python can be found in our GitHub repository~\cite{code}. By studying the numerical results, we find our analytical lower bound agrees with the numerical upper bound to at least 16 digits.   

\subsection{Other extremal Bell inequalities}
In the $N$-partite 2-input 2-output scenario, the MABK family is just one class of extremal Bell inequalities, and the techniques developed in this section can be readily applied to others. For example, when $N=3$ and restricting to Bell inequalities containing only three party correlators, there are a total of 3 non-trivial, new classes, one being the MABK inequality $M_{3}$, and two being of the form~\cite{Werner01}
\begin{align}
    S_{1} &= \frac{1}{4} \sum_{x,y,z}A_{x}B_{y}C_{z} - A_{1}B_{1}C_{1}, \label{eq:S1}\\
    S_{2} &= A_{0}B_{0}(C_{0} + C_{1}) - A_{1}B_{1}(C_{0} - C_{1}), \label{eq:S2}
\end{align}
where we used $A_{x}$ for $A_{x_{1}}^{(1)}$, $B_{y}$ for $A_{x_{2}}^{(2)}$ etc.\ for legibility. $S_{1}$ and $S_{2}$ have local bounds of $1$ and $2$, and maximum quantum values $5/3$ and $2\sqrt{2}$, respectively. We applied our numerical technique (see \cref{app:conj2E}) to find upper bounds on the maximum amount of randomness certifiable whilst achieving a given violation of $S_{1}$ and $S_{2}$. A comparison can be found in \cref{fig:3comp}, and we provide additional figures and details in \cref{app:S1S2}.
\begin{figure}[h]
\includegraphics[width=8.4cm]{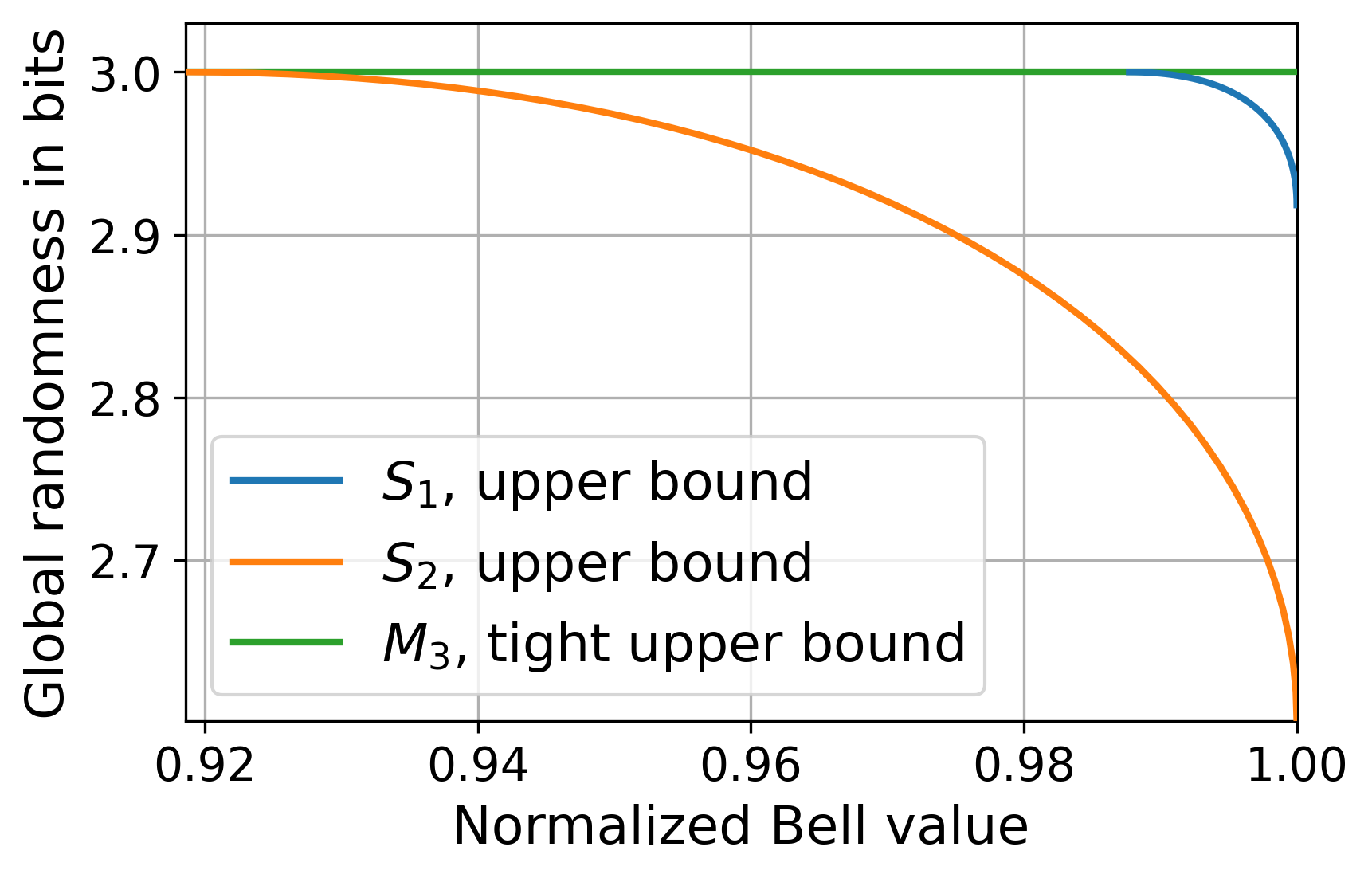}
\centering
\caption{Upper bounds on the trade-off between asymptotic global DI randomness and violation of non-trivial extremal Bell inequalities, for the tripartite scenario with two binary measurements per party. The violation has been normalized by the maximum quantum value. $S_{1}$ and $S_{2}$ are given by \cref{eq:S1} and \cref{eq:S2} respectively, and $M_{3}$ is the MABK expression. }  
\label{fig:3comp}
\end{figure}
We find that both $S_{1}$ and $S_{2}$ exhibit a trade-off for violations close to the quantum maximum. Due to its structure, $S_{2}$ exhibits identical trade-off characteristics to the CHSH inequality found in~\cite{WBC}, with the maximum CHSH value achievable with maximum randomness being numerically close to $3\sqrt{3}/2$. On the other hand, we find that more randomness can be generated from maximum violation of $S_{1}$, which is still strictly less than the maximum certified by MABK. Exact values are given in \cref{app:S1S2}.   

\section{Discussion} \label{sec:dis}
In this work, we studied how MABK violation constrains the generation of global DI randomness for an arbitrary number of parties. Whilst there is no trade-off in the odd case, we conjectured the precise trade-off for the even case. This conjecture is supported by an analytical lower bound and a numerical upper bound that agree up to at least the first 16 digits and have been checked for up to $N=12$. We additionally showed that, in the asymptotic limit, this trade-off vanishes. Our main technical contribution was the extension of a recent tool, originally introduced to witness genuinely multi-partite nonlocality, to randomness certification. 

The difference between the odd and even cases is explained by the odd MABK expressions containing half the full party correlators, allowing for input combinations not included in the Bell expression to be uniform. For even $N$, all correlators are included and must be non-zero to obtain maximum violation, resulting in a trade-off between randomness and MABK value. The number of correlators is $2^N$, whereas the quantum bound is $2^{(N-1)/2}$. Maximum violation is achieved when the product of all correlators and their MABK coefficients have the same value, which in this case will be $2^{-(N+1)/2}$. Hence we find the penalty for fixing one correlator to zero (which is necessary for maximum randomness) diminishes as $N$ grows, and vanishes in the asymptotic limit.   

We also quantified the nonlocality (measured using the dilution) of our constructions for low values of $N$, calculated via a linear program. We found our correlations are always separated from the local boundary (except for the case $N=2$). It would be interesting to find correlations that lie arbitrarily close to the local boundary whilst still certifying maximum randomness, as was found in the $N=2$ case~\cite{AcinRandomnessNonlocality,WBC}. This could entail breaking the symmetry between the second measurements of each party, and the techniques presented here can be applied. Moreover, one could hope to go further and bound the trade-off between randomness and nonlocality in this setting; however, the challenge becomes finding a suitable measure of nonlocality that is efficient to work with. Specifically, for the case $N=2$, constraining the amount of nonlocality of a quantum correlation is equivalent to constraining its CHSH value, which is linear in the correlators. When used as a constraint for maximizing randomness, this can then be handled using the techniques presented in \cref{app:numMethod}, which rely on the Navascués-Pironio-Acín hierarchy~\cite{NPA,NPA2}. For $N > 2$, the dilution measure is calculated using a linear program, therefore requiring a different approach.     

We hope that the results found here will inform future experiments in DIRE; initial works have already made progress in this direction~\cite{seguinard2023}. In such experiments, one should consider the optimal protocol for DIRE. In this work, we have proposed protocols that rely on spot-checking~\cite{MS1,MS2,Bhavsar2023}, where randomness is generated by a single input combination. An alternative approach has been explored in Ref.~\cite{Bhavsar2023}, where it was found that generating randomness from all inputs and averaging the result boosts the rate when the CHSH inequality is used. In the case of maximally violating a self-testing Bell inequality, spot-checking will always be optimal when a specific input combination leads to maximum randomness; having one maximally random setting necessarily implies that at least one other setting combination will not be maximally random for the correlations to be nonlocal, hence averaging will only decrease the overall rate. Whether this continues to hold when the presence of noise reduces the Bell value remains an open question, and the trade-off between using spot-checking and weighted averaging for our constructions deserves future investigation, with the aim of further improving practical rates.

Another interesting direction for future research is to connect expanded Bell inequalities to multi-partite self-testing. Our decoupling result allows the structure of the post-measurement state to be certified following the maximum violation of an expanded Bell expression; the open question we pose is whether an expanded Bell inequality can constitute a self-test of the $N$-party state, opening up a new range of useful applications beyond randomness certification. Intuitively, if each pair of parties can self-test all of their measurements along with a projected state, the marginal information could be sufficient to make a statement about the global state. Such a strategy has already been adopted for self-testing certain multi-partite strategies~\cite{Supic_2018}. However, there are subtleties to consider. For example, expanded Bell expressions tailored to some quantum strategies can result in trivial inequalities (i.e., ones that exhibit no classical-quantum gap), despite the underlying strategy being nonlocal, and furthermore self-testable. We therefore believe that an additional ingredient is needed to self-test generic multi-partite states using the techniques in this work.

Finally, whilst we studied the MABK family of inequalities, one could consider a similar analysis for other families of extremal multi-partite Bell expressions~\cite{Werner01}. For example, we explored upper bounds on the trade-offs for all extremal inequalities when $N=3$, and found the MABK inequality to be optimal for global randomness; it would be interesting to find out if the upper bounds presented for the other inequalities are achievable. One might try to find explicit constructions that match these bounds, and use our techniques to make them device-independent. Analysing other extremal inequalities will build a more complete picture of how DI randomness can be generated in the multi-partite scenario, better informing the way forward for future multi-partite experiments.   

%\vspace{1cm}

\acknowledgements
This work was supported by the UK's Engineering and Physical Sciences Research Council (EPSRC) via the Quantum Communications Hub (Grant No.\ EP/T001011/1), the Integrated Quantum Networks Hub (Grant No.\ EP/Z533208/1), Grant No.\ EP/SO23607/1 and the European Union’s Horizon Europe research and innovation programme under the project “Quantum Secure Networks Partnership” (QSNP, grant agreement No.\ 101114043).

\bibliographystyle{quantum}

% \bibliography{bibtex}

\appendix

\onecolumngrid

\newpage 

\section{Notation}

\begin{table}[h!]
\centering
\begin{tabular}{|c | c | } 
 \hline
 Symbol & Meaning \\ [0.5ex] 
 \hline\hline
 $a_{k}$ & Single bit associated to party $k$\\ \hline
 $\bm{a}$ & $N$ bit string $(a_{1},...,a_{N})$\\ \hline 
 $\bm{\mu}$ & $N-2$ bit string \\
 \hline
 $\bm{\mu}_{\overline{k_{1}k_{2}...k_{m}}}$ & $N-m$ bit string of measurement outcomes for parties $\{1,...,N\} \setminus \{k_{1},k_{2},...,k_{m}\}$ \\ \hline 
 $A_{k}$ & System associated to party $k$\\
 \hline 
 $\bm{A}$ & $N$ partite system $A_{1}...A_{N}$\\
 \hline
 $\bm{A}_{\overline{k_{1}k_{2}...k_{m}}}$ & $N-m$ partite system $A_{1}...A_{N}$ excluding $A_{k_{1}}A_{k_{2}}...A_{k_{m}}$\\
 \hline
 $\tilde{Q}_{k}$ & Quantum system of party $k$ \\
 \hline
 $Q_{k}$ & Qubit system of party $k$ \\
 \hline
 $\tilde{P}_{a_{k}|x_{k}}^{(k)}$ & Projector of party $k$ associated with input $x_{k} \in \{0,1\}$ and output $a_{k} \in \{0,1\}$\\
 \hline
 $\tilde{A}_{x_{k}}^{(k)}$ & Observable of party $k$  associated with input $x_{k}$\\
 \hline
 $\tilde{P}_{\bm{a}|\bm{x}}$ & $\bigotimes_{k=1}^{N} \tilde{P}_{a_{k}|x_{k}}^{(k)}$ where $\bm{a} = (a_{1},...,a_{k})$ and $\bm{x} = (x_{1},...,x_{k})$ \\
 \hline
 $\tilde{P}_{\bm{\mu}|\bm{0}}^{\overline{(k,l)}}$ & $\bigotimes_{k' \in \{1,...,N\} \setminus \{k,l\}} \tilde{P}_{a_{k'}|x_{k'}}^{(k')}$ where $\bm{\mu} = (a_{1},...,a_{k-1},a_{k+1},...,a_{l-1},a_{l+1},..a_{N})$ \\ \hline 
 $R_{f}(\omega)$ & Asymptotic rate defined in \cref{eq:ent} with the constraint $f(P_{\text{obs}}) = \omega$\\ \hline
 $\mathcal{G}$ &  The set $(\pi/4,\pi/2) \cup (\pi/2,3\pi/4) \cup (5\pi/4,3\pi/2) \cup (3\pi/2,7\pi/4)$ \\ \hline
 $\mathcal{F}$ & The set $\Big[ (-\pi/4,\pi/4) \times \mathcal{G} \Big] \cup \Big[ (-\pi/4,\pi/4) \setminus \{0\} \times \{\pi/2,3\pi/2\}  \Big]$ \\ \hline
 $\langle M_{N} \rangle $ & $N$ partite MABK functional (see Eq.~\eqref{eq:MABK}) \\ \hline
 
\end{tabular}
\caption{Summary of notation.}
\label{tab:notation}
\end{table}

% \newpage
\twocolumngrid

\section{Proofs for expanded Bell expressions} 

\subsection{Proof of \cref{lem:conBI}} \label{app:conBI}

\noindent \textbf{Lemma 1.} \textit{Let $I$ be an expanded Bell expression according to \cref{def:expandedBI}. The maximum quantum value of $\langle I \rangle$ is upper bounded by $\eta^{\mathrm{Q}}_{N} := \sum_{k,l}c_{k,l}\eta^{\mathrm{Q}}$.}

\begin{proof}
To prove the above, we show that the operator expression $\bar{I} = \eta^{\mathrm{Q}}_{N} \mathbb{I} - I$ is non-negative. Since each bipartite expression satisfies $\langle I_{\bm{\mu}}^{(k,l)} \rangle \leq \eta^{\mathrm{Q}}$ for all $\bm{\mu},k,l$, we have positive expressions
\begin{equation}
    \bar{I}_{\bm{\mu}}^{(k,l)} = \eta^{\mathrm{Q}}\mathbb{I} - I_{\bm{\mu}}^{(k,l)} \succeq 0 , \ \forall \bm{\mu},k,l.
\end{equation}
We then have
\begin{align}
    \bar{I}_{k,l} &:= \eta^{\mathrm{Q}}\mathbb{I} - \sum_{\bm{\mu}} \tilde{P}_{\bm{\mu}|\bm{0}}^{\overline{(k,l)}} I^{(k,l)}_{\bm{\mu}} \nonumber \\
    &= \sum_{\bm{\mu}} \tilde{P}_{\bm{\mu}|\bm{0}}^{\overline{(k,l)}} \Big( \eta^{\mathrm{Q}}\mathbb{I} - I_{\bm{\mu}}^{(k,l)} \Big) \nonumber \\
    &= \sum_{\bm{\mu}} \tilde{P}_{\bm{\mu}|\bm{0}}^{\overline{(k,l)}} \bar{I}_{\bm{\mu}}^{(k,l)}  \succeq 0. \label{eq:psd1}
\end{align}
Finally, it follows from the above that $\bar{I} \succeq 0$, which proves the upper bound on the maximum quantum value. We can also upper bound the maximum local value by $\sum_{k<l}c_{k,l}\eta^{\mathrm{L}} < \sum_{k<l}c_{k,l}\eta^{\mathrm{Q}} = \eta_{N}^{\mathrm{Q}}$ since $\eta^{\mathrm{L}} < \eta^{\mathrm{Q}}$.
\end{proof}

\subsection{Uniqueness of binary distributions with fixed marginals} 
Before proving \cref{lem:BIent}, we establish the following fact about classical distributions of bit strings when their conditional distributions are fixed. We will later use this result to justify that the distribution used for randomness achieving the quantum bound of an expanded Bell inequality is unique when the seed is a self-test. 
\begin{lemma}
    Let $\bm{A} = A_{1}...A_{N}$, $N\geq 3$, be a random $N$ bit string, which takes values $\bm{a} \in \{0,1\}^{N}$ according to the distribution $p_{\bm{A}}(\bm{a})$. Let $\bm{A}_{\overline{k}}$ be an $N-2$ bit partition of the string $\bm{A}$, excluding bits $N$ and $k\in \{1,...,N-1\}$, taking values $\bm{a}_{\overline{k}} \in \{0,1\}^{N-2}$. Then the distribution $p_{\bm{A}}(\bm{a})$ is entirely determined by the set of conditional distributions $\{p_{A_{k}A_{N}|\bm{A}_{\overline{k}}}(a_{k}a_{N}|\bm{a}_{\overline{k}})\}_{k \in \{1,...,N-1\}}$, provided $p_{A_{k}A_{N}|\bm{A}_{\overline{k}}}(a_{k}a_{N}|\bm{a}_{\overline{k}}) > 0$ for all $ k,a_{k},a_{N},\bm{a}_{\overline{k}}$. \label{lem:uni2}
\end{lemma}
\begin{proof}
    The proof can be established using a recursive argument to solve for a set of marginal terms $p_{\bm{A}_{\overline{l}}}(\bm{a}_{\overline{l}})$, for a fixed $l \in \{1,...,N-1\}$, from which the full distribution can be recovered. Let us choose another partition $k \in \{1,...,N-1\}$ with $k \neq l$. Then we have 
    \begin{equation}
        \begin{aligned}
            p_{\bm{A}}(\bm{a}) &= p_{\bm{A}_{\overline{k}}}(\bm{a}_{\overline{k}}) p_{A_{k}A_{N}|\bm{A}_{\overline{k}}}(a_{k}a_{N}|\bm{a}_{\overline{k}}) \\
            &= p_{\bm{A}_{\overline{l}}}(\bm{a}_{\overline{l}}) p_{A_{l}A_{N}|\bm{A}_{\overline{l}}}(a_{l}a_{N}|\bm{a}_{\overline{l}}). \label{eq:marg1}
        \end{aligned}
    \end{equation}
    By choosing different values of $a_{k}$, this implies the two sets of equations
    \begin{equation}
        \begin{aligned}
            p_{\bm{A}_{\overline{k}}}(\bm{a}_{\overline{k}}) \, & p_{A_{k}A_{N}|\bm{A}_{\overline{k}}}(a_{k}=0,a_{N}|\bm{a}_{\overline{k}}) \\ &= p_{\bm{A}_{\overline{l}}}(a_{k}=0,\bm{a}_{\overline{kl}}) \, p_{A_{l}A_{N}|\bm{A}_{\overline{l}}}(a_{l}a_{N}|a_{k}=0,\bm{a}_{\overline{kl}}), \\ 
            p_{\bm{A}_{\overline{k}}}(\bm{a}_{\overline{k}}) \, & p_{A_{k}A_{N}|\bm{A}_{\overline{k}}}(a_{k}=1,a_{N}|\bm{a}_{\overline{k}}) \\ &= p_{\bm{A}_{\overline{l}}}(a_{k}=1,\bm{a}_{\overline{kl}}) \, p_{A_{l}A_{N}|\bm{A}_{\overline{l}}}(a_{l}a_{N}|a_{k}=1,\bm{a}_{\overline{kl}}).
        \end{aligned}
    \end{equation}
    Since $p_{A_{k}A_{N}|\bm{A}_{\overline{k}}}(a_{k}a_{N}|\bm{a}_{\overline{k}}) \neq 0$, we can rearrange both for $p_{\bm{A}_{\overline{k}}}$ and equate:
    \begin{multline}
        \frac{p_{\bm{A}_{\overline{l}}}(a_{k}=0,\bm{a}_{\overline{kl}}) \, p_{A_{l}A_{N}|\bm{A}_{\overline{l}}}(a_{l}a_{N}|a_{k}=0,\bm{a}_{\overline{kl}})}{p_{A_{k}A_{N}|\bm{A}_{\overline{k}}}(a_{k}=0,a_{N}|\bm{a}_{\overline{k}})} \\
        = \frac{p_{\bm{A}_{\overline{l}}}(a_{k}=1,\bm{a}_{\overline{kl}}) \, p_{A_{l}A_{N}|\bm{A}_{\overline{l}}}(a_{l}a_{N}|a_{k}=1,\bm{a}_{\overline{kl}})}{p_{A_{k}A_{N}|\bm{A}_{\overline{k}}}(a_{k}=1,a_{N}|\bm{a}_{\overline{k}})},
    \end{multline}
    where we have explicitly written in the $a_{k}$ variable, and $\overline{kl}$ denotes a tuple excluding parties $k$ and $l$. The above equation tells us, for every choice of $\bm{a}_{\overline{kl}} \in \{0,1\}^{N-3}$, the unknowns $p_{\bm{A}_{\overline{l}}}(a_{k}=0,\bm{a}_{\overline{kl}})$ and $p_{\bm{A}_{\overline{l}}}(a_{k}=1,\bm{a}_{\overline{kl}})$ are linearly dependent. Hence we consider the $2^{N-3}$ unknowns $p_{\bm{A}_{\overline{l}}}(a_{k}=0,\bm{a}_{\overline{kl}})$, since the $a_{k}=1$ terms can be computed via this linear dependence. 

    Consider choosing $k' \in \{1,...,N-1\} \setminus \{k,l\}$. We can find a new set of equations by repeating the above process with $k'$ instead of $k$, to find:
    \begin{equation}
    \begin{aligned}
        &\frac{p_{\bm{A}_{\overline{l}}}(a_{k} = 0,a_{k'}=0,\bm{a}_{\overline{kk'l}})}{p_{A_{k'}A_{N}|\bm{A}_{\bar{k'}}}(a_{k'}=0,a_{N}|a_{k} = 0,\bm{a}_{\overline{kk'}})} \\ & \hspace{0.8cm}\cdot p_{A_{l}A_{N}|\bm{A}_{\overline{l}}}(a_{l}a_{N}|a_{k} = 0,a_{k'}=0,\bm{a}_{\overline{kk'l}})
        \\ &= \frac{p_{\bm{A}_{\overline{l}}}(a_{k} = 0,a_{k'}=1,\bm{a}_{\overline{kk'l}})}{p_{A_{k'}A_{N}|\bm{A}_{\bar{k'}}}(a_{k'}=1,a_{N}|a_{k} = 0,\bm{a}_{\overline{kk'}})} \\ &\hspace{0.8cm}\cdot p_{A_{l}A_{N}|\bm{A}_{\overline{l}}}(a_{l}a_{N}|a_{k} = 0,a_{k'}=1,\bm{a}_{\overline{kk'l}}),
    \end{aligned}
    \end{equation}
    where we have only included the $a_{k}=0$ case, and written $\overline{kk'l}$ to denote a tuple excluding parties $k,k',l$. Notice that we have now identified linear dependence between $p_{\bm{A}_{\overline{l}}}(a_{k} = 0,a_{k'}=0,\bm{a}_{\overline{kk'l}})$ and $p_{\bm{A}_{\overline{l}}}(a_{k} = 0,a_{k'}=1,\bm{a}_{\overline{kk'l}})$. Neither of these equations contain a marginal with $a_{k}=1$, so they relate pairs of unknowns distinctly to the parings from the previous equations, and are hence linearly independent. As before, we can proceed with the $2^{N-4}$ unknowns $p_{\bm{A}_{\overline{l}}}(a_{k} = 0,a_{k'}=0,\bm{a}_{\overline{kk'l}})$ and recover the others by linear dependence. 

    We can apply the above procedure iteratively a total of $N-2$ times (for every $k \in \{1,...,N-1\} \setminus \{l\}$), halving the number of unknowns on every iteration by establishing linear dependence. Starting with $2^{N-2}$ marginals $p_{\bm{A}_{\overline{l}}}(\bm{a}_{\overline{l}})$, this leaves us with 1 unknown, which can be found from normalization, completing the proof. 
\end{proof}

Now we can present a corollary which will be needed for the proof of \cref{lem:BIent}, showing the conditional distribution used for randomness, $p(\bm{a}|\bm{0})$, achieving the quantum bound of an expanded Bell expression is unique.

\begin{corollary}
    Let $I$ be an expanded Bell inequality according to \cref{def:expandedBI} with $c_{k,N} = 1$ if $k<N$ and zero otherwise, and $\eta_{N}^{\mathrm{Q}}$ be as defined in \cref{lem:conBI}. Let $I$ be constructed using a seed with the self-testing properties described in \cref{lem:BIent}, and $p(\bm{a}|\bm{x})$ denote a quantum behaviour that achieves $\langle I \rangle = \eta^{\mathrm{Q}}_{N}$. Then provided $p(a_{k}a_{N}|\bm{b}_{\overline{k}},\bm{x} =\bm{0}) > 0$ for all $k,a_{k},a_{N},\bm{b}_{\overline{k}}$, the conditional the distribution $p(\bm{a}|\bm{0})$ is unique. \label{cor:uni}       
\end{corollary}

\begin{proof}
    If the quantum value $\langle I \rangle = \eta_{N}^{\mathrm{Q}}$ is achieved, we must have
\begin{align}
    \sum_{k=1}^{N-1} \big \langle \eta^{\mathrm{Q}}\mathbb{I} - \sum_{\bm{\mu}} \tilde{P}_{\bm{\mu}|\bm{0}}^{\overline{(k,N)}} I_{\bm{\mu}}^{(k,N)} \big \rangle = 0, \label{eq:psd2}
\end{align}
which implies $\big \langle \eta^{\mathrm{Q}}\mathbb{I} - \sum_{\bm{\mu}} \tilde{P}_{\bm{\mu}|\bm{0}}^{\overline{(k,N)}} I_{\bm{\mu}}^{(k,N)} \big \rangle = 0$ from \cref{eq:psd1}. Next, we observe
\begin{multline}
    \langle \tilde{P}_{\bm{\mu}|\bm{0}}^{\overline{(k,N)}}I_{\bm{\mu}}^{(k,N)}\rangle = \mathrm{Tr}\Big[ \tilde{P}_{\bm{\mu}|\bm{0}}^{\overline{(k,N)}} \rho_{\tilde{\bm{Q}}E} \tilde{P}_{\bm{\mu}|\bm{0}}^{\overline{(k,N)}} I_{\bm{\mu}}^{(k,N)}\Big] \\ \leq | \mathrm{Tr}\Big[ \tilde{P}_{\bm{\mu}|\bm{0}}^{\overline{(k,N)}} \rho_{\tilde{\bm{Q}}E} \tilde{P}_{\bm{\mu}|\bm{0}}^{\overline{(k,N)}} I_{\bm{\mu}}^{(k,N)}\Big]| \\ \leq \|   \tilde{P}_{\bm{\mu}|\bm{0}}^{\overline{(k,N)}} \rho_{\tilde{\bm{Q}}E} \tilde{P}_{\bm{\mu}|\bm{0}}^{\overline{(k,N)}} I_{\bm{\mu}}^{(k,N)}\|_{1}, \label{eq:ineq1}
\end{multline}
where the second inequality follows from the fact $|\mathrm{Tr}[A]| \leq \|A \|_{1}$ for any operator $A$ (see e.g.~\cite[Lemma 3.3]{Tomamichel_2016}). Above, $\rho_{\tilde{\bm{Q}}E} = \ketbra{\Psi}{\Psi}$, where $\ket{\Psi}$ is the quantum state that realizes the correlations, and identities on $E$ are implicit. Next we combine \cref{eq:ineq1} with H\"older's inequality to obtain  
\begin{multline}
    \langle \tilde{P}_{\bm{\mu}|\bm{0}}^{\overline{(k,N)}}I_{\bm{\mu}}^{(k,N)}\rangle  \leq \|   \tilde{P}_{\bm{\mu}|\bm{0}}^{\overline{(k,N)}} \rho_{\tilde{\bm{Q}}E} \tilde{P}_{\bm{\mu}|\bm{0}}^{\overline{(k,N)}} I_{\bm{\mu}}^{(k,N)}\|_{1} \\ \leq \| \tilde{P}_{\bm{\mu}|\bm{0}}^{\overline{(k,N)}}\rho_{\tilde{\bm{Q}}E}\tilde{P}_{\bm{\mu}|\bm{0}}^{\overline{(k,N)}}  \|_{1} \|I_{\bm{\mu}}^{(k,N)}\|_{\infty} = \langle \tilde{P}_{\bm{\mu}|\bm{0}}^{\overline{(k,N)}} \rangle \eta^{\mathrm{Q}}.    
\end{multline}
Since 
\begin{equation}
    \sum_{\bm{\mu}}\Big( \langle \tilde{P}_{\bm{\mu}|\bm{0}}^{\overline{(k,N)}} \rangle \eta^{\mathrm{Q}} - \langle \tilde{P}_{\bm{\mu}|\bm{0}}^{\overline{(k,N)}}I_{\bm{\mu}}^{(k,N)}\rangle \Big) = 0
\end{equation}
we conclude $\langle \tilde{P}_{\bm{\mu}|\bm{0}}^{\overline{(k,N)}}I_{\bm{\mu}}^{(k,N)}\rangle = \langle \tilde{P}_{\bm{\mu}|\bm{0}}^{\overline{(k,N)}} \rangle \eta^{\mathrm{Q}}$. Now we can write
\begin{equation}
   \bra{\Psi'} I_{\bm{\mu}}^{(k,N)} \ket{\Psi'} = \eta^{\text{Q}},
\end{equation}
where we define the normalized state $\ket{\Psi'} = \tilde{P}_{\bm{\mu}|\bm{0}}^{\overline{(k,N)}} \ket{\Psi} / \sqrt{\langle \tilde{P}_{\bm{\mu}|\bm{0}}^{\overline{(k,N)}} \rangle}$. The self-testing properties of $I_{\bm{\mu}}^{(k,N)}$ (defined in \cref{lem:BIent}) then determine the unique correlations $\bra{\Psi'}\tilde{P}_{a_{k}|x}^{(k)}\tilde{P}_{a_{N}|y}^{(N)} \ket{\Psi'} = p^{k,N}_{\bm{\mu}}(a_{k}a_{N}|xy)$. Using the definition of $\ket{\Psi'}$, we find
\begin{equation}
    \big \langle \tilde{P}_{\bm{\mu}|\bm{0}}^{\overline{(k,N)}} \tilde{P}_{a_{k}|0}^{(k)}\tilde{P}_{a_{N}|0}^{(N)} \big \rangle = \langle \tilde{P}_{\bm{\mu}|\bm{0}}^{\overline{(k,N)}} \rangle p^{k,N}_{\bm{\mu}}(a_{k}a_{N}|00).
\end{equation}
We can now directly apply \cref{lem:uni2}. Explicitly, we have $p_{\bm{A}}(\bm{a}) = \big \langle \tilde{P}_{\bm{\mu}|\bm{0}}^{\overline{(k,N)}} \tilde{P}_{a_{k}|0}^{(k)}\tilde{P}_{a_{N}|0}^{(N)} \big \rangle $, and the $N-1$ fixed conditional distributions $p_{A_{k}A_{N}|\bm{A}_{\overline{k}}}(a_{k}a_{N}|\bm{a}_{\overline{k}}) = p(a_{k}a_{N}|\bm{a}_{\overline{k}},\bm{x} =\bm{0}) = p^{k,N}_{\bm{\mu}}(a_{k}a_{N}|00)$, where we have identified $\bm{\mu} = \bm{a}_{\overline{k}}$.

\end{proof}
\color{black}
\subsection{Proof of \cref{lem:BIent}} \label{app:dec}

We can now proceed to prove \cref{lem:BIent} in the main text, which is restated below:

\vspace{0.2cm}

\noindent \textbf{Lemma 2.}  \textit{Let $I$ be an expanded $N$-party Bell expression defined in \cref{def:expandedBI} with binary inputs and outputs, and $c_{k,l} = 1$ if $l=N$ and $k<N$, and zero otherwise. Suppose that for every $I_{\bm{\mu}}^{(k,N)}$, there exists an SOS decomposition that self-tests the same pure bipartite entangled state $\ket{\Phi}$ between parties $k$ and $N$, along with some ideal measurements $P^{(k)}_{a_{k}|x_{k}},P^{(N)}_{a_{N}|x_{N}}$, according to \cref{def:selfTest}, satisfying $\bra{\Phi}P^{(k)}_{a_{k}|0} \otimes P^{(N)}_{a_{N}|0} \ket{\Phi} > 0$ for all $a_{k},a_{N}$. Then for any strategy $\ket{\Psi}_{\tilde{\bm{Q}}E},\big\{\{\tilde{P}_{a_{k}|x_{k}}^{(k)}\}_{x_{k}}\big\}_{k}$ that achieves $\langle I \rangle = \eta^{\mathrm{Q}}_{N}$, the post-measurement state $\rho_{\bm{R}E|\bm{x}}$, for measurement settings $\bm{x} = \bm{0}$, admits the tensor product decomposition,
    \begin{equation}
        \rho_{\bm{R}E|\bm{0}} = \rho_{\bm{R}|\bm{0}} \otimes \rho_{E}.
    \end{equation}} 

Before proceeding, we summarize the proof. First, we provide an SOS decomposition for the shifted Bell operator $\bar{I}$. This is constructed using the SOS of every bipartite Bell expression $I_{\bm{\mu}}^{(k,N)}$, which exist by assumption. As a result, any state and set of measurements achieving $\bar{I} = 0$ (equivalently, the optimal quantum bound of $I$) must satisfy the same algebraic relations imposed by the SOS polynomials of $I_{\bm{\mu}}^{(k,N)}$ after an appropriate projection is performed by the other parties. In the next step, we employ Jordan's lemma, reducing the problem to qubits and simplifying the following analysis. In the third step we show that, for each $N$ qubit system in the Jordan decomposition, the SOS constraints self-test the bipartite state shared by parties $k$ and $N$ after the other $N-2$ parties have performed a measurement. Finally, we write down the global post-measurement state and show how the aforementioned self-testing properties imply that it must be in a tensor product with $E$.

\begin{proof} 
We proceed in the four steps described above.

\vspace{0.2cm}

\noindent \textbf{Step 1: SOS decompositions.}
By assumption, there exist SOS decompositions for the bipartite Bell operators $\bar{I}_{\bm{\mu}}^{(k,N)}$, i.e., we can write 
\begin{equation}
    \bar{I}_{\bm{\mu}}^{(k,N)} = \sum_{i} M_{\bm{\mu},i}^{(k,N)\dagger}M_{\bm{\mu},i}^{(k,N)},
\end{equation}
where $M_{\bm{\mu},i}^{(k,N)}$ is a polynomial of the operators $A_{x_{k}}^{(k)}$, $A_{x_{N}}^{(N)}$. We now use this to build an SOS decomposition for $\bar{I}$,
\begin{align}
    \bar{I} &= \sum_{k =1}^{N-1} \bar{I}_{k,N} \nonumber \\
    &= \sum_{k =1}^{N-1} \sum_{\bm{\mu}} \tilde{P}_{\bm{\mu}|\bm{0}}^{\overline{(k,N)}} \bar{I}_{\bm{\mu}}^{(k,N)} \nonumber \\
    &= \sum_{k = 1}^{N-1} \sum_{\bm{\mu}} \sum_{i} \Big( \tilde{P}_{\bm{\mu}|\bm{0}}^{\overline{(k,N)}}  M_{\bm{\mu},i}^{(k,N)} \Big)^{\dagger} \Big( \tilde{P}_{\bm{\mu}|\bm{0}}^{\overline{(k,N)}} M_{\bm{\mu},i}^{(k,N)} \Big).
\end{align}
Now, for some physical state $\ket{\Psi}$ and observables, observation of the maximum quantum value of the Bell expression $\langle I \rangle$, or equivalently $\langle \bar{I} \rangle = 0$, implies the algebraic constraints
\begin{equation}  \tilde{P}_{\bm{\mu}|\bm{0}}^{\overline{(k,N)}} M_{\bm{\mu},i}^{(k,N)} \ket{\Psi} = 0, \ \forall k\in \{1,...,N-1\},\bm{\mu},i.
\end{equation}
% \peter{whenever $c_{k,l} \neq 0$.}
Since $\tilde{P}_{\bm{\mu}|\bm{0}}^{\overline{(k,N)}}$ and $M_{\bm{\mu},i}^{(k,N)}$ commute, we define the post-measurement states 
\begin{equation}
    \ket{\Psi_{\bm{\mu},k}} = \frac{\tilde{P}_{\bm{\mu}|\bm{0}}^{\overline{(k,N)}} \ket{\Psi} }{\sqrt{\bra{\Psi}\tilde{P}_{\bm{\mu}|\bm{0}}^{\overline{(k,N)}} \ket{\Psi}}},  
\end{equation}
and arrive at the set of algebraic constraints, for a fixed $k,\bm{\mu}$,
\begin{equation}
    M_{\bm{\mu},i}^{(k,N)} \ket{\Psi_{\bm{\mu},k}} = 0 , \quad \forall i. \label{eq:const1}
\end{equation}

\vspace{0.2cm}

\noindent \textbf{Step 2: Jordan's lemma.} Next, we employ Jordan's lemma~\cite{Jordan}, which allows us to reduce our analysis to a convex combination of $N$-qubit systems. Specifically, we will use the version presented in~\cite[Lemma 4]{Bhavsar2023}, and apply~\cite[Lemma 6]{Bhavsar2023} to consider (without loss of generality) the set of post-measurement states $\rho_{\bm{R}E}$ arising from block diagonal measurement operators with block size $2\times 2$, and a state which has support only on each $2 \times 2$ block. We write $\mathcal{H}_{Q_{k}}$ to denote the qubit Hilbert space of party $k$, and introduce a flag system $\mathcal{H}_{F_{k}}$ which indicates the $2 \times 2$ block. Then for every party $\mathcal{H}_{\tilde{Q}_{k}} = \mathcal{H}_{Q_{k}} \otimes \mathcal{H}_{F_{k}}$, and the purified state takes the form
\begin{equation}
    \ket{\Psi}_{\tilde{\bm{Q}}E} = \sum_{\bm{f}} \sqrt{p_{\bm{f}}} \ket{\bm{f}}_{\bm{F}} \otimes \ket{\varphi^{\bm{f}}}_{\bm{Q}E} \otimes \ket{\bm{f}}_{E'},
\end{equation}
where $\rho^{\bm{f}} = \mathrm{Tr}_{E}\big[\ketbra{\varphi^{\bm{f}}}{\varphi^{\bm{f}}}\big]$ is an $N$-qubit state corresponding to block $\bm{f}=(f_{1},...,f_{N})$, where $f_{k}$ indexes the block for party $k$ and $p_{\bm{f}} > 0$ form a probability distribution. Tracing out $EE'$, the total state received by the parties is then a classical quantum state, where the classical register indexes which $N$-qubit state is in the quantum system. Eve then holds a purification of the state in register $E$ along with a label in register $E'$ telling her which block it pertains to. Similarly, the projectors decompose
\begin{equation}
    \tilde{P}_{a_{k}|x_{k}}^{(k)} = \sum_{f_{k}} \ketbra{f_k}{f_k}_{F_{k}} \otimes P_{a_{k}|x_{k}}^{(f_{k})}, \label{eq:opDecomp}
\end{equation}
where $P_{a_{k}|x_{k}}^{(f_{k})}$ is a rank one single qubit projector on $\mathcal{H}_{Q_{k}}$ corresponding to block $f_{k}$~\cite{Jordan,PABGMS,Bhavsar2023}. The projector performed on all systems except $k$ and $N$ for inputs $0$ is 
\begin{equation}
    \tilde{P}_{\bm{\mu}|\bm{0}}^{\overline{(k,N)}} = \sum_{\bm{f}_{\overline{kN}}} \ketbra{\bm{f}_{\overline{kN}}}{\bm{f}_{\overline{kN}}}_{\bm{F}_{\overline{kN}}} \otimes P_{\bm{\mu}|\bm{0}}^{(\bm{f}_{\overline{kN}})},
\end{equation}
where the notation $\overline{kN}$ is understood to denote a tuple excluding entries $k,N$, e.g., $\bm{f}_{\overline{kN}}$ is a tuple of indices $f_{k'}$ for $k' \neq k,N$. $P_{\bm{\mu}|\bm{0}}^{(\bm{f}_{\overline{kN}})} = \bigotimes_{k'\neq k,N}P^{(f_{k'})}_{\mu_{k'}|0}$ is then a rank one projector on the $N-2$ qubit Hilbert space shared by all parties excluding $k$ and $N$, corresponding to the block $\bm{f}_{\overline{kN}}$. 

We can now apply Jordan's lemma to the projected states $\ket{\Psi_{\bm{\mu},k}} := \tilde{P}_{\bm{\mu}|\bm{0}}^{\overline{(k,N)}} \ket{\Psi} / \sqrt{\langle \tilde{P}_{\bm{\mu}|\bm{0}}^{\overline{(k,N)}} \rangle}$, by computing
\begin{align}
     \ket{\Psi_{\bm{\mu},k}} &= \sum_{\bm{f}} \sqrt{p_{\bm{f}}} \ket{\bm{f}}_{\bm{F}} \nonumber \\
    &  \ \ \ \ \otimes \frac{\Big(\mathbb{I}_{Q_{k}Q_{N}E} \otimes P_{\bm{\mu}|\bm{0}}^{(\bm{f}_{\overline{kN}})} \Big)\ket{\varphi^{\bm{f}}}_{\bm{Q}E}}{\sqrt{\bra{\Psi}\tilde{P}_{\bm{\mu}|\bm{0}}^{\overline{(k,N)}} \ket{\Psi}}} \otimes \ket{\bm{f}}_{E'} \nonumber \\
    &=  \sum_{\bm{f} \in S_{\bm{\mu}}^{k}} \sqrt{p_{\bm{f}}}\ket{\bm{f}}_{\bm{F}} \otimes \ket{\phi_{\bm{\mu}}^{\bm{f}_{\overline{kN}}}}_{\bm{Q}_{\overline{kN}}}  \nonumber \\
    & \ \ \ \ \ \ \ \ \otimes \ket{\psi_{\bm{\mu},k}^{\bm{f}}}_{Q_{k}Q_{N}E} \otimes \ket{\bm{f}}_{E'}
\end{align}
where $S_{\bm{\mu}}^{k}$ is the set of blocks for which the projection is non-zero.
Here $\bm{f}$ is understood to be the concatenation of $f_{k},f_{N},\bm{f}_{\overline{kN}}$, and we used the fact that the qubit projectors are rank one to write the state as a tensor product,
\begin{multline}
    \frac{\Big(\mathbb{I}_{Q_{k}Q_{N}E} \otimes P_{\bm{\mu}|\bm{0}}^{(\bm{f}_{\overline{kN}})} \Big)\ket{\varphi^{\bm{f}}}_{\bm{Q}E}}{\sqrt{\bra{\Psi}\tilde{P}_{\bm{\mu}|\bm{0}}^{\overline{(k,N)}} \ket{\Psi}}} \\ = 
    \begin{cases}
        \ket{\phi_{\bm{\mu}}^{\bm{f}_{\overline{kN}}}}_{\bm{Q}_{\overline{kN}}}  \otimes \ket{\psi_{\bm{\mu},k}^{\bm{f}}}_{Q_{k}Q_{N}E} \ \mathrm{if} \ \bm{f} \in S_{\bm{\mu}}^{k},\\
        0 \ \mathrm{otherwise},
    \end{cases} 
\end{multline}
where $\ket{\phi_{\bm{\mu}}^{\bm{f}_{\overline{kN}}}} = \bigotimes_{k' \neq k,N} \ket{\phi_{\mu_{k'}}^{f_{k'}}}$, and $P^{(f_{k'})}_{\mu_{k'}|0} = \ketbra{\phi_{\mu_{k'}}^{f_{k'}}}{\phi_{\mu_{k'}}^{f_{k'}}}$. $\ket{\psi_{\bm{\mu},k}^{\bm{f}}}_{Q_{k}Q_{l}E}$ is then the state held by parties $k,N$ and Eve following the projection. By writing $\ket{\tilde{\phi}_{\bm{\mu}}^{\bm{f}_{\overline{kN}}}}_{\tilde{\bm{Q}}_{{\overline{kN}}}} = \ket{\bm{f}_{\overline{kN}}}_{\bm{f}_{\overline{kN}}} \otimes \ket{\phi_{\bm{\mu}}^{\bm{f}_{\overline{kN}}}}_{\bm{Q}_{\overline{kN}}}$, and $\ket{\tilde{\psi}_{\bm{\mu},k}^{\bm{f}}}_{\tilde{Q}_{k}\tilde{Q}_{N}E} = \ket{f_{k}f_{N}}_{F_{k}F_{N}}\otimes \ket{\psi_{\bm{\mu},k}^{\bm{f}}}_{Q_{k}Q_{N}E}$, we can re-write the constraints in \cref{eq:const1} as (suppressing the identity operator on Eve's systems)
\begin{multline}
    \sum_{\bm{f}\in S_{\bm{\mu}}^{k}} \sqrt{p_{\bm{f}}} \ket{\tilde{\phi}_{\bm{\mu}}^{\bm{f}_{\overline{kN}}}}_{\tilde{\bm{Q}}_{{\overline{kN}}}} \\ \otimes M_{\bm{\mu},i}^{(k,N)} \ket{\tilde{\psi}_{\bm{\mu},k}^{\bm{f}}}_{\tilde{Q}_{k}\tilde{Q}_{N}E} \otimes  \ket{\bm{f}}_{E'} = 0,
\end{multline}
implying, for a fixed $k,\bm{\mu}$, 
\begin{equation}
    M_{\bm{\mu},i}^{(k,N)} \ket{\tilde{\psi}_{\bm{\mu},k}^{\bm{f}}}_{\tilde{Q}_{k}\tilde{Q}_{N}E} = 0, \ \forall i,\forall \bm{f}\in S_{\bm{\mu}}^{k}. \label{eq:const2}
\end{equation}
This subsequently implies 
\begin{equation}
    M_{\bm{\mu},i}^{(k,N)|f_{k},f_{N}} \ket{\psi_{\bm{\mu},k}^{\bm{f}}}_{Q_{k}Q_{N}E} = 0, \ \forall i,\forall \bm{f}\in S_{\bm{\mu}}^{k}, \label{eq:constNew}
\end{equation}
where we wrote 
\begin{equation}
    M_{\bm{\mu},i}^{(k,N)} = \sum_{f_{k},f_{N}} \ketbra{f_{k},f_{N}}{f_{k},f_{N}}_{F_{k}F_{N}} \otimes M_{\bm{\mu},i}^{(k,N)|f_{k},f_{N}}
\end{equation}
using the decomposition in \cref{eq:opDecomp}.

\vspace{0.2cm}

\noindent \textbf{Step 3: Bipartite self-tests.} For the next part of the proof, we use the self-testing properties of the bipartite Bell expression $I^{(k,l)}_{\bm{\mu}}$, which are in turn derived from the SOS polynomials, and the constraints they impose on any state that achieves the maximum quantum value. Explicitly, for any state which satisfies the polynomial constraints, there exists a local isometry which transforms that state into an ideal state in tensor product with some junk. Due to Jordan's lemma, the constraints in \cref{eq:constNew} certify the the existence of a pair of local unitaries on the qubit registers of parties $k$ and $N$, such that, for every measurement outcome $\bm{\mu}$ and block combination $\bm{f} \in S_{\bm{\mu}}^{k}$, $U_{\bm{\mu},k}^{\bm{f}}$\footnote{Note that the notation $U_{\bm{\mu},k}^{\bm{f}}$ suppresses the tensor product of local unitaries to ease notation, that is, $U_{\bm{\mu},k}^{\bm{f}} = V \otimes W$ where $V$ is a unitary on system $Q_{k}$ and $W$ is a unitary on system $Q_{N}$.}, there exists a fixed two qubit state $\ket{\Phi}$ and projectors $P_{a_{k}|x_{k}}^{(k)}$, $P_{a_{N}|x_{N}}^{(N)}$ satisfying
\begin{multline}
    \Big(\mathbb{I}_{F_{k}F_{N}} \otimes U_{\bm{\mu},k}^{\bm{f}} \otimes \mathbb{I}_{E} \Big) \Big(\tilde{P}_{a_{k}|x_{k}}^{(k)} \otimes \tilde{P}_{a_{N}|x_{N}}^{(N)} \Big) \ket{\tilde{\psi}_{\bm{\mu},k}^{\bm{f}}}_{\tilde{Q}_{k}\tilde{Q}_{N}E} \\ 
    = \ket{f_{k}f_{N}}_{F_{k}F_{N}} \otimes \Big(P_{a_{k}|x_{k}}^{(k)} \otimes P_{a_{N}|x_{N}}^{(N)} \Big)\ket{\Phi}_{Q_{k}Q_{N}} \otimes \ket{\lambda_{\bm{\mu},k}^{\bm{f}}}_{E},
\end{multline}
where $\ket{\lambda_{\bm{\mu},k}^{\bm{f}}}_{E}$ is the junk system left over in the Eve register, which a priori may depend on $\bm{\mu},k$ and $\bm{f}$. Since we required that all bipartite expressions self-test the same state up to local unitaries, the above equation holds for all $\bm{\mu}$. 

\vspace{0.2cm}

\noindent \textbf{Step 4: Post-measurement state.} Finally, we use the existence of the local unitaries $U_{\bm{\mu},k}^{\bm{f}}$ to analyze the global post-measurement state held by all parties following their zero measurement, and storing their outcomes in the register $\bm{R}$, which takes the form 
\begin{multline}
    \rho_{\bm{R}EE'|\bm{0}} = \sum_{a_{k},a_{N},\bm{\mu}} \ketbra{a_{k}a_{N}\bm{\mu}}{a_{k}a_{N}\bm{\mu}}_{\bm{R}} \otimes \\ \mathrm{Tr}_{\tilde{\bm{Q}}}\Big[ \Big( \tilde{P}_{a_{k}|0}^{(k)} \otimes \tilde{P}_{a_{N}|0}^{(N)} \otimes \tilde{P}_{\bm{\mu}|\bm{0}}^{\overline{(k,N)}}  \otimes \mathbb{I}_{EE'}\Big) \ketbra{\Psi}{\Psi} \Big].
\end{multline}
We can rewrite the term inside the partial trace
\begin{align}
    &\Big( \tilde{P}_{a_{k}|0}^{(k)}  \otimes \tilde{P}_{a_{N}|0}^{(N)} \otimes \tilde{P}_{\bm{\mu}|\bm{0}}^{\overline{(k,N)}}  \otimes \mathbb{I}_{EE'}\Big) \ket{\Psi} \nonumber \\
    &= \sqrt{\bra{\Psi}\tilde{P}_{\bm{\mu}|\bm{0}}^{\overline{(k,N)}} \ket{\Psi}} \Big( \tilde{P}_{a_{k}|0}^{(k)} \otimes \tilde{P}_{a_{N}|0}^{(N)} \otimes \mathbb{I}_{\tilde{\bm{Q}}_{(\overline{k,N})}}  \otimes \mathbb{I}_{EE'}\Big)\ket{\Psi_{\bm{\mu},k}} \nonumber \\
    &= \sqrt{\bra{\Psi}\tilde{P}_{\bm{\mu}|\bm{0}}^{\overline{(k,N)}} \ket{\Psi}} \sum_{\bm{f}\in S_{\bm{\mu}}^{k}} \sqrt{p_{\bm{f}}} \ket{\tilde{\phi}_{\bm{\mu}}^{\bm{f}_{\overline{kN}}}}_{\tilde{\bm{Q}}_{{\overline{kN}}}}  \nonumber \\ 
    & \ \ \ \ \ \ \ \ \otimes \Big( \tilde{P}_{a_{k}|0}^{(k)} \otimes \tilde{P}_{a_{N}|0}^{(N)} \Big) \ket{\tilde{\psi}_{\bm{\mu},k}^{\bm{f}}}_{\tilde{Q}_{k}\tilde{Q}_{N}E} \otimes  \ket{\bm{f}}_{E'} \nonumber \\
    &= \sqrt{\bra{\Psi}\tilde{P}_{\bm{\mu}|\bm{0}}^{\overline{(k,N)}} \ket{\Psi}} \sum_{\bm{f}\in S_{\bm{\mu}}^{k}} \sqrt{p_{\bm{f}}} \ket{\tilde{\phi}_{\bm{\mu}}^{\bm{f}_{\overline{kN}}}}_{\tilde{\bm{Q}}_{{\overline{kN}}}} \nonumber \\ 
    & \otimes \ket{f_{k}f_{N}}_{F_{k}F_{N}} \otimes U_{\bm{\mu},k}^{\bm{f}\dagger}\Big( P_{a_{k}|0}^{(k)} \otimes P_{a_{N}|0}^{(N)} \Big) \ket{\Phi}_{Q_{k}Q_{N}} \nonumber \\ & \ \ \ \ \ \ \ \  \otimes \ket{\lambda_{\bm{\mu},k}^{\bm{f}}}_{E} \otimes \ket{\bm{f}}_{E'}.
\end{align}

Now, tracing out $\tilde{\bm{Q}}$ equates to tracing out systems $\bm{F}\bm{Q}$, which yields 
\begin{align}
    &\mathrm{Tr}_{\tilde{\bm{Q}}}\Big[ \Big( \tilde{P}_{a_{k}|0}^{(k)} \otimes \tilde{P}_{a_{N}|0}^{(N)} \otimes \tilde{P}_{\bm{\mu}|\bm{0}}^{\overline{(k,N)}}  \otimes \mathbb{I}_{EE'}\Big) \ketbra{\Psi}{\Psi} \Big] \nonumber \\
    &= \bra{\Psi}\tilde{P}_{\bm{\mu}|\bm{0}}^{\overline{(k,N)}} \ket{\Psi} \bra{\Phi} \Big( P_{a_{k}|0}^{(k)} \otimes P_{a_{N}|0}^{(N)} \Big)\ket{\Phi} \nonumber \\ & \ \ \ \ \ \ \ \ \cdot \sum_{\bm{f}\in S_{\bm{\mu}}^{k}}p_{\bm{f}} \ketbra{\lambda_{\bm{\mu},k}^{\bm{f}}}{\lambda_{\bm{\mu},k}^{\bm{f}}}_{E}\otimes \ketbra{\bm{f}}{\bm{f}}_{E'}.
\end{align}
By employing \cref{cor:uni}, 
% \peter{This part requires a particular $C_{k,l}$ structure} 
we know that, following observation of $\langle I \rangle = \eta^{\mathrm{Q}}$, the marginals $\bra{\Psi}\tilde{P}_{\bm{\mu}|\bm{0}}^{\overline{(k,l)}} \ket{\Psi}$ are unique; hence the following equality holds:
\begin{equation}
    \bra{\Psi}\tilde{P}_{\bm{\mu}|\bm{0}}^{\overline{(k,N)}} \ket{\Psi} \bra{\Phi} \Big( P_{a_{k}|0}^{(k)} \otimes P_{a_{N}|0}^{(N)}\Big) \ket{\Phi} = p(a_{k},a_{N},\bm{\mu}|\bm{0}),
\end{equation}
where $\{p(a_{k},a_{N},\bm{\mu}|\bm{0})\}_{a_{k},a_{N},\bm{\mu}}$ are those unique correlations needed to achieve $\langle I \rangle = \eta^{\mathrm{Q}}$. We then find
\begin{multline}
    \rho_{\bm{R}EE'|\bm{0}} =  \sum_{a_{k},a_{N},\bm{\mu}} p(a_{k},a_{N},\bm{\mu}|\bm{0}) \ketbra{a_{k}a_{N}\bm{\mu}}{a_{k}a_{N}\bm{\mu}}_{\bm{R}}  \\ \otimes  \sum_{\bm{f}\in S_{\bm{\mu}}^{k}}p_{\bm{f}} \ketbra{\lambda^{\bm{f}}_{\bm{\mu},k}}{\lambda^{\bm{f}}_{\bm{\mu},k}}_{E}\otimes \ketbra{\bm{f}}{\bm{f}}_{E'}. \label{eq:finalRho}
\end{multline}

For a fixed $k \in \{1,...,N-1\}$, Eve can still establish correlations with the $N$ parties excluding $k$ and $N$, by choosing the sets $\{S_{\bm{\mu}}^{k}\}_{\bm{\mu}}$ to be disjoint. She could then distinguish the different outcomes $\bm{\mu}$ by a projective measurement on $E'$. She can also establish correlations by making the set $\{\ket{\lambda^{\bm{f}}_{\bm{\mu},k}}\}_{\bm{\mu}}$ distinguishable and measuring $E$. To show that Eve is in fact uncorrelated with all parties, we use the fact that there are $N-1$ choices of $k$, i.e., one can choose different combinations of parties to maximally violate the bipartite self-tests. 

We now include extra notation $\overline{kN}$, $\bm{\mu} \rightarrow \bm{\mu}_{\overline{kN}}$, which denotes a tuple of $N-2$ outcomes for all parties excluding $k$ and $N$, and define the following state, which is just a rewriting of \cref{eq:finalRho} making the choice $k$ explicit:
\begin{multline}
    \rho_{\bm{R}EE'|\bm{0}}^{k} :=  \sum_{a_{k},a_{N},\bm{\mu}_{\overline{kN}}} p(a_{k},a_{N},\bm{\mu}_{\overline{kN}}|\bm{0}) \\ \cdot \ketbra{a_{k}a_{N}\bm{\mu}_{\overline{kN}}}{a_{k}a_{N}\bm{\mu}_{\overline{kN}}}_{\bm{R}} \\  \otimes  \sum_{\bm{f}\in S_{\bm{\mu}_{\overline{kN}}}^{k}}p_{\bm{f}} \ketbra{\lambda^{\bm{f}}_{\bm{\mu}_{\overline{kN}},k}}{\lambda^{\bm{f}}_{\bm{\mu}_{\overline{kN}},k}}_{E}\otimes \ketbra{\bm{f}}{\bm{f}}_{E'}. \label{eq:pmslam}
\end{multline}
Now, we could equally define a second state $\rho_{\bm{R}EE'|\bm{0}}^{k'}$ where $k' \neq k,N$ in the exact same way, except we choose $k'$ instead of $k$ as the self-testing party. Since both states are a rewriting of the same post measurement state $\rho_{\bm{R}EE|\bm{0}}$ given by \cref{eq:finalRho}, they must be equal,
and we establish the equalities 
\begin{equation}
    \rho_{\bm{R}EE'|\bm{0}}^{k} = \rho_{\bm{R}EE'|\bm{0}}^{k'}, \ \forall k,k' \in \{1,...,N-1\}, \ k' \neq k. \label{eq:densityEq}
\end{equation}
These equalities will allow us to place constraints on the sets $S_{\bm{\mu}_{\overline{kN}}}^{k}$, and ultimately show that they are all equal. Specifically, \cref{eq:densityEq} implies\footnote{Note that we have implicitly equated the rewriting of the same outcome string $\bm{a}=a_{k}a_{N}\bm{\mu}_{\overline{kN}} = a_{k'}a_{N}\bm{\mu}_{\overline{k'N}}$, and only considered cases that appear in the summation, i.e., strings with nonzero probability, $p(\bm{a}|\bm{0}) > 0$. }
\begin{multline}
    \sum_{\bm{f}\in S_{\bm{\mu}_{\overline{kN}}}^{k}}p_{\bm{f}} \ketbra{\lambda^{\bm{f}}_{\bm{\mu}_{\overline{kN}},k}}{\lambda^{\bm{f}}_{\bm{\mu}_{\overline{kN}},k}}_{E}\otimes \ketbra{\bm{f}}{\bm{f}}_{E'} \\
    = \sum_{\bm{f}\in S_{\bm{\mu}_{\overline{k'N}}}^{k'}}p_{\bm{f}} \ketbra{\lambda^{\bm{f}}_{\bm{\mu}_{\overline{k'N}},k'}}{\lambda^{\bm{f}}_{\bm{\mu}_{\overline{k'N}},k'}}_{E}\otimes \ketbra{\bm{f}}{\bm{f}}_{E'} .\label{eq:eveEq}
\end{multline}
By tracing out system $E$, and using the fact that $p_{\bm{f}} \neq 0$, \cref{eq:eveEq} can only hold if
\begin{equation}
    S_{\bm{\mu}_{\overline{kN}}}^{k} = S_{\bm{\mu}_{\overline{k'N}}}^{k'}, \ \forall a_{1},...,a_{N-1}. 
\end{equation}
Consider the case $k = 1$ and $k' = 2$. Then we have for a fixed $a_{3},...,a_{N-1}$,
\begin{equation}
    S_{(a_{2},a_{3},...,a_{N-1})}^{1} = S_{(a_{1},a_{3},a_{4},...,a_{N-1})}^{2}, \ \forall a_{2}.
\end{equation}
Since the RHS is independent of $a_{2}$, we must have that $S^{1}_{(a_{2},a_{3},...,a_{N-1})} \equiv S^{1}_{(a_{3},...,a_{N-1})}$, i.e., the LHS is also independent of $a_{2}$. We can apply the same argument for $k' = 3,4,...,N-1$, deducing that 
\begin{equation}
    S_{(a_{2},a_{3},...,a_{N-1})}^{1} \equiv S,
\end{equation}
which is independent of the outcome string $\bm{\mu}_{\overline{1N}} = (a_{2},a_{3},...,a_{N-1})$.
There is nothing unique about choosing $k = 1$, allowing us to write
\begin{multline}
    \rho_{\bm{R}EE'|\bm{0}}^{k} =  \sum_{a_{k},a_{N},\bm{\mu}_{\overline{kN}}} p(a_{k},a_{N},\bm{\mu}_{\overline{kN}}|\bm{0}) \\ \cdot \ketbra{a_{k}a_{N}\bm{\mu}_{\overline{kN}}}{a_{k}a_{N}\bm{\mu}_{\overline{kN}}}_{\bm{R}} \\  \otimes  \sum_{\bm{f}\in S}p_{\bm{f}} \ketbra{\lambda^{\bm{f}}_{\bm{\mu}_{\overline{kN}},k}}{\lambda^{\bm{f}}_{\bm{\mu}_{\overline{kN}},k}}_{E}\otimes \ketbra{\bm{f}}{\bm{f}}_{E'},
\end{multline}
Hence Eve can learn nothing about any of the outcomes by measuring her register $E'$. What remains is to deal with the vectors $\ket{\lambda^{\bm{f}}_{\bm{\mu}_{\overline{kN}},k}}$. We will show that the set $\{\ket{\lambda^{\bm{f}}_{\bm{\mu}_{\overline{kN}},k}}\}_{\bm{\mu}_{\overline{kN}}}$ contains one linearly independent vector, hence Eve can learn nothing about the outcome $\bm{\mu}_{\overline{kN}}$ by measuring her register $E$.

Using the same approach as before, \cref{eq:densityEq} now implies, for every $\bm{f} \in S$,
\begin{multline}
    \sum_{a_{k}a_{N}\bm{\mu}_{\overline{kN}}} \bra{\Psi}\tilde{P}_{\bm{\mu}_{\overline{kN}}|\bm{0}}^{\overline{(k,N)}} \ket{\Psi} \bra{\Phi} \Big( P_{a_{k}|0}^{(k)} \otimes P_{a_{N}|0}^{(N)}\Big) \ket{\Phi} \\ \cdot \ketbra{a_{k}a_{N}\bm{\mu}_{\overline{kN}}}{a_{k}a_{N}\bm{\mu}_{\overline{kN}}} \otimes \ketbra{\lambda^{\bm{f}}_{\bm{\mu}_{\overline{kN}},k}}{\lambda^{\bm{f}}_{\bm{\mu}_{\overline{kN}},k}} \\
    = \sum_{a_{k'}a_{N}\bm{\mu}_{\overline{k'N}}} \bra{\Psi}\tilde{P}_{\bm{\mu}_{\overline{k'N}}|\bm{0}}^{\overline{(k',N)}} \ket{\Psi} \bra{\Phi} \Big( P_{a_{k'}|0}^{(k')} \otimes P_{a_{N}|0}^{(N)}\Big) \ket{\Phi} \\ \cdot \ketbra{a_{k'}a_{N}\bm{\mu}_{\overline{k'N}}}{a_{k'}a_{N}\bm{\mu}_{\overline{k'N}}} \otimes \ketbra{\lambda^{\bm{f}}_{\bm{\mu}_{\overline{k'N}},k'}}{\lambda^{\bm{f}}_{\bm{\mu}_{\overline{k'N}},k'}}. \label{eq:denstConst2}
\end{multline}
Now equating for a given value of $\bm{a} = (a_{1},...,a_{k},...,a_{k'},...,a_{N})$, we obtain
\begin{multline}
    \bra{\Psi}\tilde{P}_{\bm{\mu}_{\overline{kN}}|\bm{0}}^{\overline{(k,N)}} \ket{\Psi} \bra{\Phi} \Big( P_{a_{k}|0}^{(k)} \otimes P_{a_{N}|0}^{(N)}\Big) \ket{\Phi} \\ \cdot \ketbra{\lambda^{\bm{f}}_{\bm{\mu}_{\overline{kN}},k}}{\lambda^{\bm{f}}_{\bm{\mu}_{\overline{kN}},k}} \\ = \bra{\Psi}\tilde{P}_{\bm{\mu}_{\overline{k'N}}|\bm{0}}^{\overline{(k',N)}} \ket{\Psi} \bra{\Phi} \Big( P_{a_{k'}|0}^{(k)} \otimes P_{a_{N}|0}^{(N)}\Big) \ket{\Phi} \\ \cdot \ketbra{\lambda^{\bm{f}}_{\bm{\mu}_{\overline{k'N}},k'}}{\lambda^{\bm{f}}_{\bm{\mu}_{\overline{k'N}},k'}}.
\end{multline}
We can hence conclude the two vectors are linearly dependent,
\begin{equation}
    \ket{\lambda^{\bm{f}}_{\bm{\mu}_{\overline{kN}},k}} \sim \ket{\lambda^{\bm{f}}_{\bm{\mu}_{\overline{k'N}},k'}}, 
\end{equation}
where $\sim$ denotes equality up to a constant\footnote{Note it is sufficient to just consider the cases for which the probabilities $\bra{\Psi}\tilde{P}_{\bm{\mu}_{\overline{kN}}|\bm{0}}^{\overline{(k,N)}} \ket{\Psi} \bra{\Phi} \Big( P_{a_{k}|0}^{(k)} \otimes P_{a_{N}|0}^{(N)}\Big) \ket{\Phi}$ are nonzero. If they are zero, then the vector $\ket{\lambda^{\bm{f}}_{\bm{\mu}_{\overline{kN}},k}}$ will not appear in the summation, and Eve will never observe it.}. Notice that on the left hand side of the above equation, one has the freedom to choose any value of $a_{k'}$ without changing the right hand side up to a constant. Therefore 
\begin{equation}
    \ket{\lambda^{\bm{f}}_{\bm{\mu}_{\overline{kN}}|a_{k'}=0,k}} \sim \ket{\lambda^{\bm{f}}_{\bm{\mu}_{\overline{kN}}|a_{k'}=1,k}}.
\end{equation}
There are $2^{N-3}$ such equations written above, corresponding to the different choices of $\bm{\mu}_{\overline{kN}}$ once $a_{k'}$ is fixed. We have $2^{N-2}$ unknowns that we want to fix, hence this is sufficient to show linear dependence for the $N=3$ case. If $N > 3$, we can select another party $k''$ not equal to $k,k',N$ and add another $2^{N-3}$ equations, 
\begin{equation}
    \ket{\lambda^{\bm{f}}_{\bm{\mu}_{\overline{kN}},k}} \sim \ket{\lambda^{\bm{f}}_{\bm{\mu}_{\overline{k''N}},k''}}, 
\end{equation}
allowing us to write
\begin{align}
    \ket{\lambda^{\bm{f}}_{\bm{\mu}_{\overline{kN}}|a_{k'}=0,k}} &\sim \ket{\lambda^{\bm{f}}_{\bm{\mu}_{\overline{kN}}|a_{k'}=1,k}} \nonumber \\
    \ket{\lambda^{\bm{f}}_{\bm{\mu}_{\overline{kN}}|a_{k''}=0,k}} &\sim \ket{\lambda^{\bm{f}}_{\bm{\mu}_{\overline{kN}}|a_{k''}=1,k}}.
\end{align}
The combination of the two will solve the $N=4$ case. We can see that for the first choice of $k'$, we get $2^{N-3}$ equations, which halves the number of unknowns to $2^{N-3}$. Each new choice of $k'$, of which there are $N-2$ in total, halves the number of unknowns, until we are left with one, which shows that every vector can be written as a scalar multiple of a fixed vector, i.e.,
\begin{equation}
    \ket{\lambda^{\bm{f}}_{\bm{\mu}_{\overline{kN}},k}} = \beta_{\bm{\mu}_{\overline{kN}}}^{k} \ket{\lambda^{\bm{f}}_{k}}, \label{eq:lamEq}
\end{equation}
for some constants $\beta_{\bm{\mu}_{\overline{kN}}}^{k}$, which, due to normalization, must be equal to a phase. Substituting the above into \eqref{eq:denstConst2} (noting that any phase factors will cancel when taking the projector), we find $\ket{\lambda^{\bm{f}}_{k}} = \ket{\lambda^{\bm{f}}_{k'}}$, hence we can drop the $k$ label.

Now we finally conclude
\begin{multline}
    \rho_{\bm{R}EE'|\bm{0}} \equiv \rho_{\bm{R}EE'|\bm{0}}^{k} =  \Bigg( \sum_{a_{k},a_{N},\bm{\mu}_{\overline{kN}}} p(a_{k},a_{N},\bm{\mu}_{\overline{kN}}|\bm{0}) \\ \cdot  \ketbra{a_{k}a_{N}\bm{\mu}_{\overline{kN}}}{a_{k}a_{N}\bm{\mu}_{\overline{kN}}}_{\bm{R}} \Bigg) \\ \otimes \Bigg(   \sum_{\bm{f}\in S} p_{\bm{f}} \ketbra{\lambda^{\bm{f}}}{\lambda^{\bm{f}}}_{E} \otimes \ketbra{\bm{f}}{\bm{f}}_{E'} \Bigg),  \label{eq:finalRho2}
\end{multline}
which is of the desired tensor product form, and completes the proof. 
\end{proof}

\subsection{Proof of \cref{lem:rate}}

\noindent \textbf{Corollary 1.} \textit{Let $I,\eta_{N}^{\mathrm{Q}}$ be defined as in \cref{lem:BIent}. Then 
    \begin{equation}
        R_{I}(\eta_{N}^{\mathrm{Q}})
        = H(\{ p(\bm{a}|\bm{0}) \}),
    \end{equation}}
    where $H(\{p_{i}\})$ is the Shannon entropy of a distribution $\{p_{i}\}_{i}$.

\begin{proof}
\cref{lem:rate} directly follows from \cref{lem:BIent} since the state on which the entropy is evaluated has Eve's part in tensor product. This implies, for all compatible post measurement states, $H(\bm{R}|\bm{X}=\bm{0},E)_{\rho_{\bm{R}E|\bm{0}}} = H(\bm{R}|\bm{X}=\bm{0})_{\rho_{\bm{R}|\bm{0}}}$. Since $\rho_{\bm{R}|\bm{0}}$, derived explicitly in the proof of \cref{lem:BIent}, is the same for all compatible states and measurements, we obtain the claim:
\begin{multline}
    \inf_{\substack{\ket{\Psi}_{\tilde{\bm{Q}}E},\big\{\{\tilde{P}_{a_{k}|x_{k}}^{(k)}\}_{a_{k}} \big\}_{k} \\ \mathrm{compatible \ with} \ \langle I \rangle  \\ = \eta^{\mathrm{Q}}_{N} }} H(\bm{R}|\bm{X}=\bm{0},E)_{\rho_{\bm{R}E|\bm{0}}} \\
        = H(\{ p(a_{k},a_{l},\bm{\mu}|\bm{0}) \}).
\end{multline}
\end{proof}

\section{Proofs and details for DI randomness certification} 
Recall the bipartite Bell-inequalities $I_\theta$ which are defined as
\begin{multline}
    \langle I_{\theta}^{(k,l)} \rangle = \cos \theta \cos 2\theta \langle A_{0}^{(k)}A_{0}^{(l)}\rangle - \\ \cos 2\theta \big( \langle A_{0}^{(k)}A_{1}^{(l)} \rangle + \langle A_{1}^{(k)}A_{0}^{(l)}\rangle \big) + \cos \theta \langle A_{1}^{(k)}A_{1}^{(l)} \rangle
\end{multline}
for $\theta \in \mathcal{G}$ where
\begin{equation}
    \mathcal{G} = (\pi/4,\pi/2) \cup (\pi/2,3\pi/4) \cup (5\pi/4,3\pi/2) \cup (3\pi/2,7\pi/4).
\end{equation}
We now prove a few results about these Bell-inequalities that were stated in the main text. 
\subsection{Proof of \cref{lem:NevenBI}}
\label{app:lem5proof}
\noindent \textbf{Lemma 4.} \textit{Let $\bm{\mu}$ be a tuple of $N-2$ measurement outcomes for all parties excluding $k,l$, and $n_{\bm{\mu}} \in \{0,1\}$ be the parity of $\bm{\mu}$. Let $\theta \in \mathcal{G}$ and $\{ I_{\bm{\mu}}^{(k,l)}\}_{\bm{\mu}}$ be a set of bipartite Bell expressions between parties $k,l$, where
\begin{equation}
    I_{\bm{\mu}}^{(k,l)} = (-1)^{n_{\bm{\mu}}}I_{\theta}^{(k,l)} .
\end{equation}
Then the expanded Bell expression given by
\begin{equation}
    I_{\theta} = \sum_{k = 1}^{N-1} \Bigg( \sum_{\bm{\mu}} \tilde{P}^{\overline{(k,N)}}_{\bm{\mu}|\bm{0}}I_{\bm{\mu}}^{(k,N)}\Bigg) \label{eq:NthetaBI_app}
\end{equation}
has quantum bounds $\pm\eta_{N,\theta}^{\mathrm{Q}}$ where $\eta_{N,\theta}^{\mathrm{Q}} = 2(N-1)\sin^{3}\theta$. Moreover, $\langle I_{\theta} \rangle = \eta_{N,\theta}^{\mathrm{Q}}$ is achieved up to relabelings by the strategy in \cref{eq:evenStrat}, and cannot be achieved classically. }

\begin{proof}
The above is an example of an expanded Bell expression, constructed according to \cref{def:expandedBI}, where we used the bipartite expression in \cref{lem:n2} as a seed, and chose $c_{k,l} = 1$ if $l = N$ and $k \in\{1,..,N-1\}$, and 0 otherwise. First we notice that if the bound $\eta_{N,\theta}^{\mathrm{Q}}$ is achievable, then $I_{\theta}$ must define a nontrivial Bell inequality. This is because the maximum local value is upper bounded by $(N-1)\eta_{\theta}^{\mathrm{L}} < \eta_{N,\theta}^{\mathrm{Q}}$, since $\eta_{\theta}^{\mathrm{L}} < \eta_{\theta}^{\mathrm{Q}}$. A similar argument holds for the minimum local value.

We now proceed to show $\langle I \rangle = \eta_{N,\theta}^{\mathrm{Q}}$ is achieved by the strategy in \cref{eq:evenStrat}. Notice that, for the ideal operators $P_{a_{k'}|0}^{(k')} = \ketbra{+}{+}$ ($\ketbra{-}{-}$) if $a_{k'} = 0$ ($a_{k'} = 1$), 
\begin{equation}
    P^{\overline{(k,N)}}_{\bm{\mu}|\bm{0}}\ket{\psi_{\mathrm{GHZ}}} = \sqrt{\bra{\psi_{\mathrm{GHZ}}}P^{\overline{(k,N)}}_{\bm{\mu}|\bm{0}}\ket{\psi_{\mathrm{GHZ}}}}\ket{\phi_{\bm{\mu}}} \otimes \ket{\Phi_{\bm{\mu}}},
\end{equation}
where $\ket{\phi_{\bm{\mu}}}$ spans the support of $P^{\overline{(k,N)}}_{\bm{\mu}|\bm{0}}$, and $\ket{\Phi_{\bm{\mu}}} = (\ket{00}+i(-1)^{n_{\bm{\mu}}}\ket{11})/\sqrt{2}$. We then have 
\begin{align}
    \sum_{\bm{\mu}} &\bra{\psi_{\mathrm{GHZ}}}P^{\overline{(k,N)}}_{\bm{\mu}|\bm{0}} \otimes I_{\bm{\mu}}^{(k,N)}\ket{\psi_{\mathrm{GHZ}}} \nonumber \\ &= \sum_{\bm{\mu}} \bra{\psi_{\mathrm{GHZ}}}P_{\bm{\mu}|\bm{0}}^{\overline{(k,N)}}\ket{\psi_{\mathrm{GHZ}}}  \bra{\Phi_{\bm{\mu}}} I_{\bm{\mu}}^{(k,N)}\ket{\Phi_{\bm{\mu}}}  \nonumber \\ &= 2 \sin^{3}\theta, \ \forall k \in \{1,...,N-1\},
\end{align}
where the final equality comes from the fact that when $n_{\bm{\mu}} = 0$, we recover the state and measurements in \cref{eq:thSt} which yields the maximum quantum value of $I^{(k,N)}_{\theta}$, and when $n_{\bm{\mu}} = 1$, we obtain a strategy equivalent to \cref{eq:thSt} up to a relabelling, which yields the maximum quantum value of $-I^{(k,N)}_{\theta}$. 

Since the maximum quantum bound is achievable, a nontrivial Bell inequality follows. The minimum quantum value $-\eta_{N,\theta}^{\mathrm{Q}}$ is simply achieved by relabelling the outcomes.
\end{proof}

\subsection{Proof of \cref{prop:theta}} \label{app:prop1}
\noindent \textbf{Proposition 3.} \textit{For even $N$, let
    \begin{equation}
        \theta^{*}_{N} = \frac{2\pi t_{N}}{N+1},
    \end{equation}
    where $t_{N}$ is the $(N/2)^{\mathrm{th}}$ element of the sequence $1,1,5,7,3,3,11,13,5,5,...$ given by
    \begin{multline}
        t_{N} = \begin{cases}
            N/4 + 1/2, \  \mathrm{if} \ N = 8n + 2, \\
            N/4,\  \mathrm{if} \ N = 8n + 4, \\
            3N/4 + 1/2,\  \mathrm{if} \ N = 8n + 6, \\
            3N/4 + 1,\  \mathrm{if} \ N = 8n + 8, \ n \in \mathbb{N}_{0}.
        \end{cases} \label{eq:m_app}
    \end{multline}
    Then $\theta^{*}_{N} \in \mathcal{G}$.}

\begin{proof}

The sequence $\theta_N^*$ splits naturally into four subsequences for $n \in \mathbb{N}_{0}$
\begin{equation}
    \begin{aligned}
        \theta_{8n +2}^* &= \frac{2 \pi (2 n +1)}{8 n + 3} \qquad 
        &\theta_{8n+4}^* = \frac{2 \pi (2 n + 1)}{8 n + 5} \\
        \theta_{8n+6}^* &= \frac{2 \pi (6 n + 5)}{8 n + 7} \qquad 
        &\theta_{8n + 8}^* = \frac{2 \pi (6 n + 7)}{8 n + 9}\,.
    \end{aligned}
\end{equation}
First notice that the first terms in each of these subsequences are contained in one of the subintervals of $\mathcal{G}$
\begin{equation}
    \begin{aligned}
        \theta_{2}^* &\in  (\pi/2, 3\pi/4) \qquad 
        &\theta_{4}^* \in (\pi/4, \pi/2) \\
        \theta_{6}^* &\in (5\pi/4, 3 \pi/2) \qquad 
        &\theta_{8}^* \in (3\pi/2, 7\pi/4)\,.
    \end{aligned}
\end{equation}
Looking at the ratio of subsequent terms in each subsequence 
\begin{equation*}
    \begin{aligned}
        \frac{\theta_{8n +2}^*}{\theta_{8n+10}^{*}} &= \frac{16 n^2 + 30 n + 11}{16 n^2 + 30 n + 9} \quad \frac{\theta_{8n +4}^*}{\theta_{8n+12}^*} = \frac{16 n^2 + 34 n + 13}{16 n^2 + 34 n + 15} \\
        \frac{\theta_{8n +6}^*}{\theta_{8n+14}^*} &= \frac{48 n^2 + 130 n + 75}{48 n^2 + 130 n + 77} \quad 
        \frac{\theta_{8n +8}^*}{\theta_{8n+16}^*} = \frac{48 n^2 + 158 n + 119}{48 n^2 + 158 n + 117}\,
    \end{aligned}
\end{equation*}
we see that the subsequences $\theta_{8n +2}^*$ and $\theta_{8n+8}^*$ are strictly monotonically decreasing sequences whereas $\theta_{8n +4}^*$ and $\theta_{8n+6}^*$ are strictly monotonically increasing sequences. Furthermore, their limits are
\begin{equation}
    \begin{aligned}
        \lim_{n\to \infty} \theta_{8n + 2}^* &= \lim_{n\to \infty} \theta_{8n + 4}^* = \pi/2 \\
        \lim_{n\to \infty} \theta_{8n + 6}^* &= \lim_{n\to \infty} \theta_{8n + 8}^* = 3\pi/2\,. \label{eq:asymp}
    \end{aligned}
\end{equation}
The limits correspond to the boundaries of the respective subinterval of $\mathcal{G}$ that contain the first terms of the subsequences. As they converge strictly monotonically to these limits, every term in the subsequences is contained within their respective subintervals and hence $\theta_N^* \in \mathcal{G}$ for all even $N$.
\end{proof}

\subsection{Proof of \cref{prop:even1,prop:odd1}} \label{app:evi2}

\noindent \textbf{Proposition 4.} \textit{Let $N$ be an even integer. For every MABK value $s$ in the range $(1,m_{N}^{*}]$, there exists a $\theta_{s} \in \mathcal{G}$ that satisfies $s = \langle M_{N}(\theta_{s})\rangle$.}
\begin{proof}
    We examine each subsequence independently. 
    
    \vspace{0.2cm}

    \noindent \textbf{Case 1.} For $N = 8n + 2$, recall that $\theta_{N}^{*} \in (\pi/2,3\pi/4)$. When $n=0$, the claim can be seen from \cref{fig:merm1}, which indicates MABK values below $1$ for $\theta\in(\pi/2,3\pi/4)$, and the fact that $M_N(\theta)$ is continuous.  We hence consider $n\geq 1$ in the following. First note
    \begin{equation}
        \langle M_{8n+2}(\pi/2)\rangle = 2^{4n}.
    \end{equation}
    Consider the point $\tilde{\theta}_{n} = \pi/2 - \pi/(4n+1) > \pi/4$. We have
    \begin{multline}
        \langle M_{8n+2}(\tilde{\theta}_{n})\rangle = -2^{4n} \Big( \sin^{8n+2}\Big( \frac{\pi}{8n+2}\Big) \\+ \cos^{8n+2}\Big( \frac{\pi}{8n+2}\Big) \Big) < 0.
    \end{multline}
    Therefore, by continuity of the function $\langle M_{N}(\theta)\rangle$ in $\theta$, the range of achievable MABK values corresponding to $\theta \in (\tilde{\theta}_{n},\pi/2) \subset \mathcal{G}$ contains the sub-interval $(1,2^{4n})$. By a similar argument, the sub-interval $(2^{4n},m_{8n+2}^{*})$ must also be achievable for $\theta \in (\pi/2,\theta_{8n+2}^{*}) \subset \mathcal{G}$. For the value $2^{4n}$, notice
    \begin{equation}
        \langle M_{8n+2}(3\pi/4)\rangle = 2^{(8n+1)/2}(-1)^{n}\sin^{8n+2}(3\pi/8),
    \end{equation}
    and $\sin^{8n+2}(3\pi/8) < 1/\sqrt{2}$. Hence there must exist a $\theta \in (\theta_{8n+2}^{*},3\pi/4) \subset \mathcal{G}$ such that $\langle M_{8n+2}(\theta)\rangle = 2^{4n}$.

    \vspace{0.2cm}

    \noindent \textbf{Case 2.} For $N = 8n + 4$, recall $\theta_{N}^{*} \in (\pi/4,\pi/2)$. As before, we only need to consider $n \geq 1$ by analyzing \cref{fig:merm1}. Let $\tilde{\theta}_{n} = \pi/2 + \pi/(4n+2) < 3\pi/4$. Then
    \begin{multline}
        \langle M_{8n+4}(\tilde{\theta}_{n})\rangle = 2^{4n+1} \Big( \sin^{8n+4}\Big( \frac{\pi}{8n+4}\Big) \\- \cos^{8n+4}\Big( \frac{\pi}{8n+4}\Big) \Big) < 0.
    \end{multline}
    We also note
    \begin{equation}
        \langle M_{8n+4}(\pi/2)\rangle = 2^{4n+1},
    \end{equation}
    and
    \begin{multline}
        \langle M_{8n+4}(\pi/4)\rangle = 2^{4n+1}(-1)^{n}\big(\cos^{8n+4}(3\pi/8)\\ +\sin^{8n+4}(3\pi/8)\big) < 2^{4n+1}. 
    \end{multline}
    We can therefore apply the same arguments as in the previous case. 

    \vspace{0.2cm}

    \noindent \textbf{Case 3.} For $N = 8n + 6$, recall $\theta_{N}^{*} \in (5\pi/4,3\pi/2)$. We observe that the claim holds for $n = 0$ by examining \cref{fig:merm1}, and consider $n\geq 1$. Let $\tilde{\theta}_{n} = 3\pi/2 + \pi/(4n+3)$, and note
    \begin{multline}
        \langle M_{8n+6}(\tilde{\theta}_{n})\rangle = -2^{4n+2} \Big( \sin^{8n+6}\Big( \frac{\pi}{8n+6}\Big) \\+ \cos^{8n+6}\Big( \frac{\pi}{8n+6}\Big) \Big) < 0.
    \end{multline}
    By checking
    \begin{equation}
        \langle M_{8n+6}(3\pi/2)\rangle = 2^{4n+2},
    \end{equation}
    and
    \begin{multline}
        \langle M_{8n+6}(5\pi/4)\rangle \\ = -2^{(8n+5)/2}(-1)^{n}\sin^{8n+6}(\pi/8) < 2^{4n+2},
    \end{multline}
    we can apply the same arguments as case 1. 

    \vspace{0.2cm}

    \noindent \textbf{Case 4.} Finally, for $N = 8n + 8$, we have $\theta_{N}^{*} \in (3\pi/2,7\pi/4)$. The case $n=0$ can be seen by examining \cref{fig:merm1}, and for $n \geq 1$ we define $\tilde{\theta}_{n} = 3\pi/2 - \pi/(4n+4) > 5\pi/4$. Then
    \begin{multline}
        \langle M_{8n+8}(\tilde{\theta}_{n})\rangle = 2^{4n+3} \Big( \sin^{8n+8}\Big( \frac{\pi}{8n+8}\Big) \\- \cos^{8n+8}\Big( \frac{\pi}{8n+8}\Big) \Big) < 0.
    \end{multline}
    We also have
    \begin{equation}
        \langle M_{8n+8}(3\pi/2) \rangle = 2^{4n+3},
    \end{equation}
    and
    \begin{multline}
        \langle M_{8n+8}(7\pi/4) \rangle = 2^{4n+3}(-1)^{n+1}\big( \cos^{8n+8}(\pi/8) \\ - \sin^{8n+8}(\pi/8)\big) < 2^{4n+3}.
    \end{multline}
    By the arguments in case 1, this completes the proof. 
\end{proof}

\noindent \textbf{Proposition 5.} \textit{Let $N$ be an odd integer. For every MABK value $s$ in the range $(1,2^{(N-1)/2})$, there exists a $\theta_{s} \in \mathcal{G}$ that satisfies $s = \langle M_{N}(\theta_{s})\rangle$.}
\begin{proof}
    As was done for the even case, we consider four subsequences.

    \vspace{0.2cm}

    \noindent \textbf{Case 1.} For $N = 8n + 3$, we note $\langle M_{8n+3}(\pi/2)\rangle = 2^{4n+1}$. When $n=0$, the claim can be verified from \cref{fig:merm2}, hence we consider $n \geq 1$. Consider the point $\tilde{\theta}_{n} = \pi/2 + \pi/(4n+2) \in (\pi/2,3\pi/4) \subset \mathcal{G}$. Then 
    \begin{multline}
        \langle M_{N}(\tilde{\theta}_{n})\rangle = -2^{4n+1}\Big(\sin^{8n+3}\Big( \frac{\pi}{8n+4}\Big) \sin\big(\pi f(n) \big) \\ + \cos^{8n+3}\Big(\frac{\pi}{8n+4}\Big) \cos\big( \pi f(n) \big)\Big),
    \end{multline}
    where $f(n) = \frac{16n^{2} + 24n + 7}{8n+4}$. Notice that
    \begin{multline}
        \sin\big(\pi f(n) \big) + \sin\Big (\frac{\pi}{8n+4}\Big) \\ = 2\sin\big(\pi(n+1)\big)\cos\Big( \frac{\pi}{2}\frac{8n^{2}+12n + 3}{4n+2}\Big) = 0.
    \end{multline}
    Similarly 
    \begin{multline}
        \cos\big(\pi f(n) \big) - \cos\Big (\frac{\pi}{8n+4}\Big) \\ = -2\sin\big(\pi(n+1)\big)\sin\Big( \frac{\pi}{2}\frac{8n^{2}+12n + 3}{4n+2}\Big) = 0.
    \end{multline}
    Hence 
    \begin{multline}
        \langle M_{N}(\tilde{\theta}_{n})\rangle = 2^{4n+1}\Big( \sin^{8n+4}\Big( \frac{\pi}{8n+4}\Big) \\ - \cos^{8n+4}\Big(\frac{\pi}{8n+4}\Big)\Big)<0.
    \end{multline}
    Therefore, the achievable MABK values from the interval $(\pi/2,\tilde{\theta}_{n}) \subset \mathcal{G}$ include the sub-interval $(2^{4n+1},1)$ by continuity of the function $\langle M_{N}(\theta)\rangle$ in $\theta$. 

    \vspace{0.2cm}

    \noindent \textbf{Case 2.} For $N = 8n + 7$, note $\langle M_{8n+7}(3\pi/2)\rangle = 2^{4n+3}$. The claim for $n=0$ can be verified using \cref{fig:merm2}, so we consider $n \geq 1$. Let $\tilde{\theta}_{n} = 3\pi/2 + \pi/(4n+4)< 7\pi/4$, and note
    \begin{multline}
        \langle M_{N}(\tilde{\theta}_{n})\rangle = 2^{4n+3}\Big(\sin^{8n+7}\Big( \frac{\pi}{8n+8}\Big) \cos\big(\pi f(n) \big) \\ - \cos^{8n+7}\Big(\frac{\pi}{8n+8}\Big) \sin\big( \pi f(n) \big)\Big),
    \end{multline}
    where $f(n) = \frac{48 n^{2} + 100 n + 51}{8n+8}$. Following the same steps as before, we have
    \begin{equation}
        \sin\big(\pi f(n) \big) - \cos\Big (\frac{\pi}{8n+8}\Big) = 0,
    \end{equation}
    and
    \begin{equation}
        \cos\big(\pi f(n) \big) - \sin\Big (\frac{\pi}{8n+8}\Big) = 0,
    \end{equation}
    implying 
    \begin{multline}
        \langle M_{N}(\tilde{\theta}_{n})\rangle = 2^{4n+3}\Big(\sin^{8n+8}\Big( \frac{\pi}{8n+8}\Big)  \\ - \cos^{8n+8}\Big(\frac{\pi}{8n+8}\Big)\Big)< 0.
    \end{multline}
    The same arguments used in case 1 prove the claim. 

    \vspace{0.2cm}

    \noindent \textbf{Case 3.} For $N = 8n + 5$, we consider the MABK expression obtained by relabelling the inputs of every party, followed by the output of every party's first measurement, i.e., $\tilde{A}_{0}^{(k)} \mapsto \tilde{A}_{1}^{(k)}$ and $\tilde{A}_{1}^{(k)} \mapsto -\tilde{A}_{0}^{(k)}$. The resulting MABK value of the strategy in \cref{eq:evenStrat} is given by\footnote{Rather than relabelling the MABK expression, we could equivalently consider the strategy in \cref{eq:evenStrat_2}, which achieves the MABK value in \cref{eq:MABKval2} of the original expression (\cref{eq:MABK}), with the substitution $\phi \mapsto \theta, \ \theta \mapsto 0$. This also results in \cref{eq:MABK3}.}  
    \begin{multline}
        \langle \tilde{M}_{N}(\theta) \rangle = 2^{\frac{N-1}{2}}\Big(  \cos^{N}\big(\theta/2 + \pi/4\big) \sin \big(-N\theta/2 + \pi/4 \big) \\ - \cos^{N}\big(\theta/2 - \pi/4\big) \sin \big(N\theta/2 + \pi/4 \big) \Big).\label{eq:MABK3}
    \end{multline}
    We then have $\langle \tilde{M}_{8n+5}(\pi/2) \rangle = 2^{4n+2}$. We can verify the claim for $n=0$ from \cref{fig:merm2}, hence we consider $n \geq 1$. Let $\tilde{\theta} = \pi/2+\pi/(4n+3) < 3\pi/4$, and note
    \begin{multline}
        \langle \tilde{M}_{N}(\tilde{\theta}_{n})\rangle = -2^{4n+2}\Big(\sin^{8n+5}\Big( \frac{\pi}{8n+6}\Big) \cos\big(\pi f(n) \big) \\ +8 \cos^{8n+5}\Big(\frac{\pi}{8n+6}\Big) \sin\big( \pi f(n) \big)\Big),
    \end{multline}
    where $f(n) = \frac{8n^{2} + 16n + 7}{4n + 3}$. We have
    \begin{equation}
        \sin\Big( \frac{\pi}{8n+6} \Big) - \cos(\pi f(n)) = 0,
    \end{equation}
    and
    \begin{equation}
        \cos\Big( \frac{\pi}{8n+6} \Big) - \sin(\pi f(n)) = 0,
    \end{equation}
    which implies 
    \begin{multline}
        \langle \tilde{M}_{N}(\tilde{\theta}_{n})\rangle = -2^{4n+2}\Big(\sin^{8n+6}\Big( \frac{\pi}{8n+6}\Big) \\ + \cos^{8n+6}\Big(\frac{\pi}{8n+6}\Big)\Big) <0.
    \end{multline}
    The claim then follows from the same arguments as the previous cases. 

    \vspace{0.2cm}

    \noindent \textbf{Case 4.} For $N = 8n + 9$, we also consider the relabelled MABK expression in \cref{eq:MABK3}, and note $\langle \tilde{M}_{8n+5}(3\pi/2) \rangle = 2^{4n+4}$. Let $\tilde{\theta}_{n} = 3\pi/2 + \pi/(4n + 5) < 7\pi/4$, and consider
    \begin{multline}
        \langle \tilde{M}_{N}(\tilde{\theta}_{n})\rangle = 2^{4n+4}\Big(\sin^{8n+9}\Big( \frac{\pi}{8n+10}\Big) \sin\big(\pi f(n) \big) \\ - \cos^{8n+9}\Big(\frac{\pi}{8n+10}\Big) \cos\big( \pi f(n) \big)\Big),
    \end{multline}
    where $f(n) = \frac{48 n^{2} + 124 n + 79}{8n + 10}$. We have
    \begin{equation}
        \cos \Big( \frac{\pi}{8n + 10} \Big) - \cos(\pi f(n)) = 0,
    \end{equation}
    and
    \begin{equation}
        \sin \Big( \frac{\pi}{8n + 10} \Big) + \sin(\pi f(n)) = 0,
    \end{equation}
    implying 
    \begin{multline}
        \langle \tilde{M}_{N}(\tilde{\theta}_{n})\rangle = -2^{4n+4}\Big(\sin^{8n+10}\Big( \frac{\pi}{8n+10}\Big) \\ + \cos^{8n+10}\Big(\frac{\pi}{8n+10}\Big)\Big) < 0.
    \end{multline}
    Following the same steps as before establishes the claim. 
\end{proof}

\subsection{Evidence for \cref{conj:maxMerm}} \label{app:evC1}

\noindent \textbf{Conjecture 1.} \textit{The maximum MABK value achievable by quantum strategies that generate maximum randomness, i.e., strategies with $p(\bm{a}|\bm{0}) = 1/2^{N}, \ \forall \bm{a}$, is $m_{N}^{*}$.}

\vspace{0.2cm}
    
\noindent \cref{conj:maxMerm} is that the analytical lower bound, $m_{N}^{*}$, on the maximum MABK value achievable whilst generating maximum randomness, is optimal, in the sense there no other set of quantum correlations exists with a higher MABK value while achieving maximum randomness. To justify this claim, we developed a numerical technique to upper bound the maximum MABK value of quantum correlations achieving maximum randomness, and checked this against the achievable lower bounds $m_{N}^{*}$. We found close agreement between our numerical calculations and the analytic lower bounds, modulo a small discrepancy which pushed the numerical values slightly below the analytical ones. Since the SDP being solved is a minimization, we judge this to be caused by numerical precision, and an accumulation of errors pushing the global optimum below its true value. A comparison between our analytical lower bound and the numerics can be found in \cref{tab:evidence}, and details of the numerical method can be found in \cref{app:numMethod}. We also remark that this conjecture was proven analytically in the $N=2$ case in Ref.~\cite{WBC}. 

\begin{table*}[t]
\centering
\begin{tabular}{|c| c | c | c|} 
 \hline
 $N$ & Analytical lower bound, $m_{N}^{*}$ & Numerical upper bound  \\ [0.5ex] 
 \hline\hline
 2 & $3\sqrt{3}/4 \approx 1.299038105676658$ & $1.299038105676658$  \\ 
 \hline
 4 & $(5/8)\sqrt{(5/2)(5+\sqrt{5})} \approx 2.658283776100125$ & 2.658283776100125  \\
 \hline
 6 & 
 \begin{tabular}{@{}c@{}}$(7/8)\big(\cos(\pi/14)+3\sqrt{2}\cos(3\pi/28)-3\cos(\pi/7)$ \\ $+5\cos(3\pi/14)\big) \approx 5.412519474140316$\end{tabular}
& 5.412519474140316  \\
 \hline
 8 & 
 \begin{tabular}{@{}c@{}}$45/16+33\sqrt{3}/4 + (45\sqrt{2}/16)(3+\sqrt{3})\cos(\pi/36)$ \\ $- (1/16)(180\sqrt{3}+90)\cos^{2}(\pi/36) \approx 10.93208547970192$\end{tabular}
  & 10.93208547970192  \\
 \hline
 10 & 
 \begin{tabular}{@{}c@{}}$(11/32)\big( \cos(\pi/22)+30\sqrt{2}\cos(3\pi/44)-5\cos(\pi/11)+15\cos(3\pi/22)$ \\ $+5\sqrt{2}\cos(7\pi/44)-30\cos(2\pi/11)+42\cos(5\pi/22) \big) \approx 22.00126183851497$\end{tabular}
  & 22.00126183851497  \\ 
 \hline
 12 & 
 \begin{tabular}{@{}c@{}}$(429\sqrt{2}/16)\cos(\pi/52) + (39\sqrt{2}/32)\cos(9\pi/52)+(715\sqrt{2}/64)\cos(5\pi/52)$ \\ $+(13/64)\cos(\pi/26) + (1287/64)\cos(5\pi/26) -(715/64)\cos(2\pi/13)$ \\$+(143/32)\cos(3\pi/26) -(39/32)\cos(\pi/13)-(429/16)\cos(3\pi/13)\approx 44.19316043036512$\end{tabular}
  & 44.19316043036512\\ 
 \hline
\end{tabular}
\caption{Comparison between upper and lower bounds on the maximum MABK value which can be achieved by quantum correlations certifying maximum randomness, for different numbers of parties $N$, rounded to 16 digits. The analytical lower bound is provided by the family of strategies \cref{eq:evenStrat} in the main text, and the numerical upper bound is calculated with an SOS approach using the SDPA-GMP solver. The numerical upper bounds agree exactly with the analytical lower bounds up to 16 digits of precision, supporting the conjecture that our analytical lower bound is tight. }
\label{tab:evidence}
\end{table*}

\subsection{Proof of \cref{thm:asymp}} \label{app:asympP}
\noindent \textbf{Proposition 6.} \textit{In the limit of large even $N$, one can achieve arbitrarily close to the maximum quantum violation of the $N$ party MABK inequality, $2^{(N-1)/2}$, whilst certifying maximum device-independent randomness.}

\begin{proof}
We consider the asymptotic behaviour of $m_{N}^{*}$, which gives an achievable lower bound on the maximum MABK value compatible with maximum randomness. Our aim is to show that as $N$ becomes very large, whilst remaining an even integer, the MABK violation tends to the maximum quantum violation $2^{(N-1)/2}$, i.e.,
\begin{equation}\label{eq:lim}
    \lim_{\substack{N \to \infty \\ N \ \mathrm{even}}} \frac{m_{N}^{*} }{2^{\frac{N-1}{2}}} =  1.
\end{equation}
This can be established directly from the asymptotic behaviour of $\theta^{*}_{N}$ given in \cref{eq:asymp}, as follows.

We first insert the expression for $m_{N}^{*} = \langle M_{N}(\theta^{*}_{N}) \rangle$ in \cref{eq:mermSth} into the left hand side of~\eqref{eq:lim} to give
\begin{align}\label{lim:eq}
    \lim_{\substack{N \to \infty }} & \Big( \cos^{N}\big(\theta^{*}_{N}/2 + \pi/4\big) \sin \big(N\theta^{*}_{N}/2 + \pi/4 \big) \nonumber \\ &+ \cos^{N}\big(\theta^{*}_{N}/2 - \pi/4\big) \sin \big(N\theta^{*}_{N}/2 - \pi/4 \big) \Big).
\end{align} 

We then consider two cases:

\vspace{0.2cm}

\noindent \textbf{Case 1: $N= 8n +2$ or $N= 8n +4$.} For these values of $N$, we know $\lim_{N \to \infty } \theta_{N}^{*} = \pi/2$. One can then verify 
\begin{equation}
\begin{aligned}
    \lim_{\substack{N \to \infty }}  \cos^{N}\big(\theta_{N}^{*}/2 + \pi/4 \big) = 0, \\
    \lim_{\substack{N \to \infty }}  \cos^{N}\big(\theta_{N}^{*}/2 - \pi/4 \big) = 1.
\end{aligned}
\end{equation}
Now we split into two sub-cases to evaluate $N\theta^{*}_{N}$ in the limit~\eqref{lim:eq}.

\noindent \textbf{1a: $N=8n+2$.} Here we find
\begin{equation}
    \frac{N\theta^{*}_{N}}{2} = \frac{\pi (8n+2)(2n+1)}{8n+3},
\end{equation}
which has the oblique asymptote $2\pi n+3\pi/4$. We hence find
\begin{equation}
    \lim_{\substack{N \to\infty }} \sin \big(N\theta^{*}_{N}/2 - \pi/4\big) = \sin(\pi/2) = 1. \label{eq:case1a}
\end{equation}
\noindent \textbf{1b: $N=8n+4$.} Here we find
\begin{equation}
    \frac{N\theta^{*}_{N}}{2} = \frac{4\pi(2n+1)^2}{8n+5},
\end{equation}
which also has the oblique asymptote $2\pi n + 3\pi/4$, implying \cref{eq:case1a}.

\vspace{0.2cm}

\noindent \textbf{Case 2: $N= 8n +6$ or $N= 8n +8$.} For these values of $N$, we know $\lim_{N \to\infty } \theta_{N}^{*} = 3\pi/2$. One can then verify 
\begin{equation}
\begin{aligned}
    \lim_{\substack{N \to \infty }}  \cos^{N}\big(\theta_{N}^{*}/2 + \pi/4 \big) = 1, \\
    \lim_{\substack{N \to \infty }}  \cos^{N}\big(\theta_{N}^{*}/2 - \pi/4 \big) = 0.
\end{aligned}
\end{equation}
Splitting into two sub-cases:

\noindent \textbf{2a, $N=8n+6$:} Here 
\begin{equation}
     \frac{N\theta^{*}_{N}}{2} = \frac{\pi(8n+6)(6n+5)}{8n+7},
\end{equation}
which has the oblique asymptote $6\pi n + 17 \pi / 4$. Inserting into the limit we find
\begin{equation}
    \lim_{\substack{N \to\infty }} \sin \big(N\theta^{*}_{N}/2 + \pi/4\big) = \sin(9 \pi/2) = 1.
\end{equation}

\noindent \textbf{2b, $N=8n+8$:} For the final case
\begin{equation}
     \frac{N\theta^{*}_{N}}{2} = \frac{\pi (8n+8)(6n+7)}{8n + 9},
\end{equation}
which has the oblique asymptote $6\pi n + 25\pi/4$. Inserting into the limit we find
\begin{equation}
    \lim_{\substack{N \to\infty }} \sin \big(N\theta^{*}_{N}/2 + \pi/4\big) = \sin(13\pi /2) = 1.
\end{equation}

This shows that the MABK violation $m_{N}^{*}$ tends to the maximum quantum value in the regime of large $N$. Note however, in the limit $\theta$ tends to $\pi/2$ or $3\pi/2$; these values are not self-testable, since $\theta$ does not lie in a valid interval for there to be a gap between the local and quantum bound for the two party Bell inequality. However, the points either side of both angles are self-testable. We hence obtain the result that one can achieve arbitrarily close to maximum MABK violation whilst certifying maximum DI randomness, for an arbitrarily large finite $N$. 
\end{proof}

\subsection{Local dilution measure of nonlocality} \label{app:LP}
In this section we briefly outline the definition of local dilution (see also~\cite[Section I.B.1]{tavakoli2024} for a description of this technique) which was plotted in \cref{fig:NLvsTH}. Given a distribution $\{p(\bm{a}|\bm{x})\}$, we can represent it as a vector in a high dimensional space, $p \in \mathbb{R}^{2^{2N}}$. We denote the set of $2^{2N}$ local deterministic distributions as $\{d_{\lambda}\}_{\lambda}$, and the local polytope can be constructed by taking its convex hull, 
\begin{equation}
    \mathcal{L} = \Big\{ \sum_{\lambda}q_{\lambda} d_{\lambda} \ | \ \sum_{\lambda}q_{\lambda} = 1, \ q_{\lambda} \geq 0 \Big\}.
\end{equation}
To check whether $p$ admits a local description, we can search for a set $\{q_{\lambda}\}_{\lambda}$, where $\sum_{\lambda} q_{\lambda} =1$, $q_{\lambda} \geq 0$, using a standard linear program. 

To quantify the distance of a given $p$ from $\mathcal{L}$, we define the diluted sets, for $\epsilon \geq 0$,
\begin{equation}
    \mathcal{L}_{\epsilon} = \Big\{ \sum_{\lambda}q_{\lambda} d_{\lambda} \ | \ \sum_{\lambda}q_{\lambda} = 1, \ q_{\lambda} \geq -\epsilon \Big\}.
\end{equation}
As $\epsilon \rightarrow \infty$, $\mathcal{L}_{\epsilon}$ tends to the affine hull of the points $\{d_{\lambda}\}_{\lambda}$, and we have the inclusion:  $\mathcal{L} \subset \mathcal{L}_{\epsilon_{1}} \subset \mathcal{L}_{\epsilon_{2}} \subset \mathcal{L}_{\infty}$, where $0 < \epsilon_{1} < \epsilon_{2}$. We then define the following measure, which we call the local dilution:
\begin{equation}
    \mathcal{D}[p] := \min \{ \epsilon \ | \ p \in \mathcal{L}_{\epsilon} \}.
\end{equation}
$\mathcal{D}[p]$ can be interpreted as the smallest dilution of the local polytope that contains $p$. If $\mathcal{D}[p_{1}] < \mathcal{D}[p_{2}]$, one has to enlarge the local set more to include $p_{2}$ than to include $p_{1}$. We can hence understand $p_{2}$ to be more nonlocal than $p_{1}$. 

\subsection{Proof of \cref{lem:n2_2}} \label{app:selfTest}
In the main text, we detailed a two parameter family of bipartite Bell expressions which are used to build our $N$-partite expressions. These are generalizations of the two families introduced in \cite{WBC}, and can be recovered as sub-families of the self-tests studied in~\cite{barizien2024,WBC3,Le2023}. 

\vspace{0.2cm}

\noindent \textbf{Lemma 7.} \textit{Let $(\phi,\theta) \in \mathbb{R}^{2}$ such that $\cos(2\theta)\cos(2\phi)<0$ and $\cos(\theta-\phi)\neq0$. Define the family of Bell expressions parameterized by $\phi$ and $\theta$,
\begin{multline}
    \langle J_{\phi,\theta}^{(k,l)} \rangle = \cos 2\theta \cos (\theta - \phi) \langle A_{0}^{(k)}A_{0}^{(l)}\rangle \\- \cos 2\theta \cos 2 \phi \big( \langle A_{0}^{(k)}A_{1}^{(l)} \rangle + \langle A_{1}^{(k)}A_{0}^{(l)} \rangle \big) \\+ \cos 2\phi \cos (\theta - \phi) \langle A_{1}^{(k)}A_{1}^{(l)} \rangle.
\end{multline}
Then we have the following:
\begin{enumerate}[(i)]
    \item The local bounds are given by $\pm\eta_{\phi,\theta}^{\mathrm{L}}$, where 
    \begin{multline}
        \eta_{\phi,\theta}^{\mathrm{L}} = \max \big\{ |\cos(\theta - \phi)(\cos 2\theta  - \cos 2\phi) | , \\ | \cos(\theta - \phi)\big(\cos 2\theta + \cos 2\phi \big) | + |2\cos 2\phi \cos 2 \theta | \big \}.
    \end{multline}
    \item The quantum bounds are given by $\pm\eta_{\phi,\theta}^{\mathrm{Q}}$, where $\eta_{\phi,\theta}^{\mathrm{Q}} = 2\sin^{2}(\theta+\phi)\sin (\theta - \phi)$.
    \item $|\eta^{\mathrm{Q}}_{\phi,\theta}| > \eta^{\mathrm{L}}_{\phi,\theta}$.
    \item Up to local isometries, there exists a unique strategy that achieves $\langle J_{\phi,\theta}^{(k,l)} \rangle = \eta_{\phi,\theta}^{\mathrm{Q}}$:
    \begin{equation}
    \begin{aligned}
        \rho_{Q_{k}Q_{l}} &= \ketbra{\psi}{\psi}, \ \mathrm{where} \ \ket{\psi} = \frac{1}{\sqrt{2}}\big( \ket{00} + i\ket{11}\big),  \\
        A_{0}^{(k)} &= A_{0}^{(l)} = \cos \phi \, \sigma_{X} - \sin \phi \, \sigma_{Y},  \\
        A_{1}^{(k)} &= A_{1}^{(l)} = \cos \theta\, \sigma_{X} + \sin \theta \,\sigma_{Y}. 
    \end{aligned}\label{eq:thPhiSt_app}
    \end{equation}
\end{enumerate}    }

\begin{proof}
    The above is a direct corollary of~\cite[Proposition 1]{WBC3}. Specifically, let $\beta = -2\phi$, $\gamma = \theta - \phi$ and $\alpha = \theta + \phi$. Then the Bell expression $B_{\alpha,\beta,\gamma}$ of~\cite[Proposition 1]{WBC3} is equal to the expression $J_{\phi,\theta}$ up to a factor of $\cos(\theta-\phi)$. The condition for $\eta^{\mathrm{L}}_{\phi,\theta} < |\eta^{\mathrm{Q}}_{\phi,\theta}|$ is given by~\cite[Equation 4]{WBC3}, which reads
    \begin{multline}
        \cos(\alpha + \beta) \cos  \beta \cos(\alpha + \gamma) \cos  \gamma  = \\
        \cos^{2}(\theta - \phi) \cos  2\theta \cos  2 \phi < 0,
    \end{multline}
    and is satisfied if and only if $\cos(2\theta)\cos(2\phi)<0$ and $\cos(\theta-\phi)\neq0$. 
\end{proof}

\cref{lem:n2} can be proven as a special case of the above, by setting $\phi=0$. This follows from $\mathcal{G} \subset \mathcal{F}$, where $\mathcal{F}$ is a set of valid self-testing points defined below \cref{lem:n2_2}.

One can also recover the $I_{\delta}$ and $J_{\gamma}$ family of Bell expressions from \cite{WBC}, as special cases of the $J_{\phi,\theta}$ expressions. By choosing $\phi = 0$, and $\theta = \delta + \pi/2$, $\delta \in (0,\pi/6]$, we obtain
\begin{multline}
    \langle J_{0,\delta + \pi/2} \rangle = \sin(\delta)\cos(2\delta)\langle A_{0}^{(k)}A_{0}^{(l)} \rangle \\ + \cos (2\delta) \big( \langle A_{0}^{(k)}A_{1}^{(l)} \rangle + \langle A_{1}^{(k)}A_{0}^{(l)} \rangle \big) \\ - \sin (\delta) \langle A_{1}^{(k)}A_{1}^{(l)} \rangle \leq 2\cos^{3}(\delta),
\end{multline}
which is identical to the expression found in \cite{WBC}. One can obtain the $J_{\gamma}$ family by choosing $\phi = -3\gamma/2$ and $\theta = \phi/3 + 2\pi/3$ using \cref{eq:thPhi}, and multiplying by $4\cos^{2}(\gamma + \pi/6)/\sin^{2}(2\gamma + \pi/3)$:
\begin{multline}
    \frac{4\cos^{2}(\gamma + \pi/6)}{\sin^{2}(2\gamma + \pi/3)}\langle J_{-3\gamma/2,-\gamma/2 + 2\pi/3} \rangle = \langle A_{0}^{(k)}A_{0}^{(l)} \rangle \\ + \Big( 4\cos^{2}(\gamma + \pi/6) - 1\Big) \Big( \langle A_{0}^{(k)}A_{1}^{(l)} \rangle + \langle A_{1}^{(k)}A_{0}^{(l)} \rangle \\ - \langle A_{1}^{(k)}A_{1}^{(l)} \rangle \Big) \leq 8\cos^{3}(\gamma + \pi/6). 
\end{multline}

\subsection{Proof of \cref{lem:NevenBI_2}} 
\label{app:lem8proof}
\noindent \textbf{Lemma 8.} \textit{Let $\bm{\mu}$ be a tuple of $N-2$ measurement outcomes for all parties excluding $k,l$, and $n_{\bm{\mu}} \in \{0,1\}$ be the parity of $\bm{\mu}$. Let $(\phi,\theta) \in \mathbb{R}^{2}$,
\begin{equation}
    \phi':= \frac{\phi N}{2}, \ \theta' := \theta - \frac{N-2}{2}\phi, 
\end{equation}and 
\begin{equation}
    I_{\bm{\mu}}^{(k,l)} = (-1)^{n_{\bm{\mu}}} J_{\phi',\theta'}^{(k,l)}.
\end{equation}
Define the following Bell polynomial
\begin{equation}
    I_{\phi',\theta'} := \sum_{k = 1}^{N-1} \Bigg( \sum_{\bm{\mu}} \tilde{P}^{\overline{(k,N)}}_{\bm{\mu}|\bm{0}}I_{\bm{\mu}}^{(k,N)}\Bigg).  \label{eq:genGam_app}
\end{equation}
If $(\phi',\theta') \in \mathcal{F}$, $I_{\phi',\theta'}$ is an expanded Bell expression, and has quantum bounds $\pm\eta_{N,\phi',\theta'}^{\mathrm{Q}}$ where $\eta_{N,\phi',\theta'}^{\mathrm{Q}} = 2(N-1)\sin^{2}(\theta'+\phi')\sin (\theta'-\phi')$, which can be achieved up to relabellings by the strategy
\begin{equation}
\begin{gathered}
    \rho_{\bm{Q}} = \ketbra{\psi_{\mathrm{GHZ}}}{\psi_{\mathrm{GHZ}}},  \\
    A_{0}^{(k)} = \cos \phi \, \sigma_{X} - \sin \phi \, \sigma_{Y},  \\
    A_{1}^{(k)} = \cos \theta \, \sigma_{X} + \sin \theta \, \sigma_{Y}, \ k \in \{1,...,N\}.
\end{gathered} 
\end{equation}
In addition, this quantum bound cannot be achieved classically. }
\begin{proof}
Suppose $(\phi',\theta') \in \mathcal{F}$, and consider the ideal projection operators $P_{a_{k'}|0}^{(k')} = \ketbra{\varphi_{a_{k'}}}{\varphi_{a_{k'}}}$ derived from the observables in \cref{eq:evenStrat_2}, i.e., $\cos \phi \, \sigma_{X} - \sin \phi \, \sigma_{Y} = \ketbra{\varphi_{0}}{\varphi_{0}} - \ketbra{\varphi_{1}}{\varphi_{1}}$ where
\begin{equation}
    \ket{\varphi_{a_{k'}}} = \frac{\ket{0} + (-1)^{a_{k'}}e^{-i\phi}\ket{1}}{\sqrt{2}}. 
\end{equation}
We then have
\begin{multline}
    P_{\bm{\mu}|\bm{0}}^{\overline{(k,N)}} \ket{\psi_{\mathrm{GHZ}}} = \\ \sqrt{\bra{\psi_{\mathrm{GHZ}}}P_{\bm{\mu}|\bm{0}}^{\overline{(k,N)}}\ket{\psi_{\mathrm{GHZ}}}} \ket{\varphi_{\bm{\mu}}} \otimes \ket{\Phi_{\bm{\mu},\phi}}_{Q_{k}Q_{N}}, 
\end{multline}
where $\ket{\varphi_{\bm{\mu}}}$ spans the support of $P_{\bm{\mu}|\bm{0}}^{\overline{(k,N)}}$ and 
\begin{equation}
    \ket{\Phi_{\bm{\mu},\phi}} = \frac{\ket{00} + i(-1)^{n_{\bm{\mu}}} e^{i(N-2)\phi} \ket{11}}{\sqrt{2}}.
\end{equation}
Notice $\ket{\Phi_{\bm{\mu},\phi}} = (U_{\phi}\otimes U_{\phi})\ket{\Phi_{\bm{\mu}}}$ where $U_{\phi} = \ketbra{0}{0} + e^{i(N-2)\phi/2}\ketbra{1}{1}$. We now calculate 
\begin{multline}
     \bra{\Phi_{\bm{\mu},\phi}} I_{\bm{\mu}}^{(k,l)}\ket{\Phi_{\bm{\mu},\phi}} \\  
     = (-1)^{n_{\bm{\mu}}}\bra{\Phi_{\bm{\mu}}} (U_{\phi}\otimes U_{\phi})^{\dagger} J_{\phi',\theta'}^{(k,N)} (U_{\phi}\otimes U_{\phi})\ket{\Phi_{\bm{\mu}}}.   
\end{multline}
We can now calculate how $I_{\bm{\mu}}^{(k,N)}$ transforms under $(U_{\phi} \otimes U_{\phi})^{\dagger}$. Let $O=\cos(t)\,\sigma_{X} + \sin(t)\, \sigma_{Y}$ for $t \in \mathbb{R}$ be an arbitrary observable. Then 
\begin{equation}
    U_{\phi}^{\dagger}O U_{\phi} = \cos\Big( -\frac{N-2}{2}\phi + t\Big) \, \sigma_{X} + \sin\Big( -\frac{N-2}{2}\phi + t\Big) \, \sigma_{Y}.
\end{equation}
Substituting $t = -\phi$ and $t = \theta$ into the above, we find
\begin{align}
    U_{\phi}^{\dagger} A_{0}^{(k)} U_{\phi} &= \cos \phi' \, \sigma_{X} - \sin \phi' \sigma_{Y} =: A_{0}^{'(k)}\nonumber \\
    U_{\phi}^{\dagger} A_{1}^{(k)} U_{\phi} &= \cos \theta' \, \sigma_{X} + \sin \theta' \sigma_{Y} =: A_{1}^{'(k)},
\end{align}
and similarly for party $N$, where we used the definitions of $A_{x}^{(k)}$ in \cref{eq:thPhiSt_app}. We therefore find
\begin{multline}
    J':=(U_{\phi}\otimes U_{\phi})^{\dagger} J_{\phi',\theta'}^{(k,N)} (U_{\phi}\otimes U_{\phi}) \\ = \cos 2\theta' \cos (\theta' - \phi')  A_{0}^{'(k)}\otimes A_{0}^{'(N)} \\- \cos 2\theta' \cos 2 \phi' \big(  A_{0}^{'(k)} \otimes A_{1}^{'(N)}  + A_{1}^{'(k)} \otimes A_{0}^{'(N)}  \big) \\+ \cos 2\phi' \cos (\theta' - \phi')  A_{1}^{'(k)} \otimes A_{1}^{'(N)},  
\end{multline}
The largest and smallest eigenvalues of $J'$ are $\pm\eta^{\mathrm{Q}}_{\theta',\phi'}$, and are associated with the eigenvectors $(\ket{00} \pm i \ket{11})/\sqrt{2}$. We therefore have
\begin{equation}
    \bra{\Phi_{\bm{\mu},\phi}} I_{\bm{\mu}}^{(k,N)}\ket{\Phi_{\bm{\mu},\phi}}= (-1)^{n_{\bm{\mu}}}\bra{\Phi_{\bm{\mu}}} J' \ket{\Phi_{\bm{\mu}}} = \eta_{\phi',\theta'}^{\mathrm{Q}}.
\end{equation}

Putting everything together, we obtain
\begin{align}
    \sum_{\bm{\mu}} & \bra{\psi_{\mathrm{GHZ}}}P^{\overline{(k,N)}}_{\bm{\mu}|\bm{0}} \otimes I_{\bm{\mu}}^{(k,N)}\ket{\psi_{\mathrm{GHZ}}} \nonumber \\ 
    &= \sum_{\bm{\mu}} \bra{\psi_{\mathrm{GHZ}}}P_{\bm{\mu}|\bm{0}}^{\overline{(k,N)}}\ket{\psi_{\mathrm{GHZ}}} \bra{\Phi_{\bm{\mu},\phi}} I_{\bm{\mu}}^{(k,N)}\ket{\Phi_{\bm{\mu},\phi}}  \nonumber \\ 
    &= \eta_{\phi',\theta'}^{\mathrm{Q}}, \ \forall k \in \{1,...,N-1\}. 
\end{align}
The quantum bound $\eta_{N,\phi',\theta'}^{\mathrm{Q}}$ then follows by summing over $k$, which shows achievability. The minimum quantum value $-\eta_{N,\phi',\theta'}^{\mathrm{Q}}$ is achieved by relabelling the outcomes. Moreover, the quantum bound cannot be achieved classically since $I_{\phi',\theta'}$ corresponds to an expanded Bell inequality when $(\phi',\theta') \in \mathcal{F}$ (cf.\ the discussion below \cref{eq:psd1}).
\end{proof} 

\subsection{Proof of \cref{prop:phi}} \label{app:prop2}

\noindent \textbf{Proposition 7.} \textit{Let $\phi \in [0,\phi_{N}^{*}]$, and $\theta = \theta(\phi)$ as defined in \cref{eq:thPhi,eq:thPhi_2}. Then $(\phi',\theta') \in \mathcal{F}$, where $\phi'$ and $\theta'$ are defined in \cref{lem:NevenBI_2}.}
\begin{proof}
For convenience, we restate the set $\mathcal{F}$ below,
\begin{equation}
    \mathcal{F} =  \Big[ (-\pi/4,\pi/4) \times \mathcal{G} \Big] \cup \Big[ (-\pi/4,\pi/4) \setminus \{0\} \times \{\pi/2,3\pi/2\}  \Big].
\end{equation}
First we will consider the range of $\phi'$ as we vary $\phi$. By definition, $\phi' = \phi N/2$, which takes values over the interval $[0,\phi^{*}_{N}N/2]$. Then 
\begin{equation}
    \frac{\phi^{*}_{N}N}{2} = \mathrm{sgn}[\sin(2\theta_{N}^{*})]\frac{\pi}{8},
\end{equation}
which implies $-\pi/4 < \phi' < \pi/4$.

Next we consider $\theta'$ as a function of $\phi$,
\begin{equation}
    \theta' = \theta(\phi) - \frac{N-2}{2}\phi = -\frac{N(N-3)}{2(N+1)}\phi + \theta_{N}^{*}.
\end{equation}
When $\phi = 0$, $\theta' = \theta_{N}^{*}$, and when $\phi = \phi_{N}^{*}$,
\begin{align}
    \theta' &= -\mathrm{sgn}[\sin(2\theta_{N}^{*})]\frac{\pi}{8}\frac{N-3}{N+1} + \theta_{N}^{*} =: \tilde{\theta}'. \label{eq:thetap}
\end{align}
We first consider $N=8n +2$, which satisfies $\mathrm{sgn}[\sin(2\theta_{N}^{*})] = -1$ and $\pi/2 < \theta_{N}^{*} < 3\pi/4$. By inserting the definition of $\theta^{*}_{N}$ from \cref{prop:theta}, 
\begin{equation}
    \tilde{\theta}' = \frac{\pi}{8}\frac{8n-1}{8n+3} + \frac{2\pi}{8n+3}\big(\frac{8n+2}{4} + \frac{1}{2}\big) = \frac{5\pi}{8} \in (\pi/2,3\pi/4).
\end{equation}
We therefore see $\theta' \in (\pi/2,3\pi/4)$ (since $\theta'$ is a linear function of $\phi$ and lies in that range for the extremal values of $\phi$). 

For $N= 8n + 8$, $\mathrm{sgn}[\sin(2\theta_{N}^{*})] = -1$ and $3\pi/2 < \theta_{N}^{*} < 7\pi/4$. We have
\begin{multline}
    \tilde{\theta}' = \frac{\pi}{8}\frac{8n+5}{8n+9} + \frac{2\pi}{8n+9}\big(\frac{3(8n+8)}{4} + 1\big) \\ = \frac{13\pi}{8} \in (3\pi/2,7\pi/4),
\end{multline}
implying $3\pi/2 < \theta' < 7\pi/4$. For $N= 8n + 4$, we have $\mathrm{sgn}[\sin(2\theta_{N}^{*})] = +1$, $\pi/4 < \theta_{N}^{*} < \pi/2$ and
\begin{equation}
    \tilde{\theta}' = -\frac{\pi}{8}\frac{8n+1}{8n+5} + \frac{2\pi}{8n+5}\frac{8n+4}{4} = \frac{3\pi}{8} \in (\pi/4,\pi/2),
\end{equation}
implying $\pi/4<\theta'<\pi/2$. Finally, for $N = 8n + 6$, $\mathrm{sgn}[\sin(2\theta_{N}^{*})] = +1$, $5\pi/4 < \theta_{N}^{*} < 3\pi/2$ and
\begin{multline}
    \tilde{\theta}' = -\frac{\pi}{8}\frac{8n+3}{8n+7} + \frac{2\pi}{8n+7}\Big(\frac{3(8n+6)}{4}+\frac{1}{2}\Big) \\ = \frac{11\pi}{8} \in (5\pi/4,3\pi/2),
\end{multline}
implying $\pi/4<\theta'<\pi/2$. For every $N$, we have shown $\theta' \in \mathcal{G}$ and $\phi' \in (-\pi/4,\pi/4)$ for $\phi \in [0,\phi^{*}_{N}]$, hence $(\phi',\theta')\in \mathcal{F}$. 
\end{proof}

\subsection{Evidence for \cref{conj:maxMerm_2}} \label{app:conj2E}

\noindent \textbf{Conjecture 2.} \textit{For $N$ even, the maximum randomness unconditioned on Eve, $r$, that can be generated by quantum strategies achieving an MABK value, $s$, is given by  
\begin{equation}
    r(s) = \begin{cases}
        N, \ s \in (1,m_{N}^{*}], \\
        r(\phi_{s}), \ s \in (m_{N}^{*} ,2^{(N-1)/2} ],
        \end{cases}\label{eq:maxRand2}
\end{equation}
where $r(\phi)$ is defined in \cref{eq:rphi}, and 
\begin{equation}
    \phi_{s} = \mathrm{arg \, min} \big\{ |\phi| \ | \ \phi \in [0,\phi_{N}^{*}], \langle M_{N}(\phi,\theta(\phi))\rangle = s \big\}.
\end{equation}
} 

\subsubsection{Numerical approach for upper bounding maximum randomness versus MABK value} \label{app:numMethod}
The objective of this section is to generalize the technique introduced in \cite{WBC}, which found upper bounds on the maximum randomness versus CHSH value using SOS decompositions. A python implementation of this technique is provided in our GitHub repository~\cite{code}. We begin by defining the problem in detail.

Let $\mathcal{Q}$ denote the set of quantum distributions, and $\mathcal{M}(P)$ denote the MABK value of a distribution $P \in \mathcal{Q}$. Moreover, let $H(\bm{R}|\bm{X}=\bm{0},E)_{P}$ be the conditional von Neumann entropy of the outputs $\bm{R}$ for inputs $\bm{X} = \bm{0}$ given the observed distribution $P$, taking the infimum over all quantum strategies that could give rise to $P$, i.e.,
\begin{multline}
   H(\bm{R}|\bm{X} = \bm{0},E)_{P} \\ := \inf_{\substack{ \ket{\Psi}_{\tilde{\bm{Q}}E}, \\ \big\{\{\tilde{P}_{a_{k}|x_{k}}^{(k)}\}_{a_{k}} \big\}_{k}, \\ \text{compatible with} \ P }} H(\bm{R}|\bm{X} = \bm{0},E)_{\rho_{\bm{R}E|\bm{0}}}.
\end{multline}
Similarly, let $H(\bm{R}|\bm{X}=\bm{0})_{P}$ be the Shannon entropy of the distribution on $\bm{R}$ for inputs $\bm{X} = \bm{0}$.  Then the curve $R:[m_{N}^{*},2^{(N-1)/2}] \to [0,N], \ s \mapsto R(s)$ we want to find is defined by the optimization
\begin{align}
    R(s) = \max_P \ &H(\bm{R}|\bm{X} = \bm{0},E)_{P} \nonumber \\
     \text{s.t.} \  \ & \mathcal{M}(P) = s, \nonumber \\
    & P \in \mathcal{Q}.
\end{align}
Our technique proceeds by defining a sequence of upper bounds on $\text{Graph}[R(s)] = \{(s,r) \ | \ r = R(s) \}$. [We speak of upper bounds on $\text{Graph}[R(s)]$ because we either upper bound $R(s)$ or $R^{-1}(r)$, and, because the functions turn out to be monotonic, either gives an upper bound.] 

We begin by removing the dependence on Eve, which results in an upper bound $R(s) \leq \bar{R}(s)$ on the previous problem via strong subadditivity of the von Neumann entropy. Explicitly, 
\begin{align}
    R(s) \leq \bar{R}(s) = \max_P \ &H(\bm{R}|\bm{X} = \bm{0})_{P} \nonumber \\
     \text{s.t.} \  \ & \mathcal{M}(P) = s, \nonumber \\
    & P \in \mathcal{Q}. \label{eq:Rsbar}
\end{align}

Next we state the following two lemmas, adapted from Ref.~\cite{WBC}. 
\begin{lemma}\label{lem:monot}[Monotonicity of $\bar{R}(s)$] Suppose $m_{N}^{*}$ is the largest MABK value compatible with $p(\bm{a}|\bm{0}) = 2^{-N}$ (cf.\ Conjecture~\ref{conj:maxMerm}). Then the function $\bar{R}(s)$ is strictly decreasing on $[m_N^*,2^{(N-1)/2}]$.
\label{lem:mon}
\end{lemma}

We will prove Lemma~\ref{lem:monot} as a corollary of the following.
\begin{lemma}\label{lem:concavity}
    Let $\mathcal{Q}'$ be a convex set, $H:\mathcal{Q}'\to\mathbb{R}$ be concave, $\mathcal{M}':\mathcal{Q}'\to\mathbb{R}$ be linear and $\bar{R}'(s)=\max_P H(P)$ subject to $\mathcal{M}'(P)=s$, $P\in\mathcal{Q}'$. Then $\bar{R}'(s)$ is convex.
\end{lemma}
\begin{proof}
We directly calculate
    \begin{align}
    \bar{R}'[\lambda s_1 + (1-\lambda)s_2] = \max_P \ &H(P) \nonumber \\
     \text{s.t.} \  \ & \mathcal{M}'(P) = \lambda s_1 + (1-\lambda) s_2, \nonumber \\
    & P \in \mathcal{Q}', \nonumber \\
    \geq \max_{P_1,P_2}\  & H(\lambda P_1 + (1-\lambda)P_2) \nonumber \\
    \text{s.t.} \  \ & \mathcal{M}'(P_1) = s_1,\ \mathcal{M}'(P_2) = s_2, \nonumber \\
    & P_1,P_2 \in \mathcal{Q}', \nonumber \\
    \geq \max_{P_1,P_2}\  & \lambda H(P_1) + (1-\lambda)H(P_2) \nonumber \\
    \text{s.t.} \  \ & \mathcal{M}'(P_1) = s_1,\ \mathcal{M}'(P_2) = s_2, \nonumber \\
    & P_1,P_2 \in \mathcal{Q}', \nonumber \\
    = \lambda \bar{R}'(s_1 & )  +  (1-\lambda)\bar{R}'(s_2),
\end{align}
where the first inequality follows due to the linearity of $\mathcal{M}'$ and convexity of $\mathcal{Q}'$, and the second inequality due to the concavity of $H$.
\end{proof}
    
\begin{proof}[Proof of Lemma~\ref{lem:monot}]
Note that 
$H(\bm{R}|\bm{X} = \bm{0})_{P}$ is concave in $P$, the quantum set, $\mathcal{Q}$, is convex and the MABK value, $\mathcal{M}$, is linear. Lemma~\ref{lem:concavity} hence implies that $\bar{R}(s)$ is concave in $s$.

Then note that $\bar{R}(m_{N}^{*}) > \bar{R}(s) \ \forall s\in (m_{N}^{*},2^{(N-1)/2}]$. This follows from the assumption that the largest MABK value achievable when $p(\bm{a}|\bm{0}) = 2^{-N}$ is $m_{N}^{*}$, and hence $\bar{R}(s)$ is initially decreasing, which implies the claim. 
\end{proof}
In the case $N=2$, it is known that $m_{2}^{*}$ is the largest MABK value compatible with $p(a,b|0,0) = 1/4$; see~\cite[Corollary 1]{WBC}. For $N = 4,6,8,10,12$, we have verified tightness numerically, as shown in \cref{tab:evidence}\footnote{Monotonicity of $\bar{R}(s)$ is not needed to establish \cite[Corollary 1]{WBC} or the upper bounds presented in \cref{tab:evidence}.}

\begin{lemma}[Inverse function of $\bar{R}(s)$]\label{lem:inv}
    Suppose $r = \bar{R}(s)$. The function $\bar{R}(s)$ has the following inverse, denoted $\bar{R}^{-1}$, that satisfies $s = \bar{R}^{-1}(r)$, given by
    \begin{align}
    \bar{R}^{-1}(r) = \max_{P} \ &\mathcal{M}(P) \nonumber \\
     \mathrm{s.t.} \  \ & H(\bm{R}|\bm{X} = \bm{0})_{P} = r, \nonumber \\
    & P \in \mathcal{Q}. \label{eq:invR}
\end{align}
\end{lemma} 
\begin{proof}
We prove \cref{lem:inv} by showing $\bar{R}^{-1}(\bar{R}(s)) = s$, and $\bar{R}(\bar{R}^{-1}(r)) = r$, and using \cref{lem:mon}. First consider $\bar{R}^{-1}(\bar{R}(s)) = s$,
\begin{align}
    \bar{R}^{-1}(\bar{R}(s)) = \max_{P} \ &\mathcal{M}(P) \nonumber \\
     \text{s.t.} \  \ & H(\bm{R}|\bm{X} = \bm{0})_{P} = \bar{R}(s), \nonumber \\
    & P \in \mathcal{Q}.
\end{align}
The constraint $H(\bm{R}|\bm{X} = \bm{0})_{P} = \bar{R}(s)$ implies that the achievable MABK values for the distribution $P$ lies to the left of $s$, i.e., $\mathcal{M}(P) \leq s$, since the curve $\bar{R}(s)$ is strictly decreasing (cf.\ \cref{lem:mon}). To see this, suppose $s' = \mathcal{M}(P) > s$. Then $\bar{R}(s') < \bar{R}(s)$ by \cref{lem:mon}. But, $\bar{R}(s') \geq H(\bm{R}|\bm{X} = \bm{0})_{P}$, implying $H(\bm{R}|\bm{X} = \bm{0})_{P} < \bar{R}(s)$, which by assumption cannot hold, giving a contradiction. Therefore $\bar{R}^{-1}(\bar{R}(s)) \leq \max_{\{P\in \mathcal{Q} \ \text{s.t.} \ \mathcal{M}(P) \leq s\}}\mathcal{M}(P) = s$. Moreover, this bound is saturated by the behaviour $P^{*}$, which achieves the maximum defining $\bar{R}(s)$\footnote{Note that the function $P \mapsto H(\bm{R}|\bm{X}=\bm{0})_{P}$ is continuous, and the quantum set $\mathcal{Q}$ is compact (see the discussion in~\cite[Appendix B]{GohGeometry}), from which it follows that $\mathcal{Q}_{s} = \{ P \in \mathcal{Q} \ : \ \mathcal{M}(P) = s\}$ is also compact. This implies that the objective function in \eqref{eq:Rsbar} attains its maximum over its domain.}, i.e., satisfies $\bar{R}(s) = H(AB|X=0,Y=0)_{P^{*}}$ and $\mathcal{M}(P^{*}) = s$. This establishes $\bar{R}^{-1}(\bar{R}(s)) = s$.

For the other direction $\bar{R}(\bar{R}^{-1}(r))$, the same reasoning holds. The constraint $\mathcal{M}(P) = \bar{R}^{-1}(r)$ implies that $H(\bm{R}|\bm{X} = \bm{0})_{P} \leq r$ since any distribution that achieves a CHSH value of $\bar{R}^{-1}(r)$ can generate no more than $r$ bits of randomness. Hence $\bar{R}(\bar{R}^{-1}(r)) \leq r$. This bound is saturated by the behaviour $P^{*}$ that achieves the maximum defining $\bar{R}^{-1}(s)$, completing the proof.
\end{proof}

From the above lemma, we can solve for upper bounds on the points $(s,R(s)) \in \text{Graph}[R(s)]$ using the inverse function, i.e., $(s,R(s)) \leq (s,\bar{R}(s)) = (\bar{R}^{-1}(r),r)$ where $\bar{R}(s) = r$. What remains is to upper bound $\bar{R}^{-1}(r)$, which will correspond to an upper bound on $R(s)$ due to the monotonicity argument. Next we obtain an upper bound by applying a symmetrization step to the distribution $\{p(\bm{a}|\bm{0})\}$:
\begin{lemma}
    Let $\mathcal{E}$ be the local channel that maps all even (odd) parity outcome strings to a uniform mixture of all even (odd) parity outcome strings, 
    \begin{equation}
        \mathcal{E}[p(\bm{a}|\bm{0})] = \begin{cases}
            \frac{1}{2^{N-1}}\sum_{\bm{b} \ \mathrm{even}} p(\bm{b}|\bm{0}), \ \mathrm{if} \ \bm{a} \ \mathrm{has \ even \ parity,} \\
            \frac{1}{2^{N-1}}\sum_{\bm{b} \ \mathrm{odd}} p(\bm{b}|\bm{0}), \ \mathrm{if} \ \bm{a} \ \mathrm{has \ odd \ parity.} 
        \end{cases}
    \end{equation}
    The entropy after applying $\mathcal{E}$ is non-decreasing, and the MABK value is invariant under $\mathcal{E}$.  
    \label{lem:entSimp}
    \end{lemma}
\begin{proof}
The first claim follows from the data processing inequality, $H(\bm{R}|\bm{X}=\bm{0})_{P} \leq H(\bm{R}|\bm{X}=\bm{0})_{\mathcal{E}(P)}$. The second claim comes from the fact that the correlators $\langle A_{\bm{x}} \rangle$ are invariant under $\mathcal{E}$.
\end{proof}
After the map $\mathcal{E}$ is applied, the probabilities are symmeterized, i.e., $p(\bm{a}|\bm{0}) = \epsilon$ if $\bm{a}$ is even and $p(\bm{a}|\bm{0}) = (1 - 2^{N-1}\epsilon)/2^{N-1} = (2-2^{N}\epsilon)/2^{N}$ if $\bm{a}$ is odd, for $\epsilon \in [0,2^{1-N}]$. We can hence define the following upper bound on $\bar{R}(s) \leq \bar{\bar{R}}(s)$ where
\begin{align}
    \bar{\bar{R}}(s) =  \max \ &H(\bm{R}|\bm{X}=\bm{0})_{\mathcal{E}(P)} \nonumber \\
     \text{s.t.} \  \ & \mathcal{M}(\mathcal{E}(P)) = s, \nonumber \\
    & P \in \mathcal{Q} \nonumber \\
    = \max \ &H(\bm{R}|\bm{X}=\bm{0})_{P} \nonumber \\
     \text{s.t.} \  \ & \mathcal{M}(P) = s, \nonumber \\
     & p(\bm{a}|\bm{0}) = \begin{cases}
         \epsilon \ \mathrm{if} \ \bm{a} \ \mathrm{is \ even,} \\
         (2-2^{N}\epsilon)/2^{N} \ \mathrm{if} \ \bm{a} \ \mathrm{is \ odd,}
     \end{cases} \nonumber \\
    & P \in \mathcal{Q}.
\end{align}
We can now apply \cref{lem:inv} (which still holds after applying $\mathcal{E}$, since it preserves the convexity of $\mathcal{Q}$), and define the function 
\begin{align}
    \bar{\bar{R}}^{-1}(r) = \max \ &\mathcal{M}(P) \nonumber \\
     \text{s.t.} \  \ & H(\bm{R}|\bm{X}=\bm{0})_{P} = r, \nonumber \\
     & p(\bm{a}|\bm{0}) = \begin{cases}
         \epsilon \ \mathrm{if} \ \bm{a} \ \mathrm{is \ even,} \\
         (2-2^{N}\epsilon)/2^{N} \ \mathrm{if} \ \bm{a} \ \mathrm{is \ odd,}
     \end{cases} \nonumber \\
    & P \in \mathcal{Q}. \label{eq:invSym}
\end{align}
We can further relax this by considering the correlator after symmetrization, $\langle A_{\bm{0}} \rangle  = 2^{N-1}\epsilon - 2^{N-1}(2-2^{N}\epsilon)/2^{N} = 2^{N}\epsilon - 1$. Then we have $\bar{\bar{R}}^{-1}(r) \leq \bar{\bar{\bar{R}}}^{-1}(r)$ where
\begin{align}
    \bar{\bar{\bar{R}}}^{-1}(r) = \max \ &\mathcal{M}(P) \nonumber \\
     \text{s.t.} \  \ & H(\bm{R}|\bm{X}=\bm{0})_{P} = r, \nonumber \\
     & \langle A_{\bm{0}} \rangle = 2^{N}\epsilon - 1, \nonumber \\
    & P \in \mathcal{Q}. \label{eq:invSym2}
\end{align}
This follows from the fact that the feasible region of the former optimization is a subset of the latter. 

Next, consider the expression for the entropy $H(\bm{R}|\bm{X}=\bm{0})_{P}$ of a symmetrized distribution $P$. By direct computation we find
\begin{align}
    H(\bm{R}|\bm{X} = \bm{0})_{\mathcal{E}(P)} &= - \sum_{\bm{a} \ \mathrm{even}}\epsilon\log \epsilon \nonumber \\
    &\ \ \ \ \ \ \ \ \ \ -\sum_{\bm{a} \ \mathrm{odd}}\frac{2-2^{N}\epsilon}{2^N}\log \frac{2-2^{N}\epsilon}{2^N}\nonumber \\
    &= - 2^{N-1} \Big( \epsilon\log \epsilon + \frac{2-2^{N}\epsilon}{2^N} \log \frac{2-2^{N}\epsilon}{2^N} \Big) \nonumber \\
    &= -2^{N-1} \Big( \epsilon \log 2^{N-1}\epsilon - \epsilon(N-1) \nonumber \\
    &+ \frac{1}{2^{N-1}}(1-2^{N-1}\epsilon) \nonumber \\ 
    & \ \ \ \ \ \ \ \ \ \ \cdot (\log (1-2^{N-1}\epsilon) - (N-1) ) \Big) \nonumber \\
    &= N - 1 +  H_{\mathrm{bin}}(2^{N-1}\epsilon). \label{eq:symEnt}
\end{align}
Since $2^{N-1}\epsilon \in [0,1/2]$ for $\epsilon \in [0,1/2^{N}]$, there is a one-to-one correspondence between $H(\bm{R}|\bm{X} = \bm{0})_{\mathcal{E}(P)}$ and $\epsilon$. Hence if $H(\bm{R}|\bm{X} = \bm{0})_{\mathcal{E}(P)} = r$, there exists a unique $\epsilon_{r} \in [0,1/2^{N}]$ that satisfies $r = N - 1 +  H_{\mathrm{bin}}(2^{N-1}\epsilon_{r})$. We can therefore write 
\begin{align}
    \bar{\bar{\bar{R}}}^{-1}(r) = \max \ &\mathcal{M}(P) \nonumber \\
     \text{s.t.} \  \ & 
     \langle A_{\bm{0}} \rangle = 2^{N}\epsilon_{r} - 1, \nonumber \\
    & P \in \mathcal{Q}. \label{eq:invSym3}
\end{align}

Next we introduce the function
\begin{align}
    \tilde{R}^{-1}(r) = \min_{t,z,S} \ &t \nonumber \\
    \mathrm{s.t.} \ &t\mathbb{I} - M_{N} = S 
      + z\Big( A_{\bm{0}} &- \big( 2^{N}\epsilon_{r} - 1 \big) \mathbb{I} \Big)\nonumber \\
    &S \succeq 0.
\end{align}
Here $S$ is an SOS decomposition for the operator expression
\begin{equation}
    S = \Big( t + z\big( 2^{N}\epsilon_{r} - 1 \big)\Big)\mathbb{I} - \Big( M_{N} + zA_{\bm{0}} \Big) \succeq 0,
\end{equation}
hence, for any feasible point $(t,z,S)$, all quantum correlations satisfying $\langle A_{\bm{0}} \rangle = 2^{N}\epsilon_{r} - 1$ also satisfy $t \geq \langle M_{N} \rangle$. This implies $\bar{\bar{\bar{R}}}^{-1}(r) \leq \tilde{R}^{-1}(r)$. Now, by choosing a basis of monomials for $S$, one can evaluate the function $\tilde{R}^{-1}(r)$ numerically using a standard SDP solver. We refer the reader to Ref.~\cite{BampsPironio} for details of how to find such formulation, and provide a self-contained description below. 

If the monomial basis is chosen to be $\{R_{\mu}\}_{\mu}$, where each $R_{\mu}$ is a product of operators from the set $\{A_{x_{k}}^{(k)}\}_{x_{k},k}$, then one can construct polynomials of the form $M_{i} = \sum_{\mu}q^{\mu}_{i}R_{\mu}$, $q_{i}^{\mu} \in \mathbb{C}$. $S$ then admits an SOS decomposition if $q^{\mu}_{i},R_{\mu}$ are chosen such that
\begin{equation}
    S = \sum_{i} M_{i}^{\dagger}M_{i} = \sum_{\mu\nu} M^{\mu\nu} R_{\mu}^{\dagger}R_{\nu},
\end{equation}
where $M^{\mu\nu} = \sum_{i}(q_{i}^{\mu})^{*}q_{i}^{\nu}$ are matrix elements $\bm{M}$, the Gram matrix of the set of vectors $\{\bm{q}^{\mu}\}_{\mu}$, where $\bm{q}^{\mu} = [q_{1}^{\mu},...,q_{i}^{\mu},...]^{\mathrm{T}}$. One now arrives at the following optimization problem
\begin{align}
    \min_{t,z,\bm{M}} \ &t \nonumber \\
    \mathrm{s.t.} \ &\sum_{\mu\nu}M^{\mu\nu}R_{\mu}^{\dagger}R_{\nu} = \Big( t + z\big( 2^{N}\epsilon_{r} - 1 \big)\Big)\mathbb{I} \nonumber \\
      & \ \ \ \ \ \ \ \ - \Big( M_{N} + zA_{\bm{0}} \Big),\nonumber \\
    &\bm{M} \succeq 0.
\end{align}
Now one can write $R_{\mu}^{\dagger}R_{\nu} = \sum_{i}F^{i}_{\mu\nu}E_{i}$, where $\bm{F}^{i}$ is some complex matrix with elements $F^{i}_{\mu\nu}$, for $E_{i} \in \mathrm{Canonical}\big[\{R_{\mu}^{\dagger}R_{\nu}\}_{\mu\nu}\big]$, which denotes the  canonical reduction of the set $\{R_{\mu}^{\dagger}R_{\nu}\}_{\mu\nu}$ using the commutation and projective relations (see Ref.~\cite{BampsPironio} for details). Similarly, one can express the operator expression in this basis, $\Big( t + z\big( 2^{N}\epsilon_{r} - 1 \big)\Big)\mathbb{I} - \Big( M_{N} + zA_{\bm{0}}\Big) = \sum_{i}s^{i}E_{i}$, and arrive at an SDP of the form  
\begin{align}
    \min_{t,z,\bm{M}} \ &t \nonumber \\
    \mathrm{s.t.} \ &\sum_{\mu\nu}F_{\mu\nu}^{i}M^{\mu\nu} = s^{i} \ \forall i,\nonumber \\
    &\bm{M} \succeq 0.
\end{align}
For our numerical calculations, we choose the monomial basis $\Big\{ \prod_{k=1}^{N/2}A_{x_{k}}^{(k)},\prod_{k=N/2+1}^{N}A_{x_{k}}^{(k)}\Big\}_{x_{k}}$, consisting of $2^{N/2+1}$ elements.

\subsubsection{Numerical results}
We present the numerical results which support our conjecture in \cref{fig:numEv}, for $N=4,6,8,10,12$. We find the two curves coincide very well, suggesting the lower bounds are indeed tight. The case of $N=2$ was proven in Ref.~\cite{WBC}. The technique also generates the numerical results in \cref{tab:evidence}.

\subsection{Other Bell inequalities for $N=3$} \label{app:S1S2}
The above technique can be directly applied to find bounds on other Bell inequalities. We applied this to the tripartite expressions in \cref{eq:S1,eq:S2}, which, along with the MABK expression, trivial extensions of CHSH and positivity cover all inequivalent classes of facets for the local polytope with uniform marginals (see Ref.~\cite{Sliwa_2003} for the general case). The exact trade-offs we found are detailed in \cref{fig:S1S2}, and we summarize various bounds of each class in \cref{tab:S1S2}. We restate the Bell expressions studied below~\cite{Werner01}:
\begin{align}
    M_{3} &= A_{0}(B_{0}C_{1}+B_{1}C_{0})/2 + A_{1}(B_{0}C_{0}- B_{1}C_{1})/2, \nonumber \\
    S_{1} &= \frac{1}{4} \sum_{x,y,z}A_{x}B_{y}C_{z} - A_{1}B_{1}C_{1}, \nonumber\\
    S_{2} &= A_{0}B_{0}(C_{0} + C_{1}) - A_{1}B_{1}(C_{0} - C_{1}), \nonumber \\
    S_{3} &= A_{0}(B_{0}C_{0}+B_{0}C_{1}+B_{1}C_{0}-B_{1}C_{1}), \nonumber \\
    S_{4} &= A_{0}B_{0}C_{0}, \label{eq:facets}
\end{align}
where we have made the substitution $A_{x_{0}}^{(0)} \rightarrow A_{x}$, $A_{x_{1}}^{(1)} \rightarrow B_{y}$ etc.\ for readability. For $S_{1}$ and $S_{2}$, we applied our numerical technique by substituting them in for the MABK expression. When $\epsilon = 1/8$, the SDP returns an upper bound on the maximum violation with maximum randomness, which we label $s^{*}$ (the facet class in question will be clear from context). For $S_{1}$, we numerically identified that $\epsilon = 1/6$ achieves the maximum quantum value of $5/3$, for which the corresponding entropy is given by \cref{eq:symEnt}, which evaluates to $2 + H_{\mathrm{bin}}(2/3)$; by varying $\epsilon \in [1/8,1/6]$ we achieve the trade off in \cref{fig:S1S2}, which corresponds to an upper bound. Similarly, for $S_{2}$ we found $\epsilon = (2+\sqrt{2})/16$ achieves the maximum quantum value $2\sqrt{2}$, so by varying $\epsilon \in [1/8,(2+\sqrt{2})/16]$ we find the curve in \cref{fig:S1S2}. The corresponding entropy is $2 + H_{\mathrm{bin}}((2+\sqrt{2})/4)$ which is exactly 1 plus the amount certified by maximum CHSH violation~\cite{Bhavsar2023,BrownDeviceIndependent,WBC}.  

We checked the value of our constructions in \cref{eq:evenStrat}, with $N=3$, for the expressions $\langle S_{1} \rangle$ and $\langle S_{2} \rangle$. We find
\begin{align}
    \langle S_{1} \rangle &= 2 \sin (3\theta/2) \cos^{3} (\theta/2), \nonumber \\
    \langle S_{2} \rangle &= \sin \theta - \sin 2\theta + \sin 3\theta.
\end{align}
By varying $\theta$, we find one can achieve a wide range of values from the local bound, however the values of $s^{*}$ are inaccessible using this parameterization; one will need to consider a larger family of strategies to further investigate if $s^{*}$ is achievable with maximum randomness.   

\onecolumngrid

\begin{figure}[h]
\includegraphics[width=8.4cm]{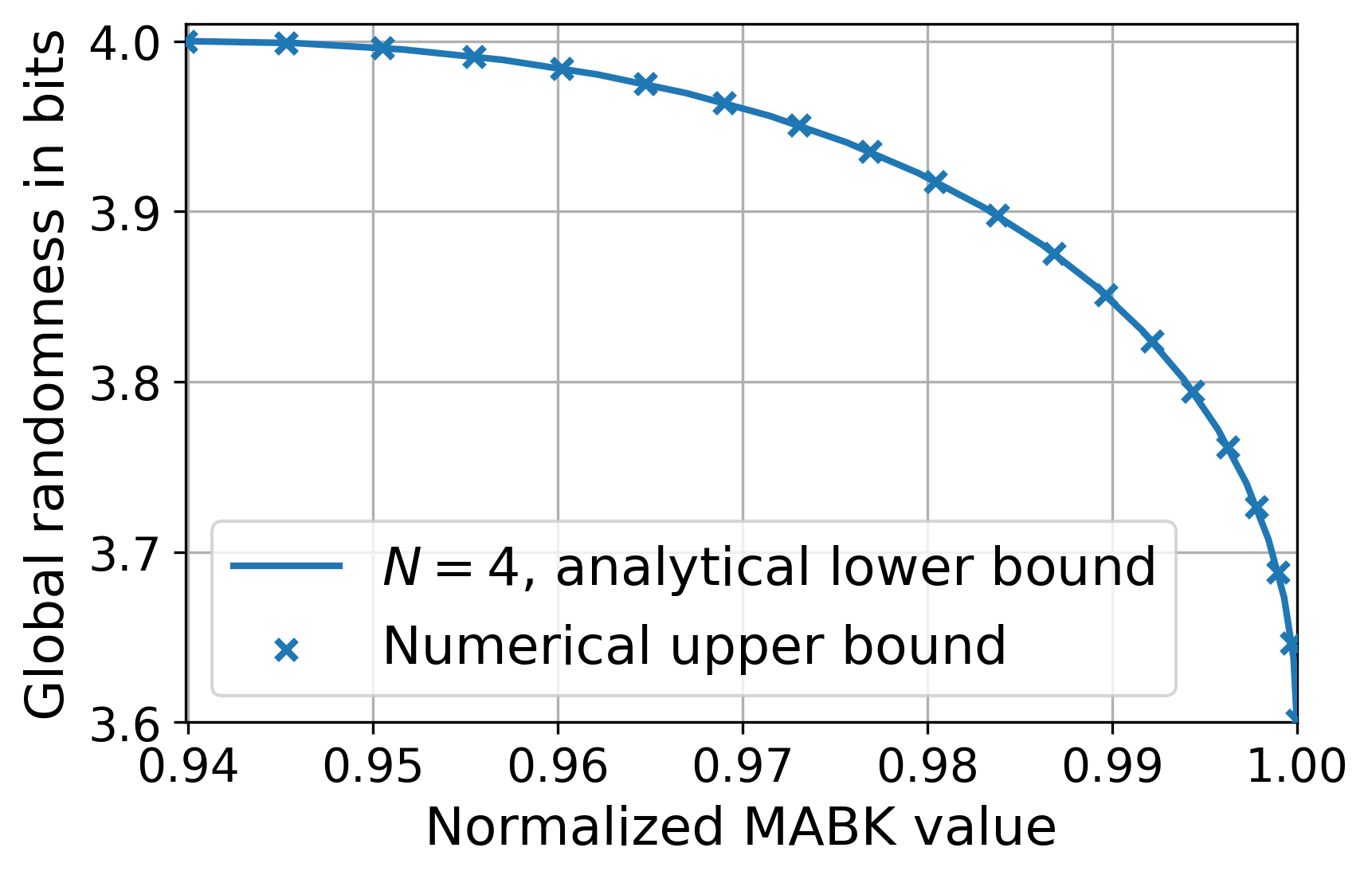}
\includegraphics[width=8.4cm]{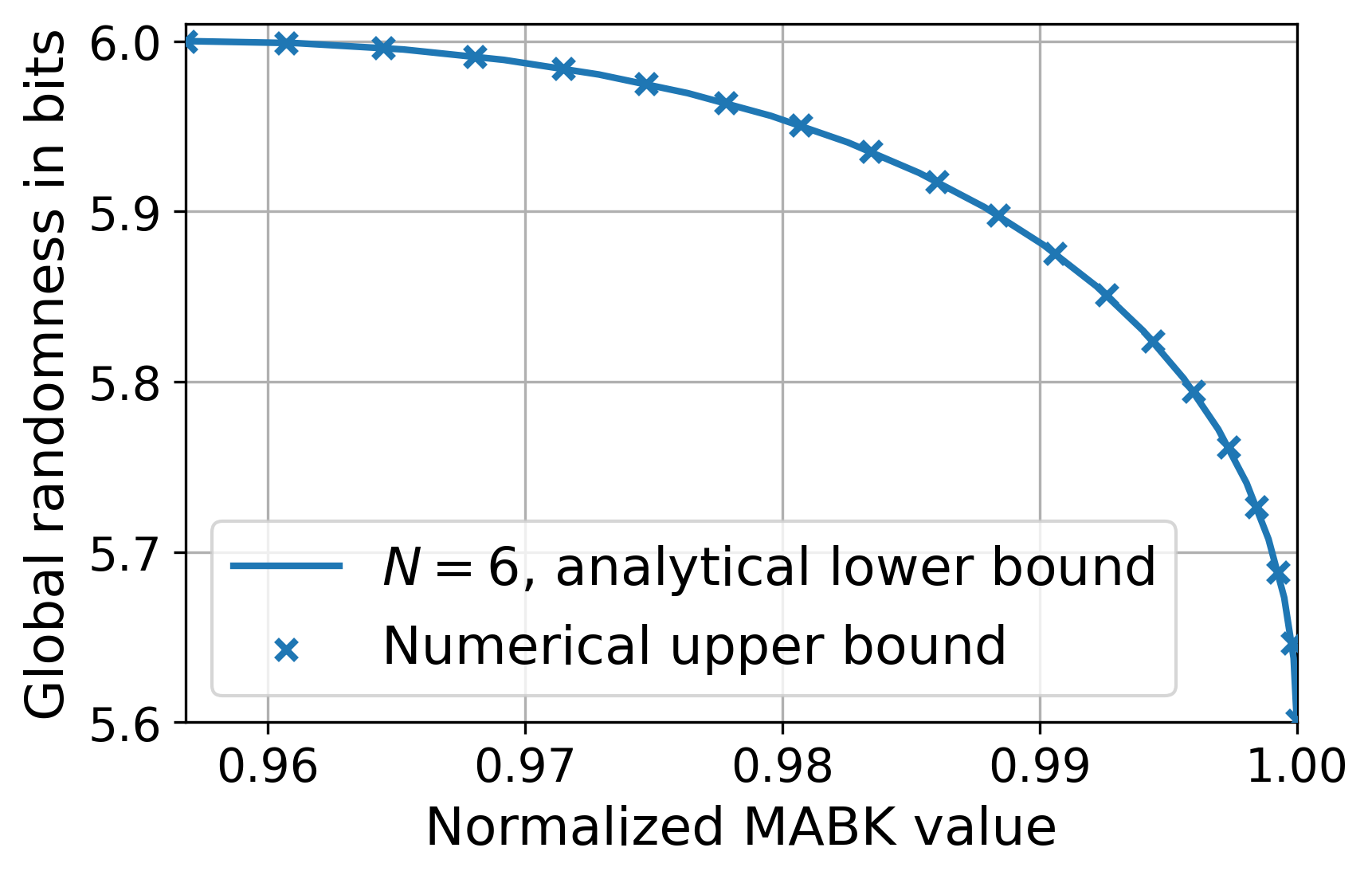}
\includegraphics[width=8.4cm]{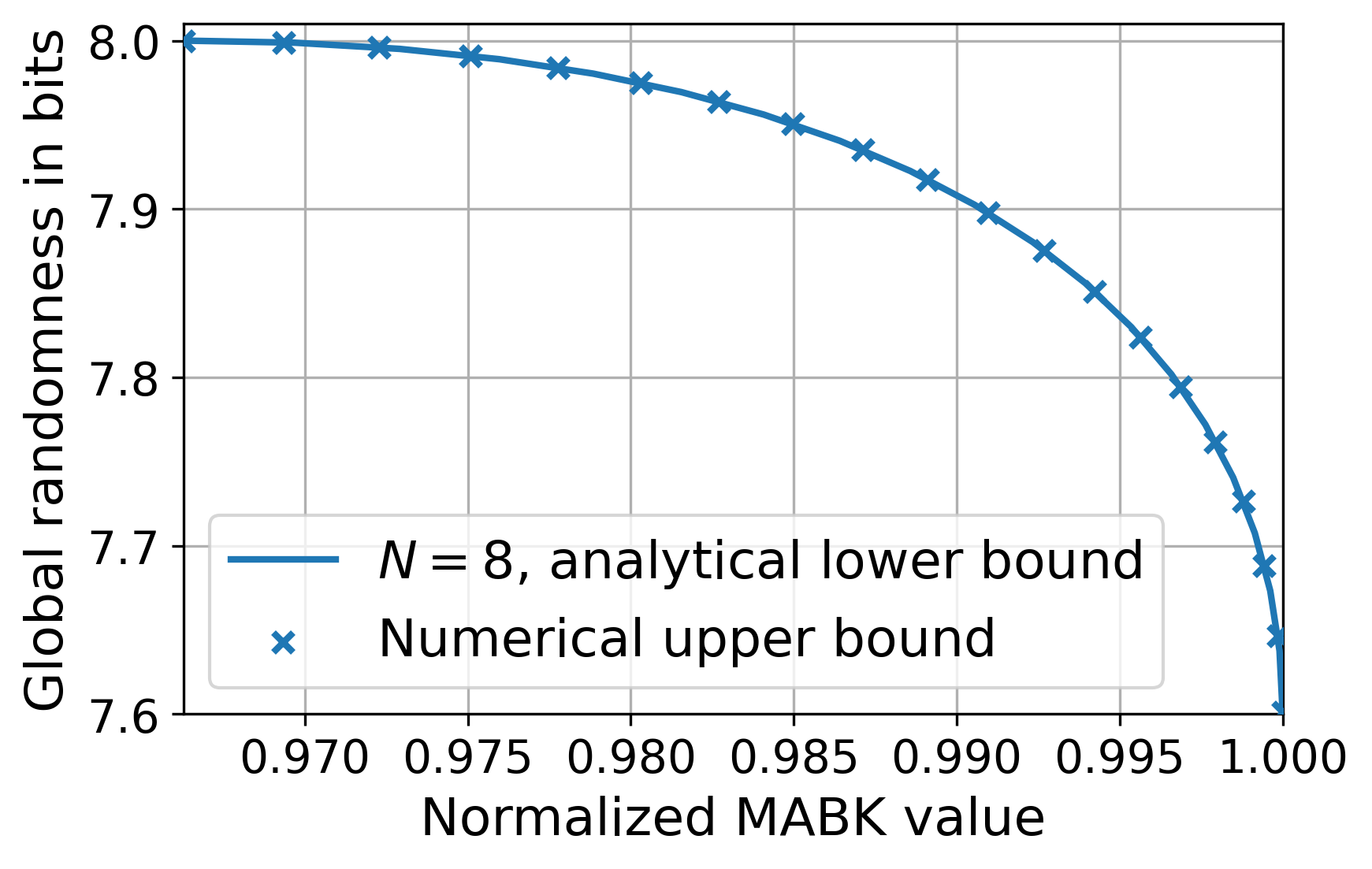}
\includegraphics[width=8.4cm]{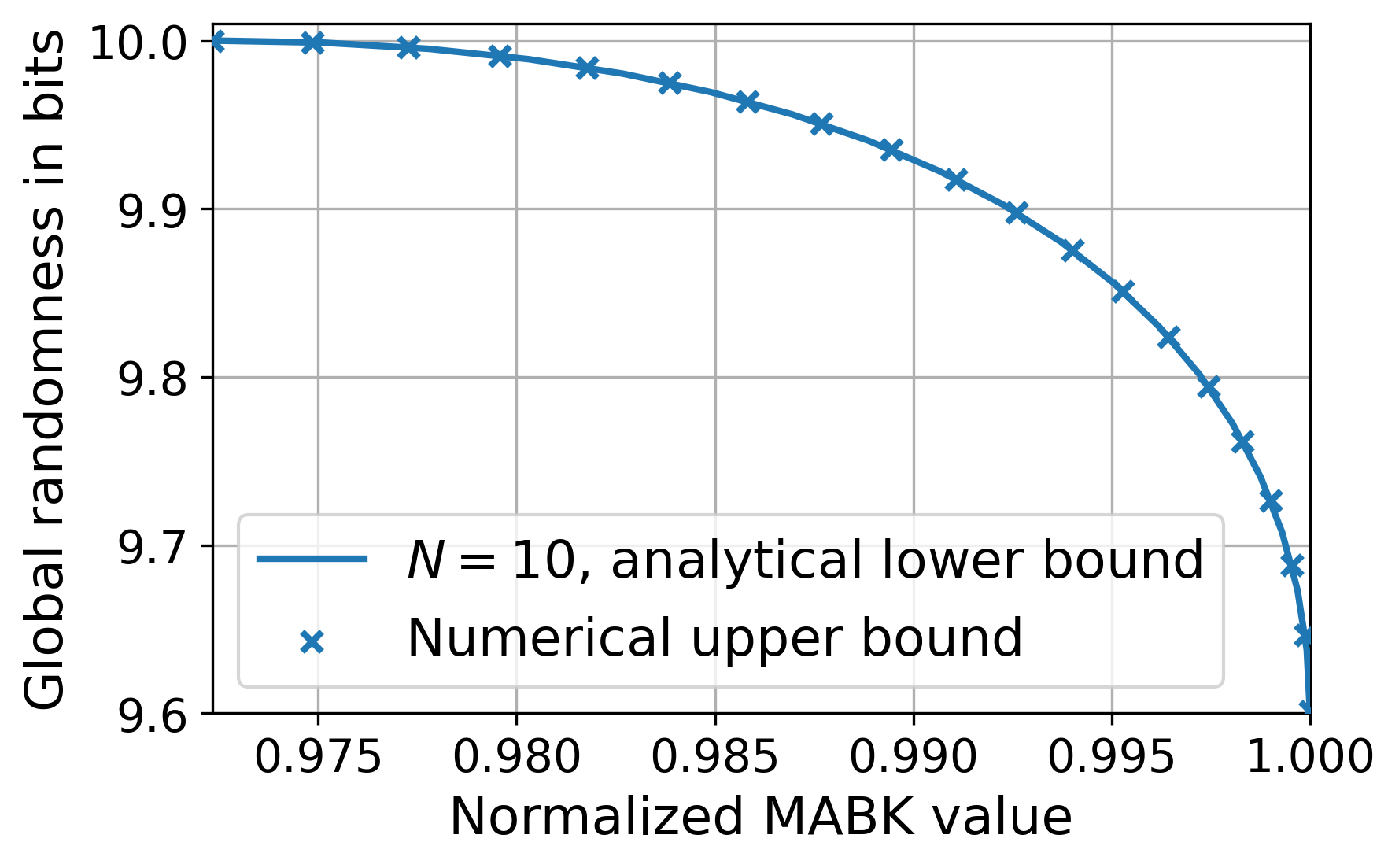}
\includegraphics[width=8.4cm]{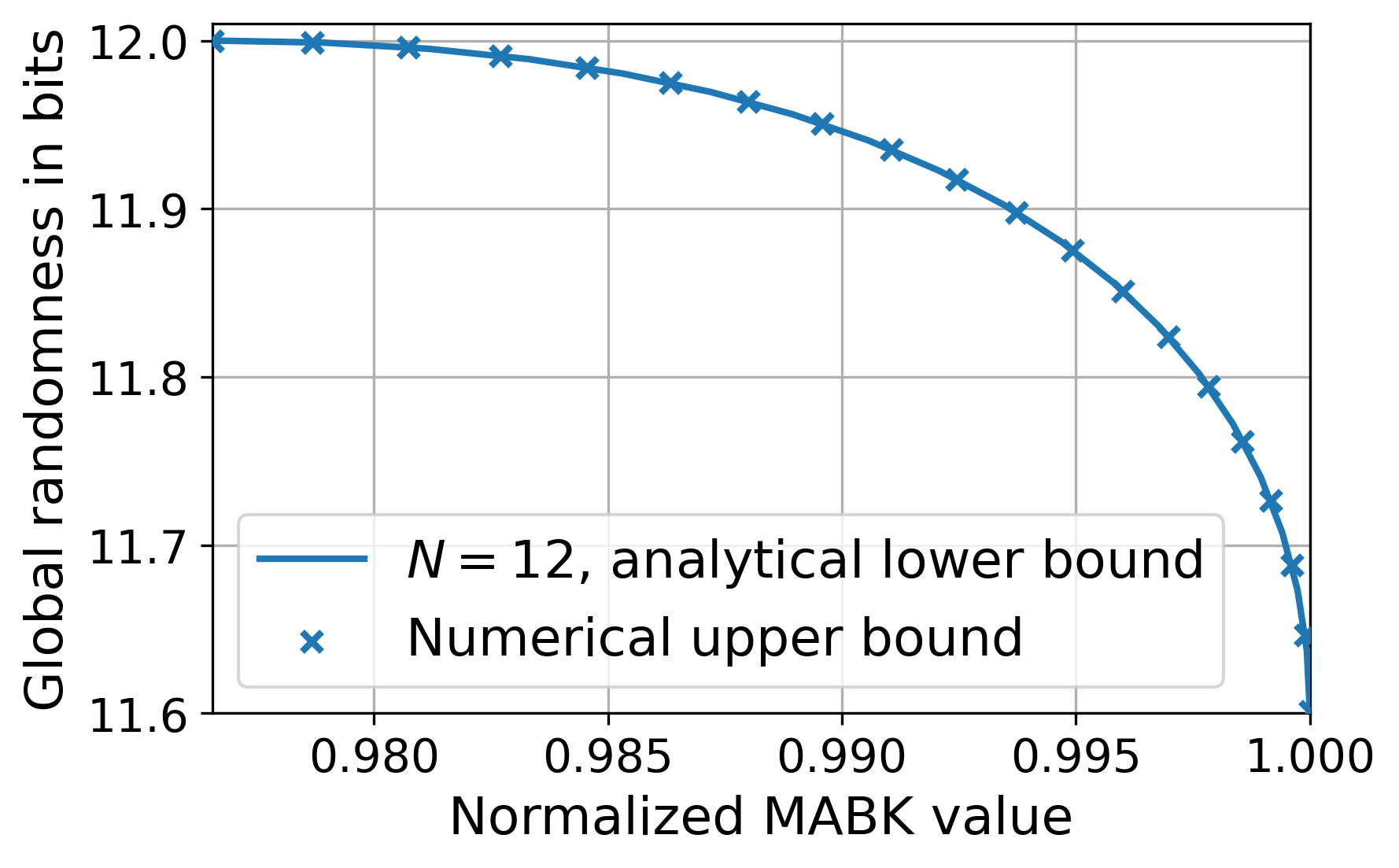}
\centering
\caption{Lower bounds on the maximum DI randomness versus MABK value, normalized by the maximum quantum value $2^{(N-1)/2}$, for even $N$. These are compared to numerical upper bounds computed via the SOS technique. Each plot starts at the conjectured maximum MABK value achievable with maximum randomness, $m_{N}^{*}$, and ends with the maximum quantum MABK value. The points coincide closely with each curve, supporting our conjecture that the analytical lower bound provided in the text is tight.}
\label{fig:numEv}
\end{figure}

\begin{table}[h!]
\begin{tabular}{|c | c | c | c| c |} 
 \hline
 Facet & $\eta^{\mathrm{L}}$ & $s^{*}$ & $\eta^{\mathrm{Q}}$  & $r(\langle S_{3} \rangle = \eta^{\mathrm{Q}})$ \\ [0.5ex] 
 \hline\hline
 $M_{3}$ & 1 & 2 & 2 & 3 \\ 
 \hline
 $S_{1}$ & 1  & 1.64621108 & $5/3 \approx 1.66666667$ & $2 + H_{\mathrm{bin}}(2/3) \approx 2.91829583$\\
 \hline
 $S_{2}$ & 2 & 2.59807617 & $2\sqrt{2} \approx 2.82842712$ & $7/2 - \log_{2}(1 + \sqrt{2})/\sqrt{2} \approx 2.60087604 $\\
 \hline
 $S_{3}$ & 2 & $3\sqrt{3}/2 \approx 2.59807621$ & $2\sqrt{2} \approx 2.82842712$ & $5/2 - \log_{2}(1 + \sqrt{2})/\sqrt{2} \approx 1.60087604$ \\
 \hline
 $S_{4}$ & 1 & 0 & 1 & 0  \\
 \hline
\end{tabular}
\caption{Relevant bounds on violations of all inequivalent facet classes of the local polytope in the $N=3$ scenario, with two binary measurement per party, given in \cref{eq:facets}. $M_{3}$ is the MABK expression, $S_{1},S_{2}$ are the other non trivial facets for this scenario, $S_{3}$ is the trivial extension of CHSH and $S_{4}$ is a positivity facet. $\eta^{\mathrm{L}}$ and $\eta^{\mathrm{Q}}$ are the local and quantum bounds respectively, and $s^{*}$ is an upper bound on the maximum value achievable with maximum randomness, computed numerically for $S_{1}$ and $S_{2}$ while all other cases are analytic. $r(\langle S_{3} \rangle = \eta^{\mathrm{Q}})$ is an upper bound on the asymptotic global DI randomness when the maximum quantum violation is achieved. The cases $M_{3}$ and $S_{3}$ are known to be tight~\cite{DharaMaxRand,WBC}. Due to its structure, $S_{2}$ retains the same characteristics as CHSH~\cite{WBC}, whereas $S_{1}$ exhibits an unexplored trade-off. Since $S_{4}$ cannot be violated, it cannot be used for DI randomness certification.}
\label{tab:S1S2}
\end{table}

\begin{figure}[h]
\includegraphics[width=8.4cm]{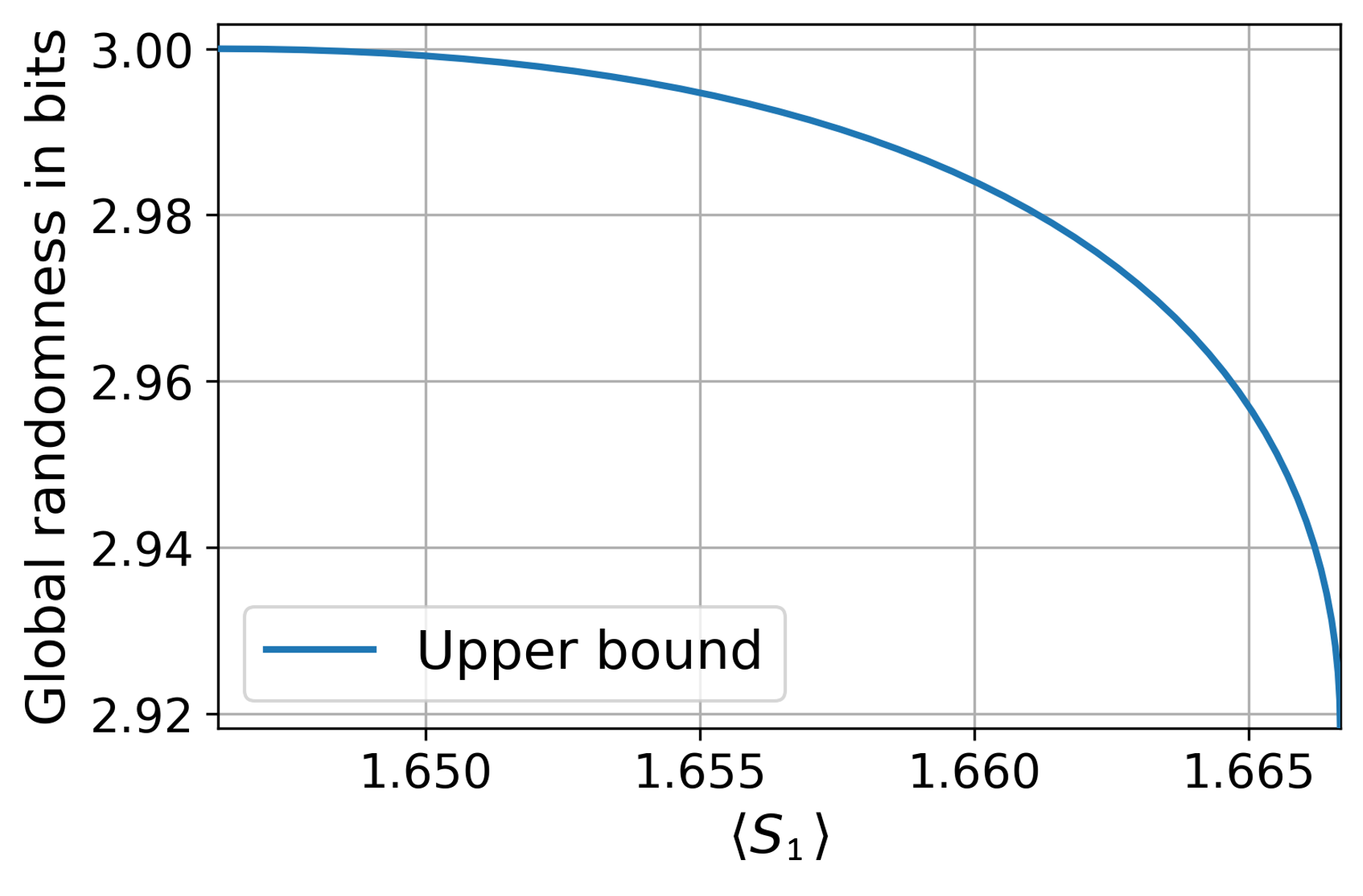}
\includegraphics[width=8.4cm]{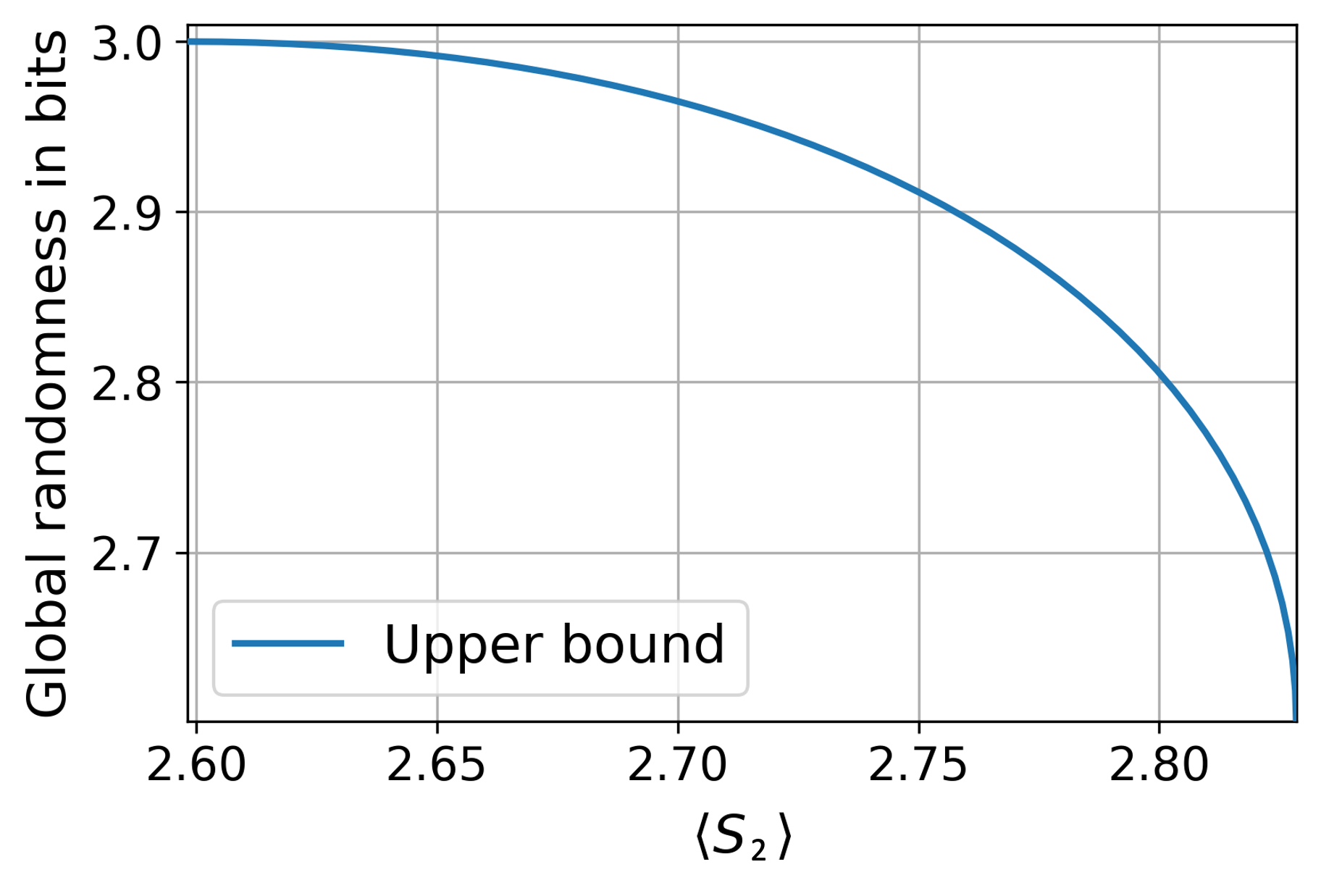}
\centering
\caption{Upper bounds on the trade-off between global DI randomness generation and violation of non-trivial facets in the tripartite scenario with two binary measurements per party. The bounds are numerical and generated by our SOS technique. $S_{1}$ and $S_{2}$ are two non-trivial extremal Bell expressions given in \cref{eq:facets}. Details of the exact numerical values are given in \cref{tab:S1S2}. }
\label{fig:S1S2}
\end{figure}

\twocolumngrid

\end{document}